\documentclass[12pt]{article}

\usepackage{rotfloat}
\RequirePackage[english]{babel}
\RequirePackage[shortlabels]{enumitem}
\RequirePackage[mathscr]{euscript}
\RequirePackage[margin=1.25in]{geometry}
\RequirePackage[utf8]{inputenc}
\RequirePackage{titlesec}
\RequirePackage[dvipsnames]{xcolor}
\RequirePackage{adjustbox,amsbsy,amsmath,amssymb,amsthm,anyfontsize,array,bm,caption,commath,chngcntr,dsfont,econometrics,etoolbox,fancyhdr,gensymb,graphicx, hhline, IEEEtrantools,longtable,marginnote,mathrsfs,mathtools,mdframed,musicography,natbib, parskip,pgf,setspace,subfig,tabularx,textcomp,tikz}

\let\savedegree\degree
\let\degree\relax
\RequirePackage{mathabx}
\let\degree\savedegree

\RequirePackage{authblk}

\RequirePackage[bookmarks=true,%
                colorlinks=true,%
                citecolor=ucladarkestblue,%
                linkcolor=darkpastelred,%
                urlcolor=ucladarkestblue,
                pdftoolbar=false,%
                pdfmenubar=true]{hyperref}

\RequirePackage{cleveref}

\renewcommand{\E}{\mathbb{E}}

\newcommand{\andbox}{\hbox{ }\text{ and }\hbox{ }}

\DeclareMathOperator{\Cov}{Cov}

\DeclareMathOperator{\Var}{Var}

\theoremstyle{definition}

\crefname{condition}{condition}{conditions}
\newtheorem*{condition*}{Condition}
\newtheorem{definition}{Definition}
\crefname{definition}{definition}{definitions}
\newtheorem*{definition*}{Definition}
\newtheorem{example}{Example}
\crefname{example}{examples}{examples}
\newtheorem*{example*}{Example}
\newtheorem{remark}{Remark}
\crefname{remark}{remark}{remarks}
\newtheorem*{remark*}{Remark}

\crefname{specification}{specification}{specifications}
\newtheorem*{specification*}{Specification}

\crefname{solution}{solution}{solutions}
\newtheorem*{solution*}{Solution}

\theoremstyle{plain}
\newtheorem{assumption}{Assumption}
\crefname{assumption}{assumption}{assumptions}
\newtheorem*{assumption*}{Assumption}
\newtheorem{corollary}{Corollary}
\crefname{corollary}{corollary}{corollaries}
\newtheorem*{corollary*}{Corollary}
\newtheorem{lem}{Lemma}
\crefname{lem}{lemma}{lemmas}
\newtheorem*{lem*}{Lemma}
\newtheorem{prop}{Proposition}
\crefname{prop}{proposition}{propositions}
\newtheorem*{prop*}{Proposition}
\newtheorem{thm}{Theorem}
\crefname{thm}{theorem}{theorems}
\newtheorem*{thm*}{Theorem}

\newtheoremstyle%
 {bluethm}%
 {}{}%
 {\color{BrickRed}}
 {}%
 {\color{ucladarkblue}\bfseries}%
 {\color{ucladarkblue}.}%
 { }{}
\theoremstyle{bluethm}

\crefname{todo}{todo}{todos}
\newtheorem*{todo*}{Todo}

\numberwithin{assumption}{section}
\numberwithin{corollary}{section}
\numberwithin{definition}{section}
\numberwithin{equation}{section}
\numberwithin{example}{section}
\numberwithin{figure}{section}
\numberwithin{lem}{section}
\numberwithin{prop}{section}
\numberwithin{remark}{section}
\numberwithin{solution}{section}
\numberwithin{specification}{section}
\numberwithin{thm}{section}

\counterwithin*{footnote}{section}

\apptocmd{\sloppy}{\hbadness 10000\relax}{}{}

\definecolor{uclablue}{HTML}{2774AE}
\definecolor{ucladarkblue}{RGB}{0, 85, 135}
\definecolor{ucladarkestblue}{RGB}{0, 59, 92}
\definecolor{uclagold}{RGB}{255, 184, 0}
\definecolor{uclamidnight}{RGB}{0, 39, 61}
\definecolor{pastelred}{HTML}{d62d0e}
\definecolor{darkpastelred}{RGB}{194, 59, 34}
\definecolor{pastelgreen}{RGB}{2,117,36}

\newcommand{\red}[1]{\textcolor{pastelred}{#1}}

\newcommand{\bflarrow}{\vphantom{\big[}\longleftarrow\,}
\newcommand{\flarrow}{\,\longrightarrow}

\title{Inference under First-Order Degeneracy\thanks{Emails: \href{mailto:xinyue.bei@austin.utexas.edu}{xinyue.bei@austin.utexas.edu} and \href{mailto:mnavjeevan@tamu.edu}{mnavjeevan@tamu.edu}. We thank Zheng Fang, Matt Masten, Andres Santos, and seminar participants at the UT Austin Econometrics Writing Seminar for their helpful comments. Mengyuan Jiang provided excellent research assistance.}}
\author[1]{Xinyue Bei}
\author[2]{Manu Navjeevan}
\affil[1]{The University of Texas at Austin}
\affil[2]{Texas A\&M University}

\date{\today}

\allowdisplaybreaks
\begin{document}

\maketitle 

\thispagestyle{empty}
\begin{abstract}
    \noindent
    We study inference in models where a transformation of parameters exhibits first-order degeneracy --- that is, its gradient is zero or close to zero, making the standard delta method invalid. A leading example is causal mediation analysis, where the indirect effect is a product of coefficients and the gradient degenerates near the origin. In these local regions of degeneracy the limiting behaviors of plug-in estimators depend on nuisance parameters that are not consistently estimable. We show that this failure is intrinsic --- around points of degeneracy, both regular and quantile-unbiased estimation are impossible. Despite these restrictions, we develop minimum-distance methods that deliver uniformly valid confidence intervals. We establish sufficient conditions under which standard chi-square critical values remain valid, and propose a simple bootstrap procedure when they are not. We demonstrate favorable power in simulations and in an empirical application linking teacher gender attitudes to student outcomes.
    \vspace{1em}
    
    \text{Keywords:}  Delta Method, Higher Order Asymptotics, Impossibility, Minimum Distance Inference
    
    \vspace{1em}
    
    \text{JEL Codes:} C12, C13, C18
\end{abstract}

\section{Introduction}
\label{sec:Intro}
The delta method is a fundamental tool in econometric analysis for deriving the asymptotic distribution of smooth functions of estimators. In its standard form, the delta method relies on a non-zero gradient of the smooth function with respect to the  primitive parameter. When this condition holds, first-order linearization provides an accurate approximation to the sampling distribution. However, in many empirically relevant scenarios,  the gradient may be zero at certain points in which case higher-order terms become essential.  As a result, the limiting distribution near points of degeneracy typically differs substantially from that obtained under standard regularity. A leading example arises in causal mediation analysis, where the indirect effect is the product of two primitive parameters: the effect of the treatment on the mediator and the
effect of the mediator on the outcome. When these effects are zero, the gradient of the indirect effect degenerates, leading to a nonstandard limiting distribution of the Wald statistic \citep{sobel1982asymptotic}. 

In practice,  the researcher does not know whether the true parameter lies near such points of degeneracy, and thus cannot know which asymptotic approximation is appropriate for inference. This challenge has motivated a number of recent papers on hypothesis testing when the gradient may be degenerate, see, for example, \cite{van_garderen_nearly_2024} for an analysis of the mediation model mentioned above and \citet{Dufour2025WaldTW} in a more general treatment of Wald-type statistics. These papers acknowledge discontinuities in  limiting distributions at the point of degeneracy, and analyze the resulting distortions in Wald-type statistics. Yet, they do not propose a unified asymptotic framework for studying local regions of degeneracy.
Moreover, most existing papers are interested in specific point hypotheses, leaving  open the broader question of how to construct uniformly valid confidence intervals. 

This paper makes two main contributions to the literature. Our first contribution is a formal asymptotic framework  for studying the behavior of statistics in local regions of first-order degeneracy. In our framework, the primitive parameter is modeled as local to the point of degeneracy, in the spirit of weak identification asymptotics, where identifying information is local to zero. Under this setup, the behavior of simple plug-in estimators becomes nonstandard and depends on local parameters that cannot be consistently estimated, formally capturing an observation in the existing literature that the standard delta method does not properly approximate the behavior of plug-in estimators near degenerate points \citep{miller2024testing}. 

Leveraging Le Cam's limit of experiments framework \citep{lecam1970assumptions,le1972limits}, we show that the problem of estimating a smooth function under degeneracy  is asymptotically equivalent to estimating a quadratic form of the shift parameter in a Gaussian shift model. Within this model, we show that equivariant-in-law and quantile-unbiased estimators cannot be obtained.  Translated back to the original estimation problem, these results imply that regular estimation is impossible --- the limiting behavior of \emph{any properly scaled and centered estimator} in local regions of degeneracy depends on nuisance parameters that cannot be consistently estimated. Moreover, there do not exist asymptotically similar confidence intervals for the transformation of interest in local regions of degeneracy.

Our second contribution is to construct confidence intervals that are both uniformly valid and exhibit favorable power in local regions of degeneracy compared to the few existing alternatives. 
The impossibility results mentioned above imply that standard delta-method based confidence intervals may not be uniformly valid in such regions. Indeed, in the context of mediation analysis our simulation study shows that standard Wald-statistic based tests lead to confidence intervals that undercover when the true primitive parameter is near the origin. 
We thus take a different approach and propose confidence intervals based on test inversion with a minimum-distance test statistic. We first show that when the parameter dimension is two, the standard chi-square critical value is uniformly valid under either of two conditions: (i) the curvature of the null curve is not too large, or (ii) the two branches of the null curve are sufficiently close. These conditions hold in the leading mediation example. In more general settings, we propose a bootstrap critical value based on a quadratic approximation of the test statistic, and show that when the true parameter is well separated from the point of degeneracy, the bootstrap  critical value is nearly identical to the efficient one.

We demonstrate the empirical relevance of our results in both simulation study and real world. In simulation study our proposed methods are shown to control size uniformly over the parameter space while standard Wald-based inference can overreject when the true primitive parameter is close to points of degeneracy, a finding also noted by \citet{Dufour2025WaldTW}. We additionally demonstrate favorable power properties of our proposed methods when compared to a method proposed by \citet{AM-2016-Geometric}, which turns out to also be applicable in this setting.\footnote{In their original analysis, \citet{AM-2016-Geometric} were interested in testing non-linear restrictions in the context of weak identification.} These improvements in power are also seen in an application to the data of \citet{AEM-2018}, who consider the effect of teachers' gender attitudes on student outcomes. In revising their mediation analysis, we find that our proposed methods deliver tighter confidence intervals than existing methods for all parameters. Our results support a conclusion by \citet{van_garderen_nearly_2024} that the mediation effect of a one-year exposure to a teacher with traditional negative views is negative, while existing methods cannot rule out an effect of size zero.

The rest of this paper proceeds as follows. This section concludes with a review of the related literature. \Cref{sec:Examples} gives examples of when first-order degeneracy may be a concern. \Cref{sec:Impossibility} formally establishes the impossibility results mentioned above and discusses implications for hypothesis testing. \Cref{sec:Inference} introduces the minimum distance based inference procedures and discusses how uniformly valid critical values may be constructed. \Cref{sec:Simulation,sec:Empirical} contain, respectively, the simulation study and empirical application to the data of \citet{AEM-2018}. \Cref{sec:Conclusion} concludes. Proofs are deferred to \Cref{sec:Imposs-Proofs,sec:Inference-Proofs}.

\subsection{Literature Review}

Our paper is related to previous literature on econometrics and statistics studying inference under degeneracy, statistical impossibility results, and testing non-linear restrictions.

There is a growing literature examining hypothesis tests in which the null includes points of singularity. 
\citet{gaffke1999asymptotic} show that the distribution of the Wald statistic at points of degeneracy is nonstandard, and \citet{gaffke2002asymptotic}  derive its asymptotic distribution under a variety of singular null hypotheses. \citet{drton2016wald} demonstrate the conservativeness of the Wald test at degeneracy points for quadratic forms and for bivariate monomials of arbitrary degree. \citet{dufour2016rank} propose rank-robust regularized Wald-type tests allowing
for singular covariance matrices. \citet{Dufour2025WaldTW} analyze Wald tests for polynomial restrictions with possibly multiple constraints and show that such tests can under-reject, over-reject, or even diverge under the null; see also \citet{Dufour2025WaldTW} for additional references in this area. 
Our paper differs from this literature in two key ways. First, prior work focuses on testing problems where the null itself contains the singularity, while our interest lies in constructing uniformly valid confidence
intervals when the null may be near a singularity. In simulations, we show that the upper bound on the Wald statistic derived in \citet{Dufour2025WaldTW} does not, in general, yield valid confidence intervals. 
Second, instead of using a Wald statistic, we employ a minimum distance–based test statistic, which is bounded in probability by construction. This approach avoids the divergence issues documented in \citet{Dufour2025WaldTW}.

Our paper is also related to the hypothesis testing problem with a curved null. \citet{AM-2016-Geometric} study this problem and show that the distribution of minimum-distance statistics is dominated by a tractable distribution that depends only on the maximal curvature of the null manifold relative to the known variance matrix. Inverting their test leads to uniformly valid confidence intervals. However, when the curvature of the null hypothesis is large, for example, in testing the significance of an indirect effect, their procedure yields critical values that are close to those from projection-based methods. By contrast, our procedure exploits the possibility that the null hypothesis may include multiple manifolds that are close to one another, which in turn reduces the critical value.

Finally, our paper contributes to the econometric literature on statistical impossibility results. In particular, it is related to work by \citet{hirano2012impossibility} who show that regular estimation of directionally, but not fully, differentiable functions is unattainable. Our paper takes a similar approach to that of \citet{hirano2012impossibility} in that we rule out properties of estimators by analyzing a limiting experiment \citep{lecam1970assumptions,le1972limits}. However, the target functional in our limit experiment is distinct from that of \citet{hirano2012impossibility}. This approach of ruling out behaviors by analyzing limit experiments has also been utilized successfully by \citet{kaji2021theory} and \citet{andrewsgmm} in the study of weak identification. The present work is also related to work by \citet{chen2019inference} who show that all standard bootstrap procedures necessarily fail at points of degeneracy, similarly to how \citet{fang2019inference} establish that the bootstrap necessarily fails as an inference procedure for the functionals considered in \citet{hirano2012impossibility}. 

\section{Overview and Examples}
\label{sec:Examples}
Consider a parameter \(\theta \in \Theta \subseteq \SR^d\) and a twice continuously differentiable function \(g: \Theta \to \SR\). We are interested in inference on \(g(\theta)\) in local neighborhoods of a point \(\theta_\star\) for which \(\nabla g(\theta_\star) = 0\).
Below, we give some empirically relevant examples of when such a phenomenon may occur.

\begin{example}[Mediation Analysis]
    \label{ex:mediation}
    Consider a causal mediation analysis with parameter \(\theta = (\theta_1,\theta_2)'\), where \(\theta_1\) represents the effect of a treatment variable on a mediator and \(\theta_2\) represents the effect of the mediator on the outcome. The indirect effect of the treatment on the outcome is then given by \(g(\theta) = \theta_1\theta_2\). At \(\theta_\star = (0,0)'\), we have \(\nabla g(\theta_\star) = 0\), which complicates inference on \(g(\theta)\) in local regions of \(\theta_\star\). As a result, recent works have proposed tests for the specific null-alternate pair, \(H_0: g(\theta) = 0\) against \(H_1: g(\theta) \neq 0\), see \citet{van_garderen_nearly_2024} or \citet{hillier_improved_2024}. However, these works do not consider the more general problem of constructing confidence intervals in local regions of the origin. \qed
\end{example}

\begin{example}[Impulse Response Function]
    \label{ex:impulse-response}
    Consider an  autoregressive \(\text{AR}(1)\) model of the form \(y_t = \theta y_{t-1} + u_t\) where \(y_t,y_{t-1} \in \SR\), \(\theta \in \SR\), and  \(u_t \in \SR\) is a white noise process. The ``impulse response function'' is defined as \(g(\theta) = \theta^h\) and measures the impact at time period \(h\) of an initial shock. Due to the importance of this in macroeconomic analysis, inference on \(g(\theta)\) has received attention from the econometric literature \citep{IK-2002, Gospodinov-2004, Mikusheva-2012},  mostly related to inference when \(\theta\) is close to one --- the so-called ``unit-root'' problem. However, due to degeneracy, inference  on the impulse response function can also be complicated when \(\theta\) is close to \(\theta_\star = 0\) as  \(\frac{\partial }{\partial \theta} g(\theta_\star) = h\theta_\star^{h-1} = 0\) \citep{benkwitz2000problems,lutkepohl2013introduction}.
    \qed
\end{example}

\begin{example}[Breakdown Point Analysis]
    \label{ex:sq-hellinger}
    Consider a missing data setup in which the researcher observes \(\{Y_iD_i, D_i, X_i\}_{i=1}^n\), where \(D_i \in \{0,1\}\) represents whether or not an observation's outcome $Y_i$ is observed, and \(X_i\) is a set of discrete covariates, i.e., \(X_i \in \calX \coloneqq \{x_1,\dots,x_{K}\}\). To achieve identification of parameters of interest, assumptions are typically made about the selection mechanism, such as ``missing conditionally at random'', i.e, \(Y_i \perp D_i \mid X_i\). \citet{ober2026robustness} proposes assessing the robustness of results to these assumptions through a breakdown point analysis. This assessment involves generating a confidence interval for the squared Hellinger distance between \(P_0\), the distribution of \(X_i \mid D_i = 0\), and \(P_1\), the distribution of \(X_i \mid D_i = 1\). Since \(X_i\) is discrete, this distance can be written as
     \[
        g(\theta) = H^2(P_0,P_1) = \frac{1}{2}\sum_{k=1}^K (\sqrt{\theta_{0,k}} - \sqrt{\theta_{1,k}})^2
    \]
    where \(\theta_{d,k} = \Pr(X = x_k \mid D = d)\). Simple sample analog estimators of  \(\theta_{0,k}\) and \(\theta_{1,k}\) are \(\sqrt{n}\)-consistent and asymptotically normal under mild assumptions. However, the derivative of \(g(\theta)\) with respect to the quantities \(\theta_{0,k}\) and \(\theta_{1,k}\) are given by
    \begin{align*}
        \frac{\partial}{\partial\, \theta_{0,k}}g(\theta)  
        = \frac{1}{2}\frac{\sqrt{\theta_{0,k}} - \sqrt{\theta_{1,k}}}{\sqrt{\theta_{0,k}}}, 
        \;\;\;\;\;\;\;\;\;\;
        \frac{\partial}{\partial \theta_{1,k}} g(\theta)
        = -\frac{1}{2}\frac{\sqrt{\theta_{0,k}} - \sqrt{\theta_{1,k}}}{\sqrt{\theta_{1,k}}}.
    \end{align*}
    Let \(\theta_\star\) be a point such that \(\theta_{0,k} = \theta_{1,k}\) for all \(k\). This occurs if the data is missing completely at random, that is \(D \perp (Y,X)\). At \(\theta_\star\), the derivatives above are uniformly equal to zero and the squared Hellinger distance \(g(\theta)\) between \(P_0\) and \(P_1\) is zero. Thus, when \(\theta\) is close to \(\theta_\star\), that is when \(P_0\) is close to \(P_1\), standard approaches to inference on  \(g(\theta)\) will fail.
    \qed
\end{example}

\begin{example}[Weak IV Bias and Size Distortion]
    \label{ex:weakiv} 
    Consider a standard homoskedastic linear IV model, 
    \begin{equation*}       
        \begin{split}
            y_i &= x_i\beta + \eps_i \\ 
            x_i &= z_i'\theta + v_i
        \end{split}
    \end{equation*}
    where \(y_i, x_i \in \SR\), \(\E[(\eps_i, v_i)'] = 0\), and \(Z = (z_i',\dots,z_n')' \in \SR^{n \times d_z}\) is treated as fixed. When \(\theta\) is close to zero, identification is referred to as ``weak'' and it is well known that standard inference procedures for \(\beta\) fail to control size  \citep{staiger1994instrumental}. \citet{StockYogo-2005} provide bounds on the size distortion of Wald tests for \(\beta\) in terms of the concentration parameter, 
    \[
        g(\theta) = \theta'(Z'Z)\theta/\sigma_v^2.
    \]
    \citet{GIR-2021} extend this analysis and develop confidence intervals for the bias and size distortion. In both papers, the researcher makes inferences about the concentration parameter by examining the distribution of the \(F\)-statistic, a scaled version of \(g(\hat\theta)\), where \(\hat\theta\) is the OLS estimate of \(\theta\).
    The analyses of both \citet{StockYogo-2005} and \citet{GIR-2021} are complicated by the fact that the limiting distribution of \(g(\hat\theta)\) is non-standard when \(\theta\) is close to zero, that is, when identification is weak. This can also be seen as inference in local regions of degeneracy --- at \(\theta_\star = 0\) we have that \(\nabla g(\theta_\star) = 2(Z'Z)\theta/\sigma_v^2 = 0\).
    \qed
\end{example}

\begin{example}[Explained Variance in Linear Regression]
    \label{ex:reg-var}
    Consider a linear regression model,
    \begin{equation*}
        Y = X'\theta + \eps, \;\;\; \E[\eps X] = 0
    \end{equation*}
    and define \(\sigma_Y^2 = \text{Var}(Y)\) and \(\Sigma_X = \E[XX']\). A parameter of interest is the proportion of variance in \(Y\) explained by the linear model with \(X\), i.e
    \begin{equation*}
        g(\theta) = \theta'\Sigma_X \theta/\sigma_Y^2,
    \end{equation*}
    that is, the population $R^2$. Although empirical work typically reports only a point estimate of $R^2$, reporting a confidence set for $R^2$ is informative for comparing the explanatory power or predictive performance of competing models \citep{hawinkel2024out}. When $R^2$ is bounded away from zero and one, standard errors and confidence intervals can be obtained using conventional asymptotic approximations \citep{cohen2013applied}. However, at \(\theta_\star = 0\) we have that \(\nabla g(\theta) = 2\Sigma_X\theta_\star/\sigma_Y^2 = 0\). Consequently, inference on the explained variance is non-standard when \(\theta\) is close to zero, or equivalently, when $R^2$ is close to zero.

    This type of parameter is also  of interest in labor economics when explaining variation in wage regressions. If a model under consideration can only explain a weak amount of variation in wage dispersion, we may expect \(\theta\) to be close to zero.
   \citet{CHK-2013} compare the baseline \citet*{AKM-1999} (AKM) model with various extensions in terms of each models ability to explain increases in wage inequality in West Germany. They find that these extensions provide little explanatory power on top of the baseline AKM model. The additional variance explained by these extensions corresponds to the linear regression model with \(Y\) equal to the residual from the AKM model and \(X\) equal to the new fixed-effect terms introduced by the extended models. They find that these new fixed-effect terms are close to zero suggesting that degeneracy may be a concern when conducting inference on \(g(\theta)\).\qed
\end{example}
%
%

\section{Impossibility Results}
\label{sec:Impossibility}
In this section we establish that standard approaches to inference on $g(\theta)$ necessarily fail in local regions  of first-order degeneracy. We begin in \Cref{sec:Impossibility}\red{.1} by introducing a parametric framework to study this problem and defining what it means for an estimator to be ``regular''  in this setting. \Cref{sec:Impossibility}\red{.2} then uses a representation theorem to show that the problem reduces to estimation of quadratic forms in a Gaussian shift experiment, where we prove that well-behaved estimators cannot be constructed. \Cref{sec:Impossibility}\red{.3} extends the analysis in two directions: first, to hypothesis testing problems where the null hypothesis that \(g(\theta) = g(\theta_\star)\) may hold on a nontrivial subset of the parameter space, and second, to infinite-dimensional models, where we show that the impossibility results remain valid so long as the model contains a suitable parametric submodel. Together, these results demonstrate that standard approaches to inference
necessarily break down in local regions of degeneracy.

\subsection{Preliminaries}

We begin by assuming that the researcher observes data \(X^{(n)} = (X_1,\dots,X_n)\) drawn from a parametric model \(P_{n,\theta}\),
\begin{equation}
    \label{eq:data-draw}
    \begin{split}
        X^{(n)} \sim P_{n,\theta}
    \end{split}
\end{equation}
where \(\theta \in \Theta \subset \SR^d\) and \(\Theta\) is a compact set with a nonempty interior, \(\Theta^\circ \neq \emptyset\). Let \(\calX_i\) denote the support of \(X_i\), which could be a general space, and denote \(\calX^n = \bigtimes_{i=1}^n \calX_i\). We assume that the sequence of statistical models \(\left(P_{n,\theta}: \theta \in \Theta^\circ\right)\), indexed by the sample size \(n\), is locally asymptotically normal in the sense of \citet{lecam1960lan}. 
\begin{assumption}[Local Asymptotic Normality]
    \label{assm:lan}
    There exists a sequence \(r_n \to \infty\) such that for every \(\theta \in \Theta^\circ\) and every sequence \(h_n \to h \in \SR^d\)
    \begin{equation}
        \label{eq:lan0}
        \begin{split}
            \log \left(\frac{dP_{n,\theta + h_n/r_n}}{dP_{n,\theta}}(X^{(n)})\right) = h'\Delta_n - \frac{1}{2}h'\Gamma_\theta h + Z_n(h)
        \end{split}
    \end{equation}
    where \(\Delta_n\) converges in distribution to \(N(0,\Gamma_\theta)\) under the sequence of measures \(P_{n,\theta}\), \(\Delta_n \overset{\theta}{\rightsquigarrow} N(0,\Gamma_\theta)\), and \(Z_n(h)\) converges in probability to zero under \(P_{n,\theta}\) for every \(h \in \SR^d\), \(Z_n(h) \xrightarrow p 0\).
\end{assumption}

\begin{example}[Smooth Parametric Models]
    \label{ex:smooth-parametric}
    A leading example of a model that satisfies \eqref{eq:lan0} is when the researcher observes i.i.d data, \(X_i \overset{iid}{\sim} P_\theta\) where \(\theta \in \Theta\). 
    Assume that there exists a dominating measure \(\mu\) such that \(P_\theta \ll \mu\) for all \(\theta \in \Theta^\circ\) and the Radon-Nikodym densities \(p_\theta = dP_\theta/d\mu\) are differentiable in quadratic mean, that is there is a function \(\dot\ell_\theta\) such that for any \(\theta \in \Theta^\circ\),
    \begin{equation}
    \label{eq:DQM}
            \int \left[\sqrt{p_{\theta + h}} - \sqrt{p_\theta} - \frac{1}{2}h'\dot\ell_\theta \sqrt{p_\theta}\right]^2d\mu = o(\|h\|^2),\;\;\; h \to 0
    \end{equation}
    and such that the Fisher information, \(\Gamma_\theta = P_\theta\dot\ell_\theta\dot\ell_\theta'\), is nonsingular. Let  \(P_{n,\theta} = \otimes_{i=1}^n P_\theta\), then, for any \(\theta \in \Theta^\circ\) and any \(h_n \to h \in \SR^d\), the sequence of log likelihood ratios satisfies (\cite{vanDerVaart1998}, Theorem 7.2):
    \begin{equation*}
        \begin{split}
            \log\left(\frac{dP_{n,\theta + h/\sqrt{n}}}{dP_\theta}(X^{(n)})\right) 
            &= \log \prod_{i=1}^n \frac{p_{\theta + h/\sqrt{n}}}{p_\theta}(X_i) \\
            &= \frac{1}{\sqrt{n}}\sum_{i=1}^n h'\dot\ell_\theta(X_i) - \frac{1}{2} h'\Gamma_\theta h + R_n(h),
        \end{split}
    \end{equation*}
   where \(R_n(h) = o_{P_{n,\theta}}(1)\) for all \(h \in \SR^d\). By the central limit theorem, \(\frac{1}{\sqrt{n}}\sum_{i=1}^n \dot\ell_\theta(X_i) \overset{P_{n,\theta}}{\rightsquigarrow} N(0, \Gamma_{\theta})\). Thus, by letting \(\Delta_n = \frac{1}{\sqrt{n}}\sum_{i=1}^n \dot\ell_\theta(X_i)\) we see that \Cref{assm:lan} is satisfied with \(r_n = \sqrt{n}\). 
    \qed
\end{example}

We study a twice continuously differentiable scalar functional \(g:\Theta \to \SR\), and the behavior of estimators of \(g(\theta)\) in local regions of a point \(\theta_\star \in \Theta^\circ\) which is such that the first-order derivatives of \(g(\cdot)\) at \(\theta_\star\) are zero. We will refer to \(\theta_\star\) as the ``point of degeneracy'' and local neighborhoods of \(\theta_\star\) as ``local regions of degeneracy.''

\begin{assumption}[Differentiability]
    \label{assm:differentiability}
    The function \(g:\Theta \to \SR\) is twice continuously differentiable on \(\Theta\), a compact subset of \(\SR^d\), with \(\nabla g(\theta_\star) = 0\) and \(\nabla^2g(\theta_\star) \neq 0\) for some \(\theta_\star \in \Theta^\circ\).
\end{assumption}

Given the maintained assumption of a locally asymptotically normal model, it is useful to examine regions close to \(\theta_\star\) by adopting a local parameterization around \(\theta_\star\), defining
\[
    \theta_{n, h} = \theta_\star + h/r_n
\]
and letting \(P_{n,h} = P_{n, \theta_\star + h/r_n}\). In our framework an estimator is an arbitrary measurable function of the data, \(\Psi_n: \calX^n \to \SR\). We consider sequences of estimators, \(\Psi_n\), that converge in distribution under every sequence of alternative distributions \(P_{n,h}\) to some limiting law, \(\calL_h\). This is denoted
\begin{equation}
    \label{eq:ncov}
    r_n^2\left(\Psi_n - g(\theta_{n,h})\right) \overset{h}{\rightsquigarrow} \calL_h.
\end{equation}
where we note that the convergence rate is \(r_n^2\) instead of \(r_n\) due to the fact that \(g(\theta)\) is ``flat'' around \(\theta_\star\). It is straightforward to show that tests for \(g(\theta)\) based on estimators whose convergence rates are slower than \(r_n^2\) when \(\theta\) is close to \(\theta_\star\) have trivial power against local alternatives of the form \(g(\theta_\star + h/r_n)\).

\begin{example}[Plug-In Estimators]
    \label{ex:plug-in}
    Suppose the researcher has access to an estimator \(\hat\theta\) of \(\theta\) that satisfies
    \[
        r_n(\hat\theta - \theta_{n,h}) \overset{h}{\rightsquigarrow} \calW
    \]
    for every \(h \in \SR^d\). In the smooth parametric models described in \Cref{ex:smooth-parametric}, such an estimator could be the maximum likelihood estimator or a Bayes estimator such as the posterior mean. Since \(g(\cdot)\) is assumed to be twice continuously differentiable, the limiting behavior of the plug-in estimator \(g(\hat\theta)\) can be found via the second order delta method 
    \[
        r_n^2(g(\hat\theta) - g(\theta_{n,h})) \overset{h}{\rightsquigarrow} h'\nabla^2g(\theta_\star)\calW + \frac{1}{2}\calW'\nabla^2g(\theta_\star)\calW
    \]
   The behavior of the plug in estimator, \(g(\hat\theta)\), depends on the local parameter \(h\). \qed
\end{example}

We focus on ruling out regular and locally asymptotically \(\alpha\)-quantile unbiased estimation of \(g(\theta_{n,h})\) in local regions of \(\theta_\star\). Formally, these notions are defined as follows.

\begin{definition}[Regularity]
    \label{def:regularity}
    Let $\Psi_n$ be an estimator satisfying \eqref{eq:ncov}, and let $\alpha \in (0,1)$.
    \begin{enumerate}[(i)]
        \item $\Psi_n$ is \emph{regular} if its limiting distribution does not depend on $h$, i.e.\ there exists a distribution $\calL$ on $\mathbb{R}$ such that $\calL_h = \calL$ for all $h \in \mathbb{R}^d$.
        \item $\Psi_n$ is \emph{locally asymptotically $\alpha$-quantile unbiased} if its limiting $\alpha$-quantile is zero for every $h$, i.e.\ $\calL_h\{(-\infty,0]\} = \alpha$ for all $h \in \mathbb{R}^d$.
    \end{enumerate}
\end{definition}
The existence of regular estimators is closely tied to the validity of Wald-type inference procedures --- without regular estimators, standard Wald-type inference procedures that compare test statistics to fixed critical values will not have correct asymptotic size. Similarly, the existence of locally asymptotically \(\alpha\)-quantile unbiased estimators is closely tied to the existence of asymptotically similar confidence intervals for \(g(\theta)\) in local regions of \(\theta_\star\). Since any asymptotically similar confidence interval of the form \((-\infty, \hat c]\) can be converted into a locally asymptotically \(\alpha\)-quantile unbiased estimator by taking \(\Psi_n = \hat c\), by ruling out such estimators we also rule out the possibility of similar one-sided confidence intervals.\footnote{The focus on one-sided confidence intervals is largely for simplicity of exposition. If \((-\infty, \hat c_1]\) and \([\hat c_2, \infty)\) are two asymptotically similar confidence intervals for \(g(\theta)\) each with coverage rate \(1 - \alpha/2\) and \(\hat c_1 \leq \hat c_2\) with probability approaching one, then \([\hat c_2, \hat c_1]\) is asymptotically similar with coverage rate \(1 - \alpha\).}

\subsection{Analysis in the Limiting Experiment}

To examine the possible behavior of such estimators, we make use of a representation result, given below in \Cref{thm:lim-exp}, which is a slight adaptation of Theorem 8.3 in \citet{vanDerVaart1998}. This earlier result is, in turn, a version of Le Cam's limit of experiments analysis for locally asymptotically normal models \citep{lecam1970assumptions,le1972limits}.

\begin{prop}[Limit Experiment]
    \label{thm:lim-exp}
    Suppose \Cref{assm:lan} holds, and let $\Psi_n$ be a sequence of estimators satisfying \eqref{eq:ncov}. Then there exists a randomized statistic $\Psi(Z,U)$, where $Z$ is drawn from the Gaussian shift experiment 
    \[
        Z \sim N(h, \Gamma^{-1}_{\theta_\star}), \quad h \in \mathbb{R}^d, 
    \]
    and $U \sim \mathrm{Unif}(0,1)$ independent of $Z$, such that
    \[
        \Psi(Z,U) - \tfrac{1}{2} h^\top \nabla^2 g(\theta_\star) h \;\sim\; \calL_h \quad \text{for all } h \in \mathbb{R}^d.
    \]
\end{prop}


\Cref{thm:lim-exp} establishes an equivalence between estimating \(g(\theta)\) in local regions of first-order degeneracy and estimating of a quadratic form of the mean parameter in a Gaussian shift model in which one observes a single draw \(Z \sim N(h, \Gamma_{\theta_\star}^{-1})\), where \(\Gamma_{\theta_\star}^{-1}\) is known but \(h\) is not. In particular, in a spirit similar to the approach in \citet{hirano2012impossibility}, we can rule out sufficiently regular behavior of estimators of \(g(\theta)\) in local regions of \(\theta_\star\) if the corresponding behavior is not permissible in the Gaussian shift model. Intuitively, sufficiently regular estimation of quadratic forms in the Gaussian shift experiment is not possible since the parameter of interest changes non-linearly as the mean parameter \(h\) varies over \(\SR^d\) while the distribution of \(Z\) changes in a linear fashion.

To illustrate, suppose that there was an estimator, \(\Psi(Z,U)\), and law, \(\calL\) with \(\int x^2\,d\calL(x) < \infty\), such that \(\Psi(Z,U) - \frac{1}{2}h'\nabla^2 g(\theta_\star)h\) is distributed according to \(\calL\) for all \(h \in \SR^d\). Since any such estimator can be turned into an unbiased estimator by subtracting off the mean of \(\calL\), it is without loss of generality to assume that \(\calL\) is mean zero and thus that \(\Psi(Z,U)\) is unbiased. On the other hand, the Cramér-Rao lower bound for the variance of any unbiased estimator of \(\frac{1}{2}h'\nabla^2g(\theta_\star)h\) in the Gaussian shift model yields
\begin{align}
    \label{eq:var-contradiction}
    \Var(\Psi(Z,U)) \geq h'(\nabla^2 g(\theta_\star))'\Gamma_{\theta_\star}^{-1}(\nabla^2 g(\theta_\star))h.
\end{align}
By letting \(h\) vary over \(\SR^d\), the right hand side of \eqref{eq:var-contradiction} can be made arbitrarily large while the left hand side is bounded by the second moment of \(\calL\). Thus, no such estimator can exist. Our full argument relies on analyzing characteristic functions, but the intuition is similar.  

\begin{remark}[]
    \label{rem:hirano-comparasion}
    It is instructive to compare the argument sketched above to the argument of \citet{hirano2012impossibility}, who rule out regular estimation of \(g(\theta)\) when \(g\) is directionally, but not fully differentiable at a point \(\theta_\star\). The \citet{hirano2012impossibility} argument relies on analyzing the behavior of a potential regular estimator as the local parameter \(h\) approaches zero. Our arguments, on the other hand, rule out regular estimation by analyzing the ``global'' behavior of a potential regular estimator, that is, the behavior as \(h\) varies over \(\SR^d\). The approach taken by \citet{hirano2012impossibility} does not apply in the present setting as the parameter of interest in the limit experiment is non-linear but continuously differentiable at zero.
    In contrast, in the limit experiment of \citet{hirano2012impossibility} the parameter of interest is a function \(\kappa(h)\) which is exactly linear around values of \(h \neq 0\), but is not continuously differentiable at zero. The argument of \citet{hirano2012impossibility} is able to additionally rule out locally unbiased estimation whereas in our setting locally unbiased estimation is possible. \qed
\end{remark}

\begin{prop}[]
    \label{prop:lim-exp2}
    Let \(Z \sim N(h, \Gamma_{\theta_\star}^{-1})\) and \(U \sim \mathrm{Unif}(0,1)\) independently of \(Z\). Let \(J\) be a \(d \times d\) non-zero, symmetric matrix.  
    \begin{enumerate}
        \item There is no randomized statistic \(\Psi(Z,U)\) and law \(\calL\) on \(\SR\) with \(\Psi(Z,U) - h'Jh \overset{h}{\sim} \calL\) for all \(h \in \SR^d\).
        \item Let \(\{\calL_h\}_{h \in \SR^d}\) be a system of probability measures on \(\SR\) such that (i) \(\calL_h\{(-\infty,0]\} =\alpha\) for some \(\alpha \in (0,1)\) and (ii) the CDFs associated with \(\calL_h\), \(F_h(\cdot)\), are differentiable at zero with derivative bounded below by some \(\eps > 0\). Then, there does not exist a randomized statistic \(\Psi(Z,U)\) such that \(\Psi(Z,U) - h'Jh \sim \calL_h\) for all \(h \in \SR^d\).
\end{enumerate}
\end{prop}

Together, \Cref{thm:lim-exp,prop:lim-exp2} can be combined for the main result of this section, which rules out sufficiently regular estimation in local areas of first-order degeneracy.\footnote{In the application of \Cref{prop:lim-exp2} to our setting, take \(J = \frac{1}{2}\nabla^2 g(\theta_\star)\). The result in \Cref{prop:lim-exp2} rules out well-behaved estimation of any quadratic form of the shift parameter, not just those associated with the Hessian of \(g(\cdot)\).} 

\begin{thm}[Impossibility of Regular Estimation]
    \label{thm:imposs-main}
    Suppose \Cref{assm:lan,assm:differentiability} hold. 
    \begin{enumerate}
        \item
        There is no estimator sequence $\Psi_n$ and law \(\calL\) on \(\SR\) such that
        \[
            r_n^2\big(\Psi_n - g(\theta_{n,h})\big) \;\overset{h}{\rightsquigarrow}\; \calL
            \quad \text{for all } h \in \mathbb{R}^d.
        \]
        \item
        Let $\{\calL_h\}_{h \in \mathbb{R}^d}$ be a family of distributions such that (i) $\calL_h\{(-\infty,0]\} = \alpha$ for some fixed $\alpha \in (0,1)$ and all $h$, and (ii) the CDFs, $F_h(\cdot)$, of $\mathcal{L}_h$ are differentiable at zero with derivatives bounded below by $\epsilon>0$.  
        Then there is no estimator sequence $\Psi_n$ such that
        \[
            r_n^2\big(\Psi_n - g(\theta_{n,h})\big) \;\overset{h}{\rightsquigarrow}\; \calL_h
        \quad \text{for all } h \in \mathbb{R}^d.
        \]
    \end{enumerate}
\end{thm}

\Cref{thm:imposs-main} rules out sufficiently well-behaved estimation of \(g(\theta)\) when the true parameter is ``close'' to \(\theta_\star\). In particular, \Cref{thm:imposs-main}(a) rules out the possibility of regular estimation --- the properly scaled and centered behavior of any estimator \(\Psi_n\) of \(g(\theta)\) must depend, in local regions of \(\theta_\star\), on the local parameter \(h\), which cannot be consistently estimated. Similarly, \Cref{thm:imposs-main}(b) rules out the possibility of \(\alpha\)-quantile unbiased estimation. As mentioned below \Cref{def:regularity}, this result has profound implications for inference on \(g(\theta)\) in local regions of \(\theta_\star\). In particular, both asymptotically exact Wald-type inference procedures and asymptotically similar confidence intervals for \(g(\theta)\) are unavailable in local regions of degeneracy \(\theta_\star\).

\Cref{thm:imposs-main}(a) also has implications for the construction of efficient estimators of \(g(\theta)\) in local regions of \(\theta_\star\). Standard notions of efficiency are tied to comparing the asymptotic risk of regular estimators. \Cref{thm:imposs-main}(a) shows that, after properly scaling, such regular estimators are not available in local regions of degeneracy. Consequently, alternative notions of efficiency must be considered in these settings and standard estimators may not be optimal in these regions. As an example, one can show that estimators of \(g(\theta)\) that are efficient under a standard asymptotic regime can be dominated by alternative estimators in local asymptotic mean squared error around points of degeneracy. 

\begin{remark}
As with the impossibility of locally asymptotically \(\alpha\)-quantile unbiased estimation for directionally but not fully differentiable parameters established in \citet{hirano2012impossibility}, the result in \Cref{thm:imposs-main} requires some regularity conditions on the system of limiting laws \(\{\calL_h\}_{h \in \SR^d}\).\footnote{Let \(F_0(\cdot)\) be the CDF associated with \(\calL_0\). \citet{hirano2012impossibility} show that, for any \(\alpha\)-quantile unbiased estimate, it must be the case that either \(F_0(\cdot)\) is not differentiable at zero or must satisfy \(F_h'(0) = 0\).} The regularity condition in \Cref{thm:imposs-main}(b) implies that, if a locally asymptotically \(\alpha\)-quantile unbiased estimator were to exist, its associated limiting laws \(\calL_h\) must be able to be made \emph{arbitrarily} flat.  In particular, if each limiting law \(\calL_h\) has a density with respect to Lebesgue measure, these densities evaluated at zero, which is by definition the \(\alpha\)-quantile of each \(\calL_h\), must be able to be made arbitrarily small. As an example, suppose that \(\{\calL_h\}_{h \in \SR^d}\) is a family of Gaussian distributions on \(\SR\) associated with a locally asymptotically \(\alpha\)-quantile unbiased estimator. Then, the variance of these Gaussian distributions must be able to become arbitrarily large as \(h\) ranges over \(\SR^d\). 
\qed
\end{remark}

\subsection{Hypothesis Testing and Infinite Dimensional Models}
The above analysis rules out standard approaches to inference on \(g(\theta)\) in local regions of \(\theta_\star\). These results are informative when one is interested in constructing confidence intervals for \(g(\theta)\) around points of first-order degeneracy when the data is drawn from a parametric model satisfying \Cref{assm:lan}. In this subsection, we consider two extensions of our results. In the first, we consider the somewhat simpler problem of testing the null hypothesis \(H_0: g(\theta) = g(\theta_\star)\). We show that if the null hypothesis contains a sufficiently rich set of values and a similar test exists, this similar test must have low power in local regions of \(\theta_\star\). In the second extension we generalize \Cref{thm:imposs-main} to infinite dimensional, i.e, semiparametric or nonparametric, models.

\subsubsection{Hypothesis Testing}

In this subsection we consider the problem of testing the null hypothesis \(H_0: g(\theta) = g(\theta_\star)\) where the alternative can be one sided, i.e, \(H_1: g(\theta) > g(\theta_\star)\) or two-sided, \(H_1: g(\theta) \neq g(\theta_\star)\). To setup, define \(\calH_\star\) to be the set of local parameters, \(h \in \SR^d\), such that \(g(\theta_{n,h})\) is asymptotically indistinguishable from \(g(\theta_\star)\), i.e, \(r_n^2(g(\theta_{n,h}) - g(\theta_\star)) \to 0\);
\[
    \calH_\star = \{h \in \SR^d: h'\nabla^2g(\theta_\star)h = 0\}.
\]
We study the behavior of asymptotically similar tests in local regions of degeneracy. In this setup, an asymptotically similar test is a statistic \(\Xi_n: \calX^n \to \{0,1\}\) such that 
\[
    \limsup_{n\to\infty} P_{\theta_{n,h}}(\Xi_n = 1) = \alpha \quad \text{for all } h \in \calH_\star.
\]
Equivalently, the test is similar if \(\limsup_{n\to\infty }P_{\theta_{n,h}}(1 - \Xi_n \leq 0) = \alpha\). Letting \(\Psi_n = 1 - \Xi_n\) it is apparent that this is a nearly identical requirement to that of local \(\alpha\)-quantile unbiasedness in \Cref{def:regularity}, with the key difference being that the requirement only needs to hold for local parameters \(h \in \calH_\star\) rather than for all \(h \in \SR^d\). 

However, unlike quantile unbiased estimation, which is ruled out in \Cref{thm:imposs-main}, asymptotically similar tests can exist --- one can imagine constructing a similar test by flipping a weighted coin.  Such a test, though, may not be powerful against local alternatives close to \(\theta_\star\). Our main result in this subsection establishes this formally: if such test exists then its local asymptotic power curve must be flat at \(\theta_\star\) in the sense that the derivative of the local asymptotic power curve with respect to the local parameter \(h\) exists and is equal to zero. 

Define the local asymptotic power curve as
\begin{equation}
    \label{eq:local-power-curve}
    \calP(h) = \limsup_{n\to\infty} P_{\theta_{n,h}}(\Xi_n = 1)
\end{equation}
\begin{prop}[]
    \label{thm:impossibility-flat}
    Let \(\Xi_n\) be an asymptotically similar test such that \(\calP(h)\) is differentiable at \(h = 0\). Then, the directional derivative of the local asymptotic power curve, \(\calP(h)\), in directions \(h \in \calH_\star\) is equal to zero:
    \[
        D_h\calP(0) = 0\;\;\text{ for all }\; h \in \calH_\star.
    \]
    In particular, if \(\nabla^2 g(\theta_\star)\) is indefinite then \(\calH_\star\) spans \(\SR^d\) and \(\nabla\calP(0)\) is equal to zero.
\end{prop}

\begin{remark}[Differentiability of \(\calP\)]
    \label{rem:diffpower}
    A common strategy in hypothesis testing is to compare a test statistic \(\Psi_n^\circ\) to a possibly data-dependent critical value \(\hat c_n\), rejecting when the former exceeds the latter. Let \(\Psi_n=\Psi_n^\circ-\hat c_n\), so the rejection rule can be written as \(\Xi_n=\bm{1}\{\Psi_n\geq 0\}\). If \(\Psi_n \overset{h}{\rightsquigarrow} \calL_h\) for each \(h \in \SR^d\), and if the CDF of \(\calL_0\) is continuous at zero, then the resulting local asymptotic power curve \(\calP\) is differentiable at \(0\); see \Cref{lemma:diff-local-power}. 
    This is a milder version of the regularity condition imposed on quantile unbiased estimators in \citet{hirano2012impossibility}.
    \qed
\end{remark}

The first statement in \Cref{thm:impossibility-flat} follows immediately from the definition of similarity along with the fact that \(\calH_\star\) is a cone: because \(\mathcal{P}(h)\) is constant on \(\calH_\star\), its directional derivatives in directions \(h \in \calH_\star\) must be zero. The force of the result, however, lies in the structure of \(\mathcal{H}^\star\) near points of degeneracy. In standard inference problems, i.e, when the true parameter is well separated from points of degeneracy, the parameter of interest in the limit experiment is a linear function of the shift parameter \(h\), so \(\mathcal{H}^\star = \{0\}\) and the zero-derivative condition carries no information about the shape of the power curve. By contrast, when \(g\) exhibits first-order degeneracy \(\mathcal{H}_\star\), can be a non-trivial cone — that is, it may contain directions other than zero — and the constraint \(D_h \mathcal{P}(0) = 0\) for \(h \in \calH_\star\) becomes substantive. When \(\nabla^2 g(\theta_\star)\) is indefinite, \(\calH_\star\) spans \(\SR^d\) and the entire gradient of the local asymptotic power curve vanishes at the origin, implying that power cannot increase at a linear rate in any direction away from \(\theta_\star\). The following examples illustrate the strength of this restriction.

\textbf{Example \ref{ex:mediation}, cont.}
    Consider again the mediation model, where the original model is given \((P_\theta: \theta \in \Theta \subseteq \SR^2)\). Suppose the researcher is interested in testing the null hypothesis \(H_0:\theta_1\theta_2 = 0\), that is \(g(\theta) = \theta_1\theta_2\) and \(\theta_\star = 0\). We can show that \(\calH_\star = \{h \in \SR^2:  h_1h_2 = 0\}\). Since \(\calH_\star\) is the union of the two coordinate axes, \(\mathrm{span}(\calH_\star) = \SR^2\). Thus, we have that \(\nabla\calP(0) = 0\) for any asymptotically similar test. 

    We can equivalently show that \(\calH_\star\) must span \(\SR^2\) by noting that the second derivative matrix of \(g(\cdot)\) at \(\theta_\star\) is given by
    \[
        \begin{pmatrix} 0 & 1\\1 & 0 \end{pmatrix},
    \]
    which is indefinite.
    \qed

\begin{example}[Squared Mean]
    \label{ex:squared-mean}
    On the other hand consider the case where \(\theta = (\theta_1,\theta_2) \in \SR^2\) and the researcher is interested in testing the null hypothesis \(H_0: \theta_1^2 + \theta_2^2 = 0\). In this case \(g(\theta) = \theta_1^2 + \theta_2^2\), \(\theta_\star = 0\), and \(\calH_\star = \{(0,0)\}\).  Since \(\nabla^2 g(\theta_\star)\) is positive definite, \(\mathrm{span}(\calH_\star) = \{0\}\). Thus, the results of \Cref{thm:impossibility-flat} do not apply and powerful similar tests for the null hypothesis \(H_0:\theta_1^2 + \theta_2^2 = 0\) can be constructed, see e.g \citet{chen2019improved}.  
    \qed 
\end{example}

\begin{example}[Standard Inference]
Suppose that the primitive parameter is univariate \(\theta_{n,h} = \theta_0 + h/r_n \in \SR\), and local to a point \(\theta_0 \)  such that \(g'(\theta_0) > 0\), that is we are well separated from points  of degeneracy. In this setting, the researcher typically has access to an estimator \(\Psi_n\) that satisfies \(r_n(\Psi_n -g(\theta_{n,h}))\overset{h}{\rightsquigarrow}N(0,\sigma^2)\). This estimator is regular and thus \(\alpha\)-quantile-unbiased  for all \(\alpha \in (0,1)\). Based on this estimator, an asymptotically similar one sided test for the null hypothesis, $H_0: g(\theta) = g(\theta_0)$, can be constructed with local asymptotic power curve \(\calP(h) = 1 - \Phi(c_{1-\alpha} - g'(\theta_0)h/\sigma)\), where \(\Phi(\cdot)\) is the standard normal CDF and \(c_{1-\alpha}\) is its \(1-\alpha\) quantile. Here, \(\frac{\partial }{\partial h}\calP(h)\big|_{h=0} = g'(\theta_0)\phi(c_{1-\alpha})/\sigma > 0\).
\end{example}

\begin{remark*}[]
    Recent papers by \citet{van_garderen_nearly_2024} and \citet{Dufour2025WaldTW} also study tests of the null hypothesis \(H_0:g(\theta) = g(\theta_\star)\) in various contexts. \citet{van_garderen_nearly_2024} consider the case of the mediation model, that is where \(\theta = (\theta_1, \theta_2)'\) and \(g(\theta) = \theta_1\theta_2\). They assume that the researcher has access to an asymptotically normal estimate of \(\theta\), \(\hat\theta = (\hat\theta_1, \hat\theta_2)'\) and show that there is no reasonable similar test of the form: reject if \(\max\{|\hat\theta_1|,|\hat\theta_2|\} > g(\min\{|\hat\theta_1|,|\hat\theta_2|\})\), where \(g(\cdot)\) may be an arbitrary function. Similarly, \citet{Dufour2025WaldTW} consider the behavior of Wald type tests based on the test statistic \(W_n = \frac{g(\hat\theta) - g(\theta_\star)}{\nabla g(\hat\theta)'\Sigma\nabla g(\hat\theta)}\), where \(\Sigma\) represents the asymptotic variance of  \(\hat\theta\). The authors show that, when \(\theta\) is close to \(\theta_\star\), the behavior of the Wald statistic can be irregular and propose alternate critical values for testing the null hypothesis \(H_0: g(\theta) = g(\theta_\star)\) using \(W_n\).

    We view our results as complementary to these existing results. The results in Proposition~\ref{thm:impossibility-flat} are narrower in their conclusion --- we establish only that similar tests must have flat power at $\theta^\star$ --- but broader in their scope, since they apply to any testing procedure that depends on the data, rather than only to procedures based on a specific initial estimator $\hat{\theta}$.
    \qed
\end{remark*}

\subsubsection{Infinite Dimensional Models}

In many settings the researcher may not be willing to assume that the data comes from a finite dimensional parametric model as described in the previous section. Following a tradition on studying semiparametric efficiency \citep{bickel1993efficient}, we show that this does not affect our impossibility results so long as the larger model contains a parametric submodel satisfying \Cref{assm:lan,assm:differentiability}. 

Formally, let the model \(\calP\) be a collection of sequences of probability measures on the sample space \(\calX^n\) from the previous section. That is, each element of \(\calP\) is a sequence of probability measures \(\{P_n\}\),  where each probability measure \(P_n\) is defined on the sample space \(\calX^n\). A finite dimensional submodel, \(\calP_f\), is some smaller collection of sequences of probability measures that can be parameterized as \(\calP_{f} = (\{P_{n,\theta}\}_{n\in\SN}: \theta \in \Theta)\) for an open set \(\Theta \in \SR^{d_f}\). Fix a ``centering'' sequence of probability measures \(\{P_{0,n}\}_{n \in \SN} \in \calP\). We say that the submodel passes through \(\{P_{0,n}\}\) if \(\{P_{0,n}\} \in \calP_f\), that is \(\{P_{0,n}\} = \{P_{n,\theta}\}\) for some \(\theta \in \Theta\). We will call such a parametric model ``regular'' if \Cref{assm:lan} holds and the model passes through \(\{P_{0,n}\}\).


We suppose that the object of interest is a quantity that depends on the sequence of underlying probability measures, that is we can think of the estimand \(g[\{P_n\}]\) as a functional defined on \(\calP\). For any regular parametric model, \(\calP_f\), this implicitly defines a function on \(\theta\) via the relation \(g_f(\theta) = g[\{P_{n,\theta}\}]\). With this notation defined, we show that the results of \Cref{thm:imposs-main} can be extended in a straightforward fashion to infinite dimensional models.

\begin{remark}[Semiparametric Models with i.i.d Data]

In the literature on semiparametric estimation with i.i.d data, where the researcher observes repeated observations drawn independently from a probability distribution \(P\) on \(\calX\) belonging to a model \(\mathcal{P}\) \citep{bickel1993efficient}, one can associate the entire sequence of probability measures \(\{\otimes_{i=1}^n P: n \in \SN\}\) with the underlying common distribution \(P\). With this association, one can consider the model \(\calP\) described above as a collection of probability measures on \(\calX\) rather than a collection of sequences of probability measures \(\{P_n\}\) where each \(P_n\) is defined on \(\calX^n\). The parameter, in turn, can be defined as a function of the underlying distribution \(P\) rather than as a function of the entire sequence by \(g[P] \equiv g[\{\otimes_{i=1}^n P: n \in \SN\}]\). However, when  dealing with time-series or network data, it may not be possible to define the parameter as a function of some representative underlying distribution and is instead a property of the sequence \(\{P_n\}\).
\qed
\end{remark}

\begin{corollary}[Impossibility in Infinite-Dimensional Models]
    \label{corr:semiparametric-imposs}
    Suppose the data are generated from a sequence of distributions $\{P_{0,n}\} \in \mathcal{P}$. Let $\mathcal{P}_f \subset \mathcal{P}$ be a regular parametric submodel passing through $\{P_{0,n}\}$, and suppose that \(g_f\) satisfies \Cref{assm:differentiability} with \(\{P_{n,\theta_\star}\} = \{P_{0,n}\}\).
    \begin{enumerate}
    \item There is no estimator sequence $\Psi_n$ and law \(\calL\) on \(\SR\) such that, along every regular parametric submodel $\mathcal{P}_f$,
    \[
        r_n^2\big( \Psi_n - g_f(\theta_\star + h/r_n) \big) \;\overset{h}{\rightsquigarrow}\; \calL \quad \text{for all } h \in \mathbb{R}^{d_f}.
    \]
    \item  Let $\{\calL_h\}_{h \in \mathbb{R}^{d_f}}$ be a family of distributions such that  
    (i) $\calL_h\{(-\infty,0]\} = \alpha$ for some $\alpha \in (0,1)$ and all $h$, and  
    (ii) each $\calL_h$ has a CDF, $F_h$, differentiable at zero with derivative bounded below by $\epsilon_f > 0$. Then there is no estimator sequence $\Psi_n$ such that, along every regular parametric submodel $\mathcal{P}_f$,
    \[
        r_n^2\big( \Psi_n - g_f(\theta_\star + h/r_n) \big) \;\overset{h}{\rightsquigarrow}\; \calL_h \quad \text{for all } h \in \mathbb{R}^{d_f}.
    \]
\end{enumerate}
\end{corollary}

\section{Minimum Distance Based Inference}
\label{sec:Inference}
In this section, we construct a uniformly valid confidence interval
for $g(\theta):\Theta\rightarrow\mathbb{R}$ using estimator $\hat{\theta}$. The confidence interval
is obtained by inverting the hypothesis $H_{0}:g(\theta)=\tau$,
and we use a minimum distance (MD) test statistic 
\[
\hat{T}_{n}(\tau)=\inf_{\theta\in\Theta:g(\theta)=\tau}r_{n}^{2}(\hat{\theta}-\theta)^{\prime}\Sigma^{-1}(\hat{\theta}-\theta).
\]
We focus on the settings where the standard first order approximation
of $g(\theta)$ fails at $\theta_{\star}$, but the second order derivative
is nondegenerate; that is, $\frac{\partial^{2}g}{\partial\theta\partial\theta^{\prime}}(\theta_{\star})=H$
with  $H$ full rank. 

In Section \ref{subsec:Bivariate-with-indefinite}, we discuss a simple
case where $\Theta \subseteq \mathbb{R}^{2}$ and $H$ is indefinite, and
we provide sufficient conditions under which the standard critical
value $Q(\chi_{1}^{2},1-\alpha)$ is uniformly valid. In Section \ref{subsec:General-Case},
we propose a computationally simple method for a general $g$, which can be
generalized to cases with higher order singularity.

\subsection{Two-Dimensional $\theta$ and Indefinite $H$}\label{subsec:Bivariate-with-indefinite}

To simplify notation, consider the null hypothesis
\begin{equation}
H_{0}:g(\theta_{1},\theta_{2}):=(1+\rho)\theta_{2}^{2}-(1-\rho)\theta_{1}^{2}=\tau,\label{eq:H0}
\end{equation}
with $|\rho|<1$ and $\tau\geq0$. The restriction $|\rho|<1$ guarantees
that $H$ is indefinite, while $\tau\geq0$ is a normalization. For
simplicity, let $n=r_{n}=1$, and assume $\hat{\theta}-\theta\sim N(0,I_{2}).$
The quadratic form $g$ and the normality of $\hat{\theta}$ can be
viewed as second order approximations, with general asymptotic results
provided in Theorem \ref{thm:d=00003D2}.

Let $X_{2}(\theta_{1})\in \mathbb{R}_+$
be the positive solution for $\theta_{2}$ such that (\ref{eq:H0}) holds. Let
$\mathcal{S}_{0}(\tau)$ be the null parameter space, which contains
two separate curves, $\mathcal{S}_{0}^{+}(\tau)$ and $\mathcal{S}_{0}^{-}(\tau)$,
\begin{align*}
\mathcal{S}_{0}(\tau) & =\mathcal{S}_{0}^{+}(\tau)\cup\mathcal{S}_{0}^{-}(\tau)
\end{align*}
where 
\[
\mathcal{S}_{0}^{+}(\tau)=\left\{ \left(x_{1},X_{2}(x_{1})\right):x_{1}\in\mathbb{R}\right\} ,\quad\mathcal{S}_{0}^{-}(\tau)=\left\{ \left(x_{1},-X_{2}(x_{1})\right):x_{1}\in\mathbb{R}\right\} .
\]
Let $\mathcal{S}(\tau,c)$ be the acceptance region with critical
value $c^{2}$, i.e., the $c$-enlargement of $\mathcal{S}_{0}(\tau)$,
\[
\mathcal{S}(\tau,c)=\left\{ (x_{1},x_{2}):(x_{1}-\theta_{1})^{2}+(x_{2}-\theta_{2})^{2}\leq c^{2},(\theta_{1},\theta_{2})\in\mathcal{S}_{0}(\tau)\right\} .
\]

\begin{prop}
\label{thm:main}Let $c=\sqrt{Q(\chi_{1}^{2},1-\alpha)}$ and $\hat{\theta}-\theta\sim N(0,I_{2})$. Suppose either $\frac{1-\rho}{\sqrt{\tau(1+\rho)}}\leq\frac{1}{c}$
or $\rho\geq0$. For all $\theta\in\mathcal{S}_{0}(\tau)$, it holds that
\[
P\left(\hat{\theta}\in\mathcal{S}(\tau,c)\right)\geq1-\alpha.
\]
\end{prop}

Proposition \ref{thm:main} shows that the standard MD test  remains valid under a curved null hypothesis when the maximum curvature of $\mathcal{S}_{0}^+(\tau)$, given by $\frac{1-\rho}{\sqrt{\tau(1+\rho)}}$, is sufficiently small, or when the two branches $\mathcal{S}_{0}^{+}(\tau)$ and
$\mathcal{S}_{0}^{-}(\tau)$ are sufficiently close.
The argument proceeds by comparing the coverage of $\mathcal{S}(\tau,c)$
with that of the   auxiliary acceptance
set,
\[\mathcal{S}_{\text{aux}}=\left\{ (x_{1},x_{2}):(x_{2}-\theta_{2})^{2}\leq c^{2}\right\} .\]
whose coverage is exactly $1-\alpha$.  Expressed in polar coordinates, the coverage
depends on  the fraction of each circle of radius $r$ centered at $\theta\in \mathcal{S}^{+}_0(\tau)$, denoted $\partial B(\theta,r)$, that is contained in the acceptance region. Consequently, it suffices to show that, for each $r$, the arc length of $\partial B(\theta,r)$ contained in $\mathcal{S}(\tau,c)$ is no smaller than that of $\mathcal{S}_{aux}$. 

When $r\leq c$, the entire circle is covered by both $\mathcal{S}_{aux}$ and $\mathcal{S}(\tau,c)$ by construction. 
For $r>c$, let   $\mathcal{C}_{u}(\tau)$ and $\mathcal{C}_{\ell}(\tau)$ denote
the upper and lower boundaries of $\mathcal{S}^{+}(\tau,c)$; see
Figure \ref{fig:Coverage-1} left panel. If $\frac{1-\rho}{\sqrt{\tau(1+\rho)}}\leq\frac{1}{c}$, the circle $\partial B\left(\theta,r\right)$
intersects $\mathcal{C}_{u}(\tau)$ at points $A$ and $B$, and $\mathcal{C}_{\ell}(\tau)$ at points $C$ and $D$. We can show that the lengths of chords $\overline{AC}$ and $\overline{BD}$ are no smaller than $2c$. Otherwise, $B(G,c)\not\subseteq\mathcal{S}^{+}(\tau,c)$, where $G$ denotes the intersection of $AC$ with $\mathcal{S}_{0}^{+}(\tau)$, contradicting the
definition of $\mathcal{S}^{+}(\tau,c)$. Note that the arcs of $\partial B(\theta,c)$ covered by $\mathcal{S}_{aux}$ correspond to chords of length $2c$. It therefore follows that the portion of $\partial B(\theta,c)$ covered by $\mathcal{S}(\tau,c)$ is larger than that covered by $\mathcal{S}_{aux}$.
\begin{figure}
\begin{centering}
\includegraphics[scale=0.5]{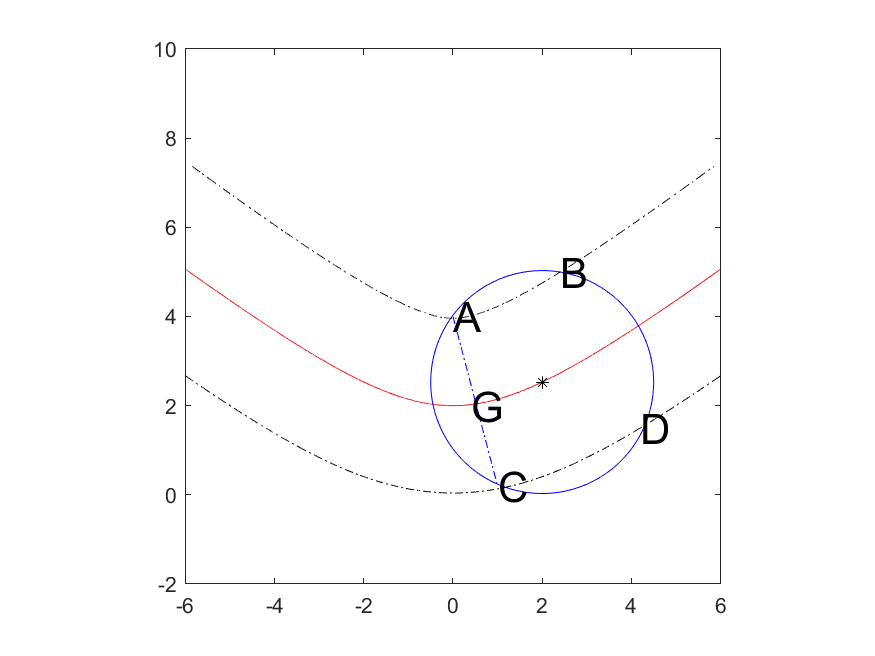}\hspace{-0.8cm}\includegraphics[scale=0.5]{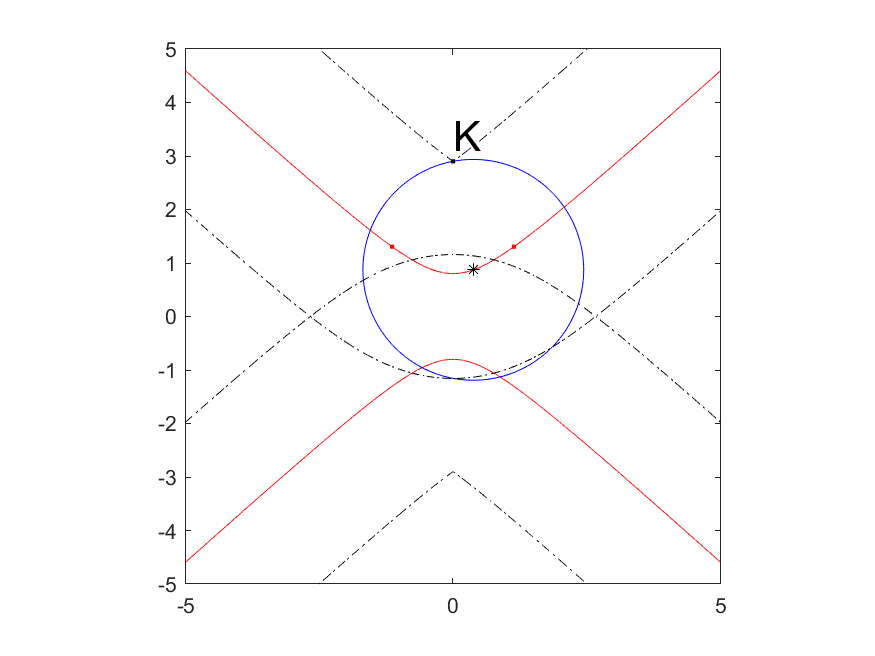}\caption{Acceptance Regions under Low (Left)
and High (Right) Curvature. }\label{fig:Coverage-1}
\par\end{centering}
\noindent\begin{minipage}[t]{1\columnwidth}%
{\scriptsize Red curves represent $\mathcal{S}_{0}^{+}(\tau)$ (left)
and $\mathcal{S}_{0}(\tau)$ (right). Black dash curves represent
the boundaries of $\mathcal{S}^{+}(\tau,c)$ (left) and $\mathcal{S}^{+}(\tau,c)$
and $\mathcal{S}^{-}(\tau,c)$ (right). ``{*}'' represents $\theta$,
and the blue curves represent $\partial B(\theta,r)$ for some $r>c$. }%
\end{minipage}
\end{figure}

If $\frac{1-\rho}{\sqrt{\tau(1+\rho)}}>\frac{1}{c}$,   $\mathcal{C}_{u}(\tau)$
has a kink due to the high curvature. Thus, there exist $\theta\in\mathcal{S}_{0}^{+}(\tau)$
and $r>c$ such that $\partial B\left(\theta,r\right)$ does not intersect
$\mathcal{C}_{u}(\tau)$; see Figure \ref{fig:Coverage-1} right panel. For such $(\theta,r)$,
the argument based on chord lengths no longer applies. 
However, when $\rho\geq0$,  $\partial B(\theta,r)$ is sufficiently close to $\mathcal{S}_{0}^{-}(\tau)$,
and we can show that
$\partial B\left((\theta_{1},\theta_{2}),r\right)\subseteq\mathcal{S}(\tau,c)$.

Next, we present the asymptotic results for general data generating
processes. 
\begin{assumption}
\label{assu:uniform=000020normal}
Suppose that
\begin{enumerate}
\item \label{enu:-is-twice=000020diff} \(\nabla^2 g(\theta_\star)\) has full rank.
\item \label{enu:AsyNormal}Let $BL_{1}$ denote the set of Lipschitz functions
which are bounded by $1$ in absolute value and have Lipschitz constant
bounded by $1$. Assume there exists $r_{n}\rightarrow\infty$ such
that 
\[
\lim_{n\rightarrow\infty}\sup_{P\in\mathcal{P}}\sup_{f\in BL_{1}}\left|E_{P}\left[f\left(\sqrt{r_{n}}\left(\hat{\theta}-\theta_{P}\right)\right)\right]-E_{P}\left[f\left(\xi_{P}\right)\right]\right|=0,
\]
where $\xi_{P}\sim N(0,\Sigma_{P})$.
\item \label{enu:compact=000020Sigma}Let $\mathcal{S}$ denote the set
of matrices with eigenvalues bounded below by $\underline{e}>0$ and
above by $\bar{e}\geq\underline{e}$. For all $P\in\mathcal{P}$,
$\Sigma_{P}\in\mathcal{S}$.
\item \label{enu:consistentSigma}For all $\varepsilon>0$, 
\[
\lim_{n\rightarrow\infty}\sup_{P\in\mathcal{P}}P\left(\left\Vert \hat{\Sigma}-\Sigma_{P}\right\Vert >\varepsilon\right)=0.
\]
\end{enumerate}
\end{assumption}

\Cref{assu:uniform=000020normal}.\ref{enu:-is-twice=000020diff} assumes that the second order derivative of \(g\) at \(\theta_\star\) is of full rank, so that there is no higher order degeneracy. 
\Cref{assu:uniform=000020normal}.\ref{enu:AsyNormal} requires that the researcher has access to an estimator \(\hat\theta\) that is uniformly asymptotically normal over the class of DGPs considered while \Cref{assu:uniform=000020normal}.\ref{enu:compact=000020Sigma} and \Cref{assu:uniform=000020normal}.\ref{enu:consistentSigma} require that the asymptotic variance of this estimator is well-behaved and consistently estimable. Since degeneracy is a property of the transformation of interest, $g$, rather than of the primitive parameter $\theta$, these are mild conditions that can be verified for most common estimators $\hat\theta$.


\begin{thm}
\label{thm:d=00003D2}Suppose $d=2$, and   let $c=\sqrt{Q(\chi_{1}^{2},1-\alpha)}$. Let $(\lambda_{P,1},\lambda_{P,2})$
be the eigenvalues of $\text{sign}\left(g(\theta_{P})-g(\theta_{\star})\right)\Sigma_{P}^{1/2}H\Sigma_{P}^{1/2}$, and define
 $\rho_{P}=\frac{\lambda_{P,1}+\lambda_{P,2}}{\left|\lambda_{P,1}-\lambda_{P,2}\right|}.$
Assume that Assumptions \ref{assm:differentiability} and \ref{assu:uniform=000020normal} hold. If for some $\eta>0$, it holds that either 
\begin{equation}
\mathcal{P}_{n}\subseteq\left\{ P\in\mathcal{P}:\frac{(1-\rho_{P})\sqrt{\left|\lambda_{P,1}-\lambda_{P,2}\right|}}{2r_{n}\sqrt{\left|g(\theta_{P})-g(\theta_{\star})\right|(1+\rho_{P})}}\leq\frac{1}{c},\;\rho_P\in[\eta-1,1-\eta]\right\} \label{eq:P2}
\end{equation}
or 
\begin{equation}
\mathcal{P}_{n}\subseteq\left\{ P\in\mathcal{P}:\rho_{P}\in[0,1-\eta]\right\},\label{eq:P1}
\end{equation}
then
\[
\liminf_{n}\inf_{P\in\mathcal{P}_{n}}P\left(\hat{T}_{n}\left(g(\theta_{P})\right)\leq c^{2}\right)\geq1-\alpha.
\]
\end{thm}
Theorem \ref{thm:d=00003D2} follows from Proposition \ref{thm:main}. To see this, let
$\tilde{g}(h) = g\!\left(\theta_{\star} + r_{n}^{-1}\Sigma^{1/2} h\right),$
where \( r_{n} \) governs the rate at which \( \theta \) is estimated and \( \Sigma \) adjusts for the covariance. For \( h = O(1) \), \( \tilde{g}(h) \) can be approximated by a hyperbola, as in \eqref{eq:H0}. Condition \eqref{eq:P2} ensures that the curvature of \( \tilde{g}(h) \) is not too large, while \eqref{eq:P1} implies that the two branches of the hyperbola are sufficiently close. The result remains uniformly valid even as \( \|h\| \to \infty \).

\begin{remark}
\label{rem:Example=000020cont}To illustrate Theorem \ref{thm:d=00003D2},
consider $g(\theta)=\theta_{1}\theta_{2}$, as motivated by Example~\ref{ex:mediation}. In this case, $H=\left[\begin{array}{cc}
0 & 1\\
1 & 0
\end{array}\right]$. With $\Sigma_{P}=\left[\begin{array}{cc}
\sigma_{1}^{2} & r\sigma_{1}\sigma_{2}\\
r\sigma_{1}\sigma_{2} & \sigma_{2}^{2}
\end{array}\right]$, we have $\lambda_{P,1}=\text{sign}\left(g(\theta_{P})\right)\left(r-1\right)\sigma_{1}\sigma_{2}$, $\lambda_{P,2}=\text{sign}\left(g(\theta_{P})\right)\left(r+1\right)\sigma_{1}\sigma_{2}$, and
$\rho_{P}=\text{sign}\left(g(\theta_{P})\right)r$. Therefore, (\ref{eq:P1})
holds when $r=0$, and the MD test with the simple critical value
yields a uniformly valid confidence interval for the mediation effect.
It is worth noting that, the rejection region for $\theta_{1}\theta_{2}=0$
is given by $\min\left\{ |\hat{\theta}_{1}|,|\hat{\theta}_{2}|\right\} >Q(\chi_{1}^{2},1-\alpha)$,
which coincides with the rejection region of the likelihood ratio test. The latter is the uniformly most powerful invariant test among information- and size-coherent tests (\citet{van2022optimality}). If $\text{sign}\left(g(\theta_{P})\right)r<0$, then  (\ref{eq:P2})
is satisfied when 
$\frac{\sqrt{\sigma_{1}\sigma_{2}}(1+|r|)}{ r_n\sqrt{2(1-|r|)\left|g(\theta_{P})\right|}}\leq\frac{1}{c}.$ Since  $\sigma_1$, $\sigma_2$ and $r$ can be consistently estimated, and $g(\theta_P)$ is known under $H_0$, conditions (\ref{eq:P2}) and (\ref{eq:P1}) are straightforward to verify in practice.  

\end{remark}

\subsection{General Cases}\label{subsec:General-Case}

In this section, we present the inference procedure for a general
function $g$. The procedure is based on a local approximation of
the test statistic. First, consider the case where the true parameter
value $\theta_{n}$ satisfies $\theta_{n}=\theta_{\star}+h_{n}/r_{n}$.
Under $H_0$, the test statistic is given by\\
\[
\hat{T}_{n}(g(\theta_n))=\inf_{\vartheta:g(\vartheta)=g(\theta_{n})}r_{n}^{2}\left(\hat{\theta}-\vartheta\right)^{\prime}\hat{\Sigma}^{-1}\left(\hat{\theta}-\vartheta\right)=r_{n}^{2}\left\Vert \hat{\Sigma}^{-1/2}(\hat{\theta}-\tilde{\theta}_{n})\right\Vert ^{2}
\]
where $\tilde{\theta}_{n}$ denotes the minimizer. Since $\hat{T}_n(g(\theta_n))\leq r_{n}^{2}\left\Vert \hat{\Sigma}^{-1/2}(\hat{\theta}-\theta_{n})\right\Vert ^{2}=O_p (1)$,
we have $\tilde{\theta}_{n}=\theta_{n}+O_{p}(\frac{1}{r_{n}}).$ If $h_n = O(1)$, a
second order Taylor expansion of   $g(\tilde{\theta}_{n})=g(\theta_{n})$
gives 
\begin{align*}
r_{n}^{2}(\tilde{\theta}_{n}-\theta_{\star})^{\prime}H(\tilde{\theta}_{n}-\theta_{\star}) & =h_{n}^{\prime}Hh_{n}+o_{p}(1).
\end{align*}
In addition, let $\mathbb{Z}_{n}=r_{n}\hat{\Sigma}^{-1/2}(\hat{\theta}-\theta_{n})$,
we can write 
\begin{align*}
r_{n}\hat{\Sigma}^{-1/2}(\hat{\theta}-\tilde{\theta}_{n}) & =r_{n}\hat{\Sigma}^{-1/2}\left((\hat{\theta}-\theta_{n})+(\theta_{n}-\theta_{\star})-(\tilde{\theta}_{n}-\theta_{\star})\right)\\
 & =\mathbb{Z}_{n}+\hat{\Sigma}^{-1/2}\left(h_{n}-r_{n}(\tilde{\theta}_{n}-\theta_{\star})\right).
\end{align*}
In sum, let $t=r_n(\tilde{\theta}_n-\theta_\star)$, given $h_{n}$, we can approximate $\hat{T}_{n}$ by 
\begin{equation}
\hat{T}_{n}^{*}(h_{n})=\inf_{t:t^{\prime}Ht=h_{n}^{\prime}Hh_{n}}\left\Vert \mathbb{Z}+\hat{\Sigma}^{-1/2}(h_{n}-t)\right\Vert ^{2}\label{eq:T*}
\end{equation}
where $\mathbb{Z}|\hat{T}_{n}(g(\theta_n))\sim N(0,I_{d})$.  We can show that $\hat{T}_{n}(g(\theta_n))$
and $\hat{T}_{n}^{*}(h_{n})$ have the same asymptotic distribution,
regardless of whether $h_{n}$ converges to $h\in\mathbb{R}$ or diverges
to infinity (Lemmas \ref{lem:Tn=000020=00003D=000020T+o(1)}
and \ref{lem:Tn=000020=00003D=000020T+o1=0000202}). Intuitively, if $h_{n}\rightarrow\infty$, the restriction
for the optimizer $\tilde{t}$ in (\ref{eq:T*}) is approximately linear.  In this case, both $\hat{T}_{n}(g(\theta_n))$ and
$\hat{T}_{n}^{*}(h_{n})$ are approximated $\chi_{1}^{2}$. 

Given $h_{n}$, we can easily get the quantile of $\hat{T}_{n}^{*}(h_n)$
by simulation. However, $h_{n}$ is a nuisance parameter that cannot
be consistently estimated. Next, we propose a two step feasible critical
value. Suppose set $\mathcal{H}_{z}$ satisfies $P\left(N(0,I_{d})\in\mathcal{H}_{z}\right)=1-\eta$.\footnote{For instance, $\mathcal{H}_z=\{z\in \mathbb{R}^d:z^\prime z\leq Q(\chi^2_d,1-\eta)\}.$}
In the first step, we construct a $(1-\eta)$ confidence set for $h_{n}$,
\begin{equation}
\mathcal{H}=r_n(\hat{\theta}-\theta_\star) -\hat{\Sigma}^{1/2}\mathcal{H}_{z}.\label{eq:big=000020H}
\end{equation}
In the second step, we construct the critical value based on the $\frac{1-\alpha}{1-\eta}$
quantile of $\hat{T}_{n}^{*}(h)$ conditional on the first step. That
is, let
\begin{align}
\hat{c} & =\sup_{h\in\mathcal{H}}Q\left(\left.\hat{T}_{n}^{*}(h)\right|\mathbb{Z}\in\mathcal{H}_{z};\frac{1-\alpha}{1-\eta}\right),\label{eq:general=000020g=000020cv}
\end{align}
and reject $H_{0}:g(\theta)=\tau$ if $\hat{T}_{n}(\tau)>\hat{c}$. In (\ref{eq:general=000020g=000020cv}),
the construction of $\hat{c}$ takes into account the first step selection,
thus it is less conservative than simple Bonferroni correction.
\begin{thm}
\label{thm:d>2}Under Assumptions \ref{assm:differentiability} and \ref{assu:uniform=000020normal}, it holds that 
\[
\liminf_{n}\inf_{P\in\mathcal{P}}P\left(\hat{T}_{n}(g(\theta_{P_n}))\leq\hat{c}\right)\geq1-\alpha.
\]
In addition, if $\left\Vert r_{n}(\theta_{P_{n}}-\theta_{\star})\right\Vert \rightarrow\infty$,
\begin{equation}
\lim_{n}P_{n}\left(\hat{T}_{n}(g(\theta_{P_n}))\leq\hat{c}\right)\in\left[1-\alpha,1-\alpha+\eta\right).\label{eq:big=000020h=000020cov=000020rate}
\end{equation}
\end{thm}

The slight conservativeness arises from the two-step procedure. Alternatively,
we can introduce a pretest to check whether $h_{n}$ is far away from
zero, e.g. $||h_{n}||>\ln r_{n}$. If so, we can use the standard
critical value $Q(\chi_{1}^{2},1-\alpha)$. The cost is that we need
to introduce an extra tuning parameter.
\begin{remark}
\label{rem:In-general,-if}In general, if $H$ is singular and $g$
is higher order identified, we can construct the critical value using
a similar two step procedure. In the first step, we construct a $1-\eta$
confidence set $\hat{\Theta}$ for $\theta$. In the second step,
we define the critical value as 
\[
\hat{c}=\sup_{\theta\in\hat{\Theta}}Q\left(\inf_{\vartheta:g(\vartheta)=g(\theta)}\left\Vert \mathbb{Z}+r_{n}\hat{\Sigma}^{-1/2}(\theta-\vartheta)\right\Vert ^{2};1-\alpha+\eta\right).
\] 
$H_{0}:g(\theta)=\tau$ is rejected if $\hat{T}_{n}(\tau)>\hat{c}$. \qed
\end{remark}

\begin{remark}
\citet{Dufour2025WaldTW} show that when $g$ is a vector-valued function
and the degree of singularity differs across elements of $g$, the
Wald-type test statistic may diverge, complicating inference. In contrast,
the MD test considered in this paper yields a test statistic that
is first-order stochastically dominated by $\chi_{d}^{2}$, regardless
of the level of singularity in $g$. Moreover, \citet{Dufour2025WaldTW}  focus solely on hypothesis
tests at a fixed point, i.e., testing $g(\theta)=g(\theta_{\star})$,
whereas this paper aims to construct uniformly valid confidence intervals. \qed
\end{remark}

\begin{remark}
\citet{AM-2016-Geometric} construct a uniformly valid MD test based
on a geometric approach that incorporates the curvature of the null
restriction $g(\theta)=\tau$. When the curvature is large, their
procedure may yield overly conservative critical values. For example,
consider the mediation analysis problem in \Cref{ex:mediation} where one is interested in testing the null hypothesis \(H_0:\theta_1\theta_2 = \tau\). As \(\tau\) approaches zero, the curvature of the null manifold can be made arbitrarily large and the critical value of \citet{AM-2016-Geometric} approaches $Q(\chi_{2}^{2},1-\alpha)$. However, as shown in Section \ref{subsec:Bivariate-with-indefinite} of this paper, a uniformly valid critical value in this setting is $Q(\chi_{1}^{2},1-\alpha)$. Indeed, even when \(\tau\) is far from zero, the \citet{AM-2016-Geometric} critical value is always larger than \(Q(\chi_1^2, 1 - \alpha)\).
\qed
\end{remark}

\section{Simulation}
\label{sec:Simulation}
In this section, we examine the size and power properties of the proposed
procedures and compare them with several alternatives. We focus on
the context of Example \ref{ex:mediation}, namely the construction of confidence intervals
for the mediation effect. In addition to the two MD-based methods
proposed in Section \ref{sec:Inference}, one using the $Q(\chi^2_1,1-\alpha)$
critical value (BN1; Section \ref{subsec:Bivariate-with-indefinite})
and one using a bootstrapped critical value (BN2; Section \ref{subsec:General-Case}),
we consider two uniformly valid MD-based alternatives: (i) the procedure
of \citet{AM-2016-Geometric} (AM),\footnote{We report results using their Section 4.1 implementation, which computes
curvature over a restricted set with tuning parameter $\eta=\alpha/10$.
We also implemented their worst-case curvature procedure from Section 2.
The two procedures have nearly identical power, with the latter performing slightly worse.} and (ii) the MD-based method with projection critical value $Q(\chi_{2}^{2},1-\alpha)$.
For comparison, we also include the naive delta method, i.e. a Wald-type
test with critical value $Q(\chi_{1}^{2},1-\alpha)$, and a naive bootstrap method. The nominal rejection rate is $\alpha=0.05$, and the tuning parameter for BN2 is $\eta=\alpha/10$.

We study confidence intervals for $g(\theta)=\theta_{1}\theta_{2}$,
where the estimators are simulated from 
\[
\hat{\theta}-\theta\sim\mathcal{N}\left(0,\left[\begin{array}{cc}
1 & r\\
r & 1
\end{array}\right]\right).
\]
Without loss of generality, we normalize the variance of $\hat{\theta}$
to one. We consider $r=0$ and $r=0.5$,
with $\theta_{2}\in\{2,6\}$ and $\theta_{1}=[-1:0.2:1]\times\theta_{2}$.

In Figure \ref{fig:Rejection-Rate-at=000020tau}, we plot the probability
that the confidence intervals exclude the true value $g(\theta)$.
The naive method does not overreject when $\theta_{1}\theta_{2}=0$,
but its rejection probability is very low near the origin, consistent
with earlier findings (e.g., \cite{Dufour2025WaldTW}). Away from the origin (see, e.g., $\theta_1=\theta_{2}=2$),
 the naive Wald test overrejects. According to
Remark \ref{rem:Example=000020cont}, BN1 is valid for $r=0$. For
$r=0.5$, Theorem \ref{thm:d=00003D2} does not guarantee validity when
$\theta_{1}<0$, $\theta_{2}=2$, or when $\theta_1\in[-1.4,0)$, $\theta_2 = 6$. Nevertheless, BN1 maintains correct size across all
designs, even when these conditions fail, suggesting that the condition
is sufficient but not necessary. As expected, all other MD-based methods
control size.


Figure \ref{fig:Rejection-Rate-at} shows the probability that the
confidence intervals exclude zero, i.e., the probability of obtaining a significant
result. When $\theta$ is close to the origin ($\theta_2=2$), our methods
have substantially higher power than AM, whose performance is close
to that of the simple projection method. When $\theta$ is further from the
origin ($\theta_{2}=6$), power curves across methods are nearly identical.


Finally, Figure \ref{fig:Median-Length-of} reports the median length
of the confidence intervals, computed across $S$ replications. BN1
consistently yields the shortest intervals, with BN2 close behind.
The projection method is the most conservative, producing intervals
19--30\% longer than BN1. AM lies between BN1 and the projection
method, with median lengths 5--18\% longer than those of  BN1. The differences
are most pronounced when $\theta$ is near the origin.
\begin{center}
\begin{sidewaysfigure}
\subfloat[\label{fig:Rejection-Rate-at=000020tau}Rejection Rate at True Value.
Left to Right: $(\theta_{2},\rho)=(2,0),(6,0),(2,0.5),(6,0.5)$.]{\includegraphics[scale=0.37]{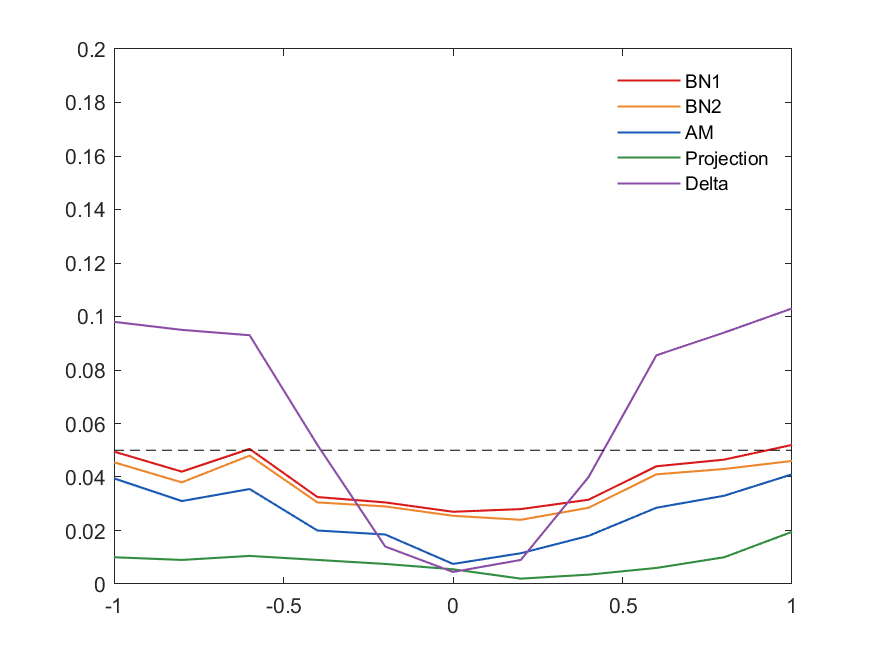}\includegraphics[scale=0.37]{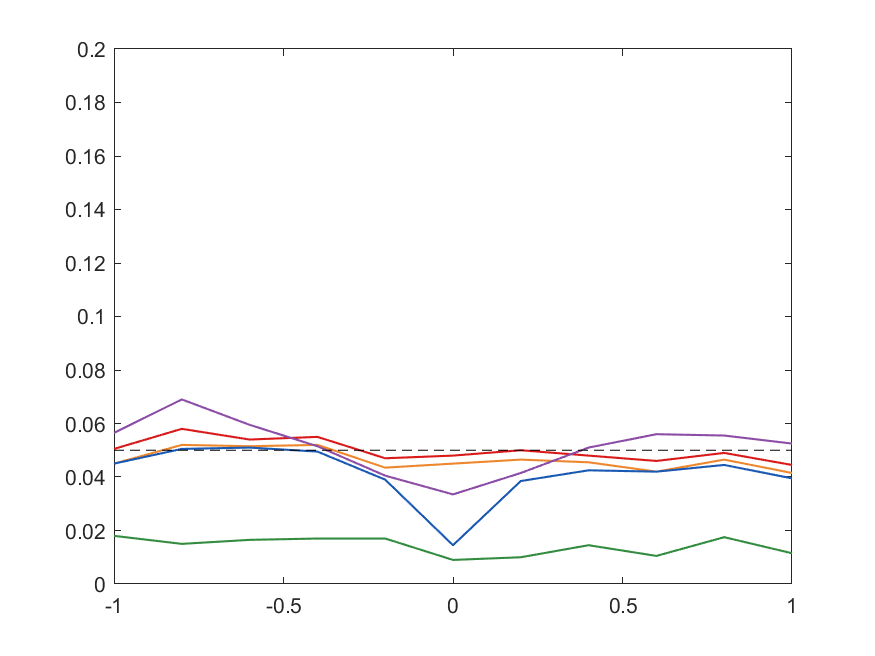}\includegraphics[scale=0.37]{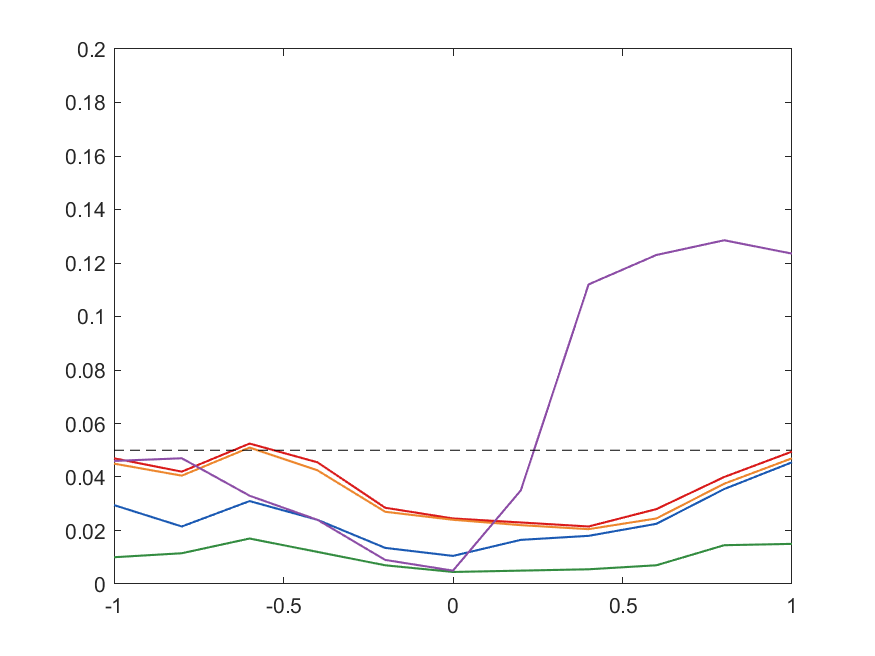}\includegraphics[scale=0.37]{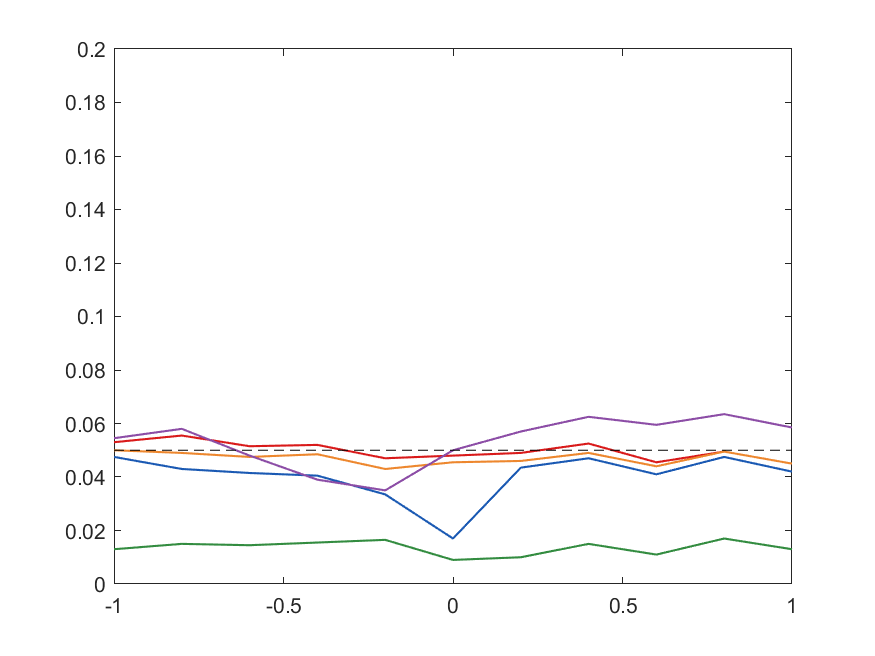}

}

\subfloat[\label{fig:Rejection-Rate-at}Rejection Rate at Zero. Left to Right:
$(\theta_{2},\rho)=(2,0),(6,0),(2,0.5),(6,0.5)$.]{\includegraphics[scale=0.37]{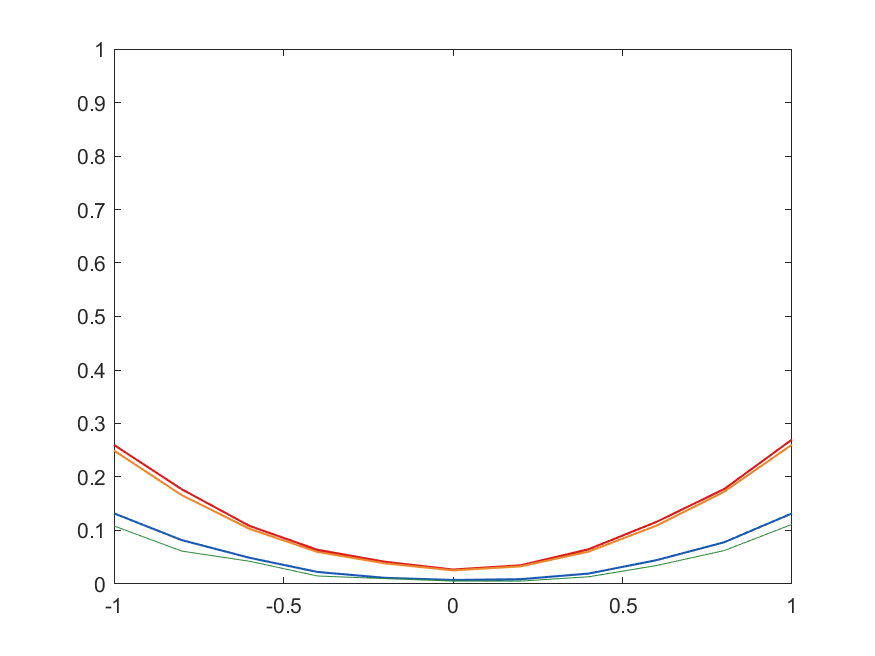}\includegraphics[scale=0.37]{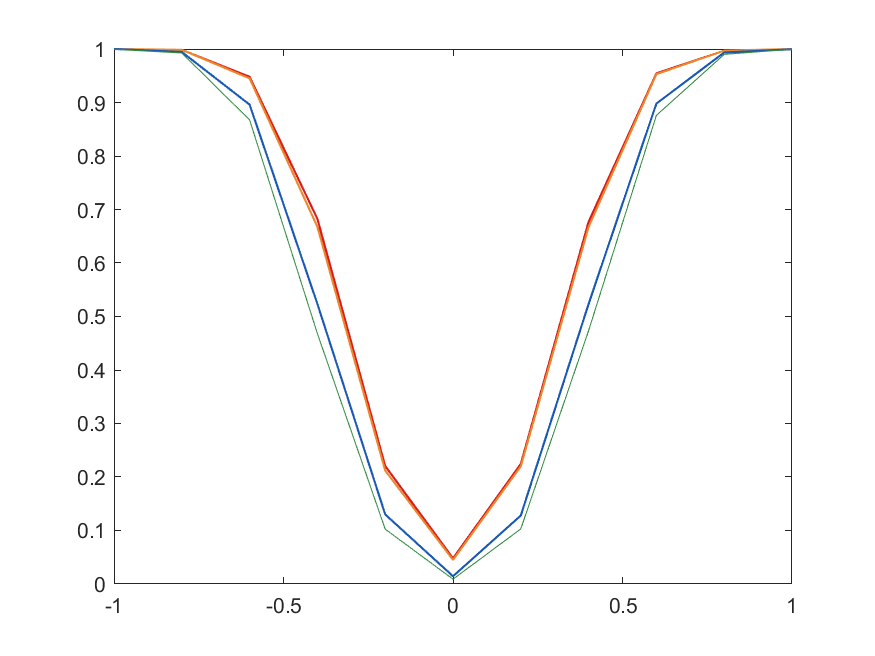}\includegraphics[scale=0.37]{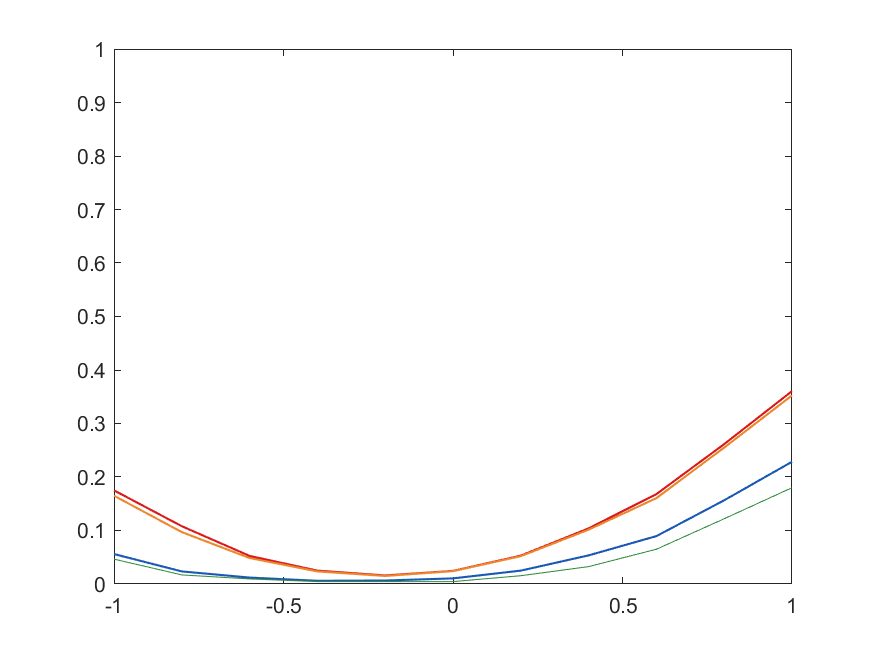}\includegraphics[scale=0.37]{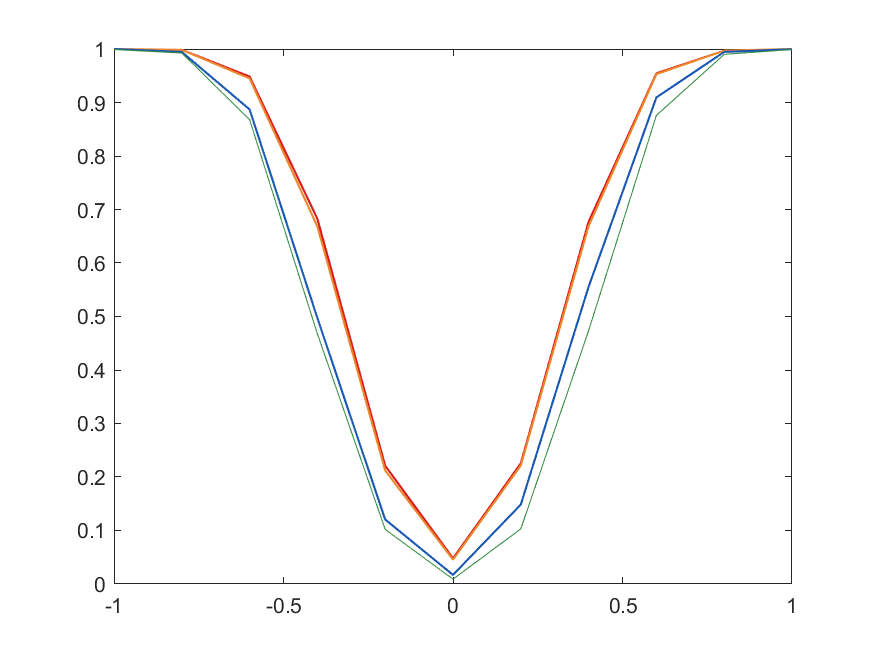}

}

\subfloat[\label{fig:Median-Length-of}Median Length of the Confidence Interval.
Left to Right: $(\theta_{2},\rho)=(2,0),(6,0),(2,0.5),(6,0.5)$.]{\includegraphics[scale=0.37]{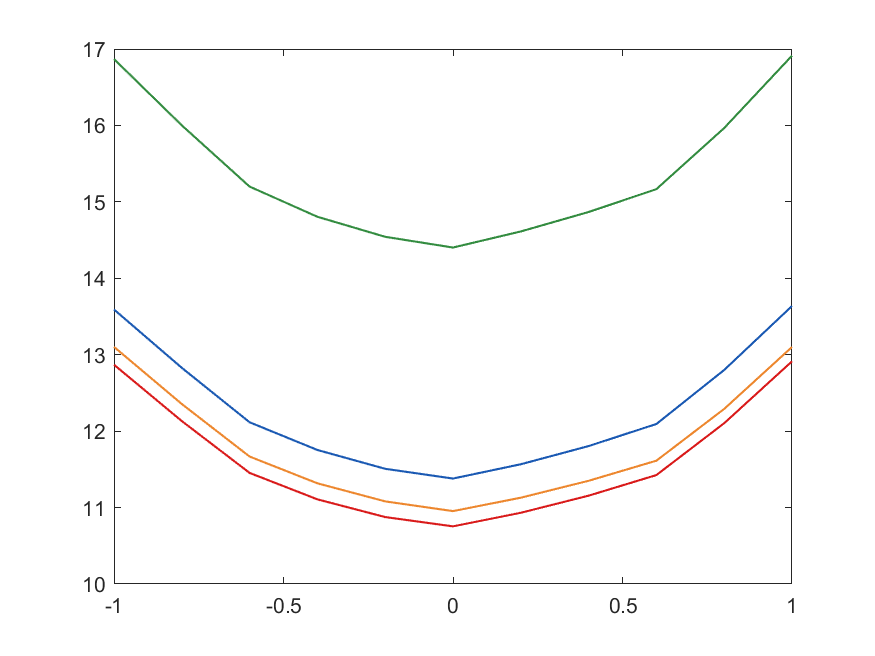}\includegraphics[scale=0.37]{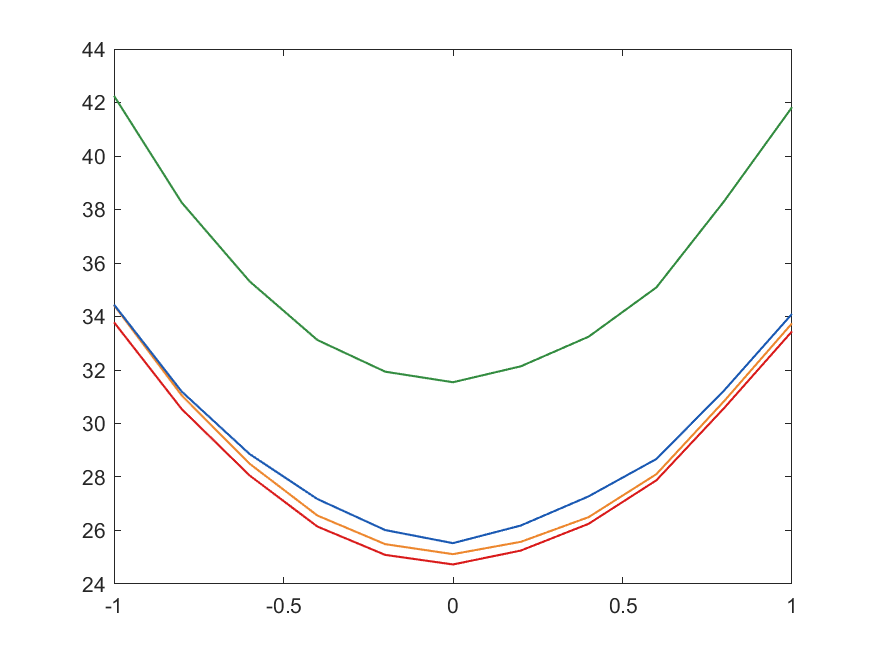}\includegraphics[scale=0.37]{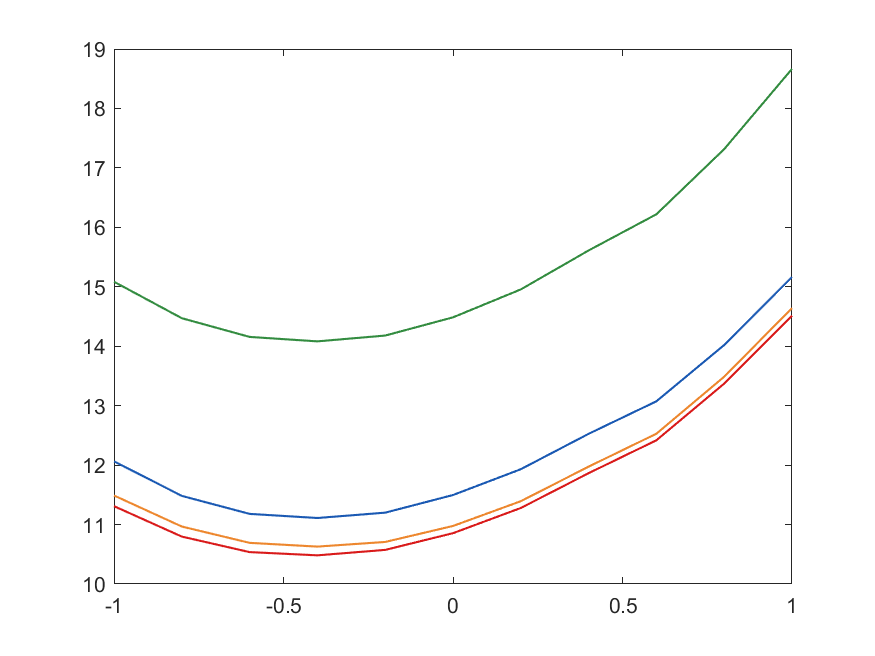}\includegraphics[scale=0.37]{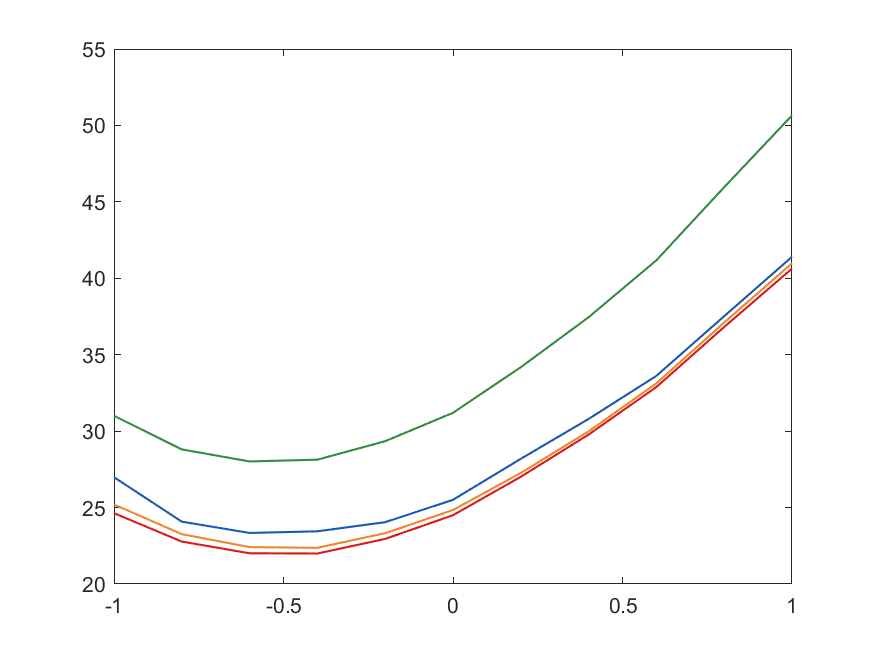}

}

\caption{Simulation Results}\label{fig:Simulation-Results}

\centering{}%
\noindent\begin{minipage}[t]{1\columnwidth}%
\scriptsize The horizontal axis represents $\theta_{1}/\theta_{2}$. The first
and second panels report the probability that the confidence intervals
do not contain the true value $\theta_{1}\theta_{2}$ and zero, respectively.
The third panel reports the median length of the CIs, where the median
is taken over $S=2,000$ samples. %
\end{minipage}
\end{sidewaysfigure}
\par\end{center}

\section{Empirical Application}
\label{sec:Empirical}
We illustrate the empirical relevance of our results using the setting analyzed by \citet{AEM-2018}. Their study takes advantage of a distinctive feature of the Turkish education system, in which elementary school teachers are randomly allocated across schools. This institutional detail generates plausibly exogenous variation in teacher characteristics that can be used to study how teachers' gender role attitudes influence student outcomes. The data include roughly 4,000 third- and fourth-grade students taught by 145 teachers, and students can be grouped according to the length of their exposure to a given teacher --- at most one year, two to three years, or up to four years. The treatment variable is whether a teacher is identified as holding traditional rather than progressive gender beliefs, while the mediator of interest is the student's own gender role beliefs. Following a similar analysis of this data in \citet{van_garderen_nearly_2024}, we focus on verbal test scores as the outcome. \citet{AEM-2018} argue that, after controlling for an extensive set of student, family, teacher, and school characteristics, the identifying assumptions for causal mediation analysis are satisfied in this context.

\begin{table}[htbp]
\centering
\begin{tabular}{|l|c|c|c|c|c|c|}
\hline
\hphantom{\(\Huge(\)}Exposure & $\hat{\theta}_1$ & $t(\hat{\theta}_1)$ & $\hat{\theta}_2$ & $t(\hat{\theta}_2)$ & $\hat{\theta}_1 \cdot \hat{\theta}_2$ & $n$ \\
\hhline{|=|=|=|=|=|=|=|}
Full sample   & 0.199 & 3.140  & -0.119 & -5.343 & -0.024 & 1885 \\
1 year        & 0.256 & 2.052  & -0.097 & -1.941 & -0.025 & 499  \\
2--3 years    & 0.109 & 1.065  & -0.125 & -4.163 & -0.014 & 906  \\
4 years       & 0.064 & 0.513  & -0.113 & -1.931 & -0.007 & 468  \\
\hline
\end{tabular}
\caption{Estimates of Mediation Effects by Teacher Exposure}
\label{tab:altmediation}
\end{table}

\Cref{tab:altmediation} reports estimates from the \citet{AEM-2018} analysis linking teachers’ gender role attitudes to students’ verbal test performance. The first coefficient, $\hat{\theta}_1$, comes from a regression of students’ gender role beliefs on the gender role attitudes of their teachers, with the standard set of student, family, teacher, and school controls included. This coefficient summarizes the extent to which traditional teachers transmit their views to students. The second coefficient, $\hat{\theta}_2$, is estimated from a regression of test scores on both student gender beliefs and teacher attitudes, again with the full set of controls. It reflects how student beliefs are associated with verbal performance once teacher attitudes are held constant. Multiplying these two coefficients gives the mediated, or indirect, effect: the part of the teacher’s influence on scores that operates through the channel of student beliefs. The estimates show that this indirect pathway is negative and relatively small, although it varies across exposure groups, being largest in the one-year sample and smallest for students exposed for four years.

Because the true mediation, or indirect, effect appears to be close to zero, the results of \Cref{sec:Impossibility} suggest that standard approaches to inference will fail. In particular, we cannot construct valid confidence intervals via the typical approach of inverting a \(t\)-test. Instead, we construct confidence intervals using the newly proposed methods of \Cref{sec:Inference}. \citet{van_garderen_nearly_2024} show that the correlation coefficient between $\hat{\theta}_1$ and $\hat{\theta}_2$ is zero; consequently, both of our methods are uniformly valid.\footnote{The validity of BN1 follows from Remark~\ref{rem:Example=000020cont}.}

\renewcommand{\arraystretch}{1.75}
\begin{table}[!htbp]
    \centering\small 
    \resizebox{0.8\columnwidth}{!}{
    \begin{tabular}{|c|cccc|}
        \hline
        Exposure & Full & 1-Year & 2-3 Year & 4 Year \\
        \hhline{|=|====|}
        Point Estimate & \(-0.024\) & \(-0.025\) & \(-0.014\) & \(-0.007\) \\
        \hline
        \(\overset{\bflarrow\text{Interval Length}\flarrow}{\text{95\% \(\mathrm{BN1}\) CI}}\)
        & \(\overset{\bflarrow 0.032 \flarrow}{[-0.042, -0.010]}\)
        & \(\overset{\bflarrow 0.070 \flarrow}{[-0.071, -0.001]}\) 
        & \(\overset{\bflarrow 0.053 \flarrow}{[-0.042, \phantom{-}0.010]}\)
        & \(\overset{\bflarrow 0.070 \flarrow}{[-0.045, \phantom{-}0.025]}\)  \\
    
        95\% \(\mathrm{BN2}\) CI 
        & \(\overset{\bflarrow 0.034 \flarrow}{[-0.044, -0.010]}\)
        & \(\overset{\bflarrow 0.076 \flarrow}{[-0.075, \phantom{-}0.001]}\) 
        & \(\overset{\bflarrow 0.058 \flarrow}{[-0.046, \phantom{-}0.012]}\)
        & \(\overset{\bflarrow 0.076 \flarrow}{[-0.049, \phantom{-}0.027]}\)  \\

        
        95\% AM CI
        & \(\overset{\bflarrow 0.038 \flarrow}{[-0.046, -0.008]}\)
        & \(\overset{\bflarrow 0.086 \flarrow}{[-0.083, \phantom{-}0.003]}\) 
        & \(\overset{\bflarrow 0.068 \flarrow}{[-0.052, \phantom{-}0.016]}\)
        & \(\overset{\bflarrow 0.094 \flarrow}{[-0.059, \phantom{-}0.035]}\)  \\
 
        \(\underset{\vphantom{(}}{\text{95\% Projection CI}}\) 
        & \(\overset{\bflarrow 0.042 \flarrow}{[-0.048, -0.006]}\)
        & \(\overset{\bflarrow 0.092 \flarrow}{[-0.085, \phantom{-}0.007]}\) 
        & \(\overset{\bflarrow 0.070 \flarrow}{[-0.052, \phantom{-}0.018]}\)
        & \(\overset{\bflarrow 0.096 \flarrow}{[-0.059, \phantom{-}0.037]}\)  \\

        \hline
    \end{tabular}}
    \captionsetup{width=0.8\columnwidth, font = small}
    \caption{Mediation Effect  in the data of \citet{AEM-2018} }
    \noindent\begin{minipage}[t]{1\columnwidth}%
{\scriptsize Confidence Intervals are generated by inverting the corresponding tests. Values are rounded to three significant figures. }%
\end{minipage}
    \label{tab:empirical}
\end{table}
We compare our confidence intervals to two other inference procedures that might be applied in this setting, both of which are based on the minimum distance statistic. The first alternate procedure is that of \citet{AM-2016-Geometric}.\footnote{In implementing the test, we follow the empirical application in the working paper version of \citet{AM-2016-Geometric} and only calculate the maximum curvature over a set ``close'' to the point estimate, adjusting the critical value accordingly.} This testing procedure technically does not cover the case where we are testing the null that the mediation effect is equal to zero since the null manifold is not smooth in this case. However, the \citet{AM-2016-Geometric} critical value approaches \(Q(\chi^2_2, 1-\alpha)\) from below as the null hypothesis value approaches zero and \(Q(\chi^2_2, 1 - \alpha)\) is a valid critical value for testing the null that the mediation
effect is equal to zero so we simply modify the procedure slightly to directly use a \(Q(\chi^2_2, 1 - \alpha)\) when the null value is equal to zero. The second method, ``Projection'', simply uses
the \(Q(\chi^2_2, 1 -\alpha)\) at all points, which is justified since, under the null hypothesis, the distance to the null manifold is always less than the distance to the point \((\theta_1,\theta_2)'\).

Consistent with the discussion in \Cref{sec:Inference}, the confidence intervals based on either the $\chi^2_{1}$ critical value (BN1) or the two-step procedure (BN2) are uniformly tighter than those obtained from the \citet{AM-2016-Geometric} simulated critical value (AM); in all specifications, our intervals are strict subsets of theirs. The difference is not only theoretical but also empirically relevant. Using the $\chi^2_{1}$ critical value, for instance, the BN1 confidence interval supports the conclusion of \citet{van_garderen_nearly_2024} that the mediation effect of a one-year exposure to a teacher with traditional views is negative, whereas the alternative methods cannot reject a null of zero at the five-percent level. As expected, the AM intervals lie strictly inside those generated by the projection method, which uses a \(\chi^2_2\) critical value at all points. However, because the true mediation effect appears small in this setting, their simulated critical value converges toward $Q(\chi^2_2, 1 -\alpha)$, which accounts for the close similarity between the two sets of intervals.

\section{Conclusion}
\label{sec:Conclusion}
We examine inference in local regions of first-order degeneracy, meaning that the gradient of the transformation is zero or nearly zero so that first-order approximations alone do not provide reliable information and second-order terms must also be considered. In such regions of local degeneracy, we show that neither regular estimation nor quantile-unbiased procedures are feasible, paralleling impossibility results for nondifferentiable functionals and ruling out standard approaches to inference. We then develop alternate inference procedures based on minimum-distance statistics that deliver uniformly valid confidence intervals. Simulation studies indicate that these procedures control size while maintaining favorable power, and the empirical application to teacher gender attitudes shows that they yield tighter confidence intervals than existing approaches.

\newpage\appendix

\section{Proofs and Supporting Results for Section~\ref{sec:Impossibility}}
\label{sec:Imposs-Proofs}
\textbf{Proof of \Cref{thm:lim-exp}.}
\begin{proof}
   Define \(S_n = r_n^2(\Psi_n - g(\theta_\star))\). Via a second order Taylor expansion, we have
   \begin{align*}
       S_n &= r_n^2(\Psi_n - g(\theta_\star)) \\ 
           &= r_n^2(\Psi_n - g(\theta_\star + h/r_n)) + r_n^2(g(\theta_\star + h/r_n) - g(\theta_\star)) \\ 
           &\overset{h}{\rightsquigarrow} \calL_h + \frac{1}{2}h'\nabla^2g(\theta_\star)h 
   \end{align*}
   where in the last line we use the fact that equation \eqref{eq:ncov} holds for any \(h \in \SR^d\) by hypothesis. Since the experiment \(\{P_\theta: \theta \in \Theta^\circ\}\) satisfies \Cref{assm:lan} with non-singular Fisher information \(\Gamma_{\theta_\star}\), by Theorem 7.10 in \citet{vanDerVaart1998} there is a randomized statistic \(\Psi(X,U)\) in the Gaussian shift experiment \(\{N(h, \Gamma_{\theta_\star}^{-1}): h \in \SR^d\}\) such that \(\Psi(X,U)\) has distribution \(\calL_h + \frac{1}{2}h'\nabla^2 g(\theta_\star) h\) when \(X \sim N(h, \Gamma_{\theta_\star}^{-1})\). Equivalently, \(\Psi(X,U) - \frac{1}{2}h'\nabla^2 g(\theta_\star)h \overset{h}{\sim} \calL_h\).
\end{proof}
\textbf{Proof of \Cref{prop:lim-exp2}.}
\begin{proof}
    (a) We proceed by contradiction, assuming there is an equivariant in law estimator. The characteristic function of the recentered estimator is given by 
    \begin{align}
        \label{eq:char-func}
        \psi(s) = E_h[\exp(is(\Psi(Z,U) - h'Jh))]
    \end{align}
    where, by assumption, \(\psi(s)\) does not depend on \(h\). Let \(\Phi_h(s) = E_h[\exp(is\Psi(Z,U))]\) and notice that \eqref{eq:char-func} implies that we can decompose \(\psi(s)\exp(isf(h)) = \Phi_h(s)\) where we let \(f(h) = h'Jh\) to save notation. We start by showing that \(\Phi_h(s)\) is twice continuously differentiable in \(h\) and deriving expressions for the derivatives. 

    For the first derivative, consider a point \(h_0 \in \SR^d\) and a deviation in the direction \(h\) of size \(r\). We save notation by letting \(\Gamma = \Gamma_{\theta_\star}\) and justify bringing the limit inside the integral by the uniform integrability condition of \citet{hirano2012impossibility}, Lemma 1(b).
    \begin{align*}
       & \lim_{r\downarrow 0} \frac{1}{r}\left[\Phi_{h_0 + rh}(s) - \Phi_{h_0}(s)\right] \\
        =& \lim_{r\downarrow 0} \frac{1}{r}\int_{[0,1]}\int \exp(is\Psi(z,u))\{\phi(z| h_0 + rh, \Gamma^{-1}) -\phi(z|h_0, \Gamma^{-1}\}\,dz\,du\\
        = &\int_{[0,1]}\int \exp(is\Psi(z,u))\lim_{r \downarrow 0}\frac{1}{r}\{\phi(z|h_0 + rh, \Gamma^{-1}) - \phi(z|h_0,\Gamma^{-1})\}\,dz\,du \\ 
        =& \int_{[0,1]}\int \exp(is\Psi(z,u))(z -h_0)'\Gamma h \phi(z|h_0,\Gamma^{-1})\,dz\,du  \\
        = & E_{h_0}[\exp(is\Psi(Z,U))(Z - h_0)']\Gamma h
    \end{align*}
    Since \(h\) is arbitrary here, we can rewrite the above as
    \[
        \nabla \Phi_{h_0}(s) = E_{h_0}[\exp(is\Psi(Z,U))(Z - h_0)']\Gamma
    \]
    where the gradient is understood to be with respect to the argument \(h_0\), i.e \(s\) is kept fixed.
    For the second derivative, we repeat the argument, again letting \(h\) be an arbitrary direction in \(\SR^d\) and justifying bringing the limit into the integral via \citet{hirano2012impossibility}, Lemma 1(b) along with the fact that \(E_{h_0}[\|\exp(is\Psi)(Z-h_0)\|]\) is uniformly bounded over \(h_0\):
    \begin{align*}
        \lim_{r\downarrow 0} &\frac{1}{r}\big[\nabla 
        \Phi_{h_0 + rh}(s) - \nabla \Phi_{h_0}(s)\big] \\ 
        &= \lim_{r\downarrow 0}\frac{1}{r}\Big\{\int_{[0,1]}\int \exp(is\Psi(z,u))(z - h_0)'\Gamma\{\phi(z|h_0 + rh,\Gamma^{-1}) - \phi(z|h_0,\Gamma^{-1})\}\,dz\,du\\
        &\;\;\;\;\;\;\;\;\;\;\;\;\;-\int_{[0,1]}\int \exp(is\Psi(z,u))rh'\Gamma\phi(z|h_0  +rh, \Gamma^{-1})\,dz\,du\Big\} \\
        &= \int_{[0,1]}\int \exp(is\Psi(z,u))(z - h_0)'\lim_{r\downarrow 0}\frac{1}{r}\Gamma\{\phi(z|h_0 + rh,\Gamma^{-1}) - \phi(z|h_0,\Gamma^{-1})\}\,dz\,du\\
        &\;\;\;\;\;\;\;\;\;\;\;\;\;-\int_{[0,1]}\int \exp(is\Psi(z,u))h'\Gamma\lim_{r\downarrow 0}\phi(z|h_0  +rh, \Gamma^{-1})\,dz\,du\Big\} \\
        &= h'\Gamma\int_{[0,1]}\int \exp(is\Psi(z,u))(z - h_0)(z -h_0)'\,\phi(z|h_0,\Gamma^{-1})\,dz\,du\Gamma \\ 
        &\;\;\;\;\;\;- h'\Gamma\int_{[0,1]}\int \exp(is\Psi(z,u))\phi(z|h_0,\Gamma^{-1})\,dz\,du \\ 
        &= h'\Gamma E_{h_0}[\exp(is\Psi(Z,U))(Z - h_0)(Z-h_0)']\Gamma - h'\Gamma E_{h_0}[\exp(is\Psi(Z,U))]
    \end{align*}
    Again, since \(h\) is arbitrary we can write this
    \begin{equation}
        \label{eq:matrix-id}
        \nabla^2 \Phi_{h_0}(s) = \Gamma E_{h_0}[\exp(is\Psi(Z,U))(Z - h_0)(Z - h_0)']\Gamma - \Phi_{h_0}(s)\Gamma
    \end{equation}
    The first and second derivatives of \(\exp(isf(h_0))\) with respect to \(h_0\) can be expressed
    \begin{equation}
        \label{eq:exp-id}
        \begin{split}
            \nabla \exp(is f(h_0)) 
            &= 2is\exp\left(isf(h_0)\right)Jh_0\\
            \nabla^2 \exp\left(is f(h_0)\right)
            &= \exp\left(isf(h_0)\right)\left(2isJ  - 4s^2(Jh_0)(Jh_0)'\right)
        \end{split}
    \end{equation}
    Recall that, by assumption, \(\Phi_{h_0}(s) = \psi(s)\exp(isf(h_0))\) for all \(h_0\). Pick an \(s \neq 0\) such that \(\psi(s) \neq 0\). This is possible since \(\psi(0) = 1\) and \(\psi(\cdot)\) is continuous. Combining \eqref{eq:matrix-id} and \eqref{eq:exp-id} yields, for any \(h_0\), that
    \begin{align}
        \label{eq:final-id}
        \begin{split}
            \psi(s)&\exp\left(isf(h_0)\right)\left(2isJ  - 4s^2(Jh_0)(Jh_0)'\right)\\
            =& \Gamma E_{h_0}[\exp(is\Psi(Z,U))(Z - h_0)(Z - h_0)']\Gamma - \Phi_{h_0}(s)\Gamma
        \end{split}
    \end{align}
    Notice that since \(|\exp(is\Psi(z,u))| = 1\) and \(|\Phi_{h_0}(s)| \leq 1\) for all \(h_0\), the operator norm of the RHS of \eqref{eq:final-id} is bounded uniformly over \(h_0 \in \SR^d\). On the other hand, looking at the LHS of \eqref{eq:final-id} we can see, using \(\|A + B\| \geq \|B\| - \|A\|\), that
    \[
        \|\mathrm{LHS}\| \geq |\psi(s)|\left(4s^2\|Jh_0\|^2 - 2|s|\|J\|\right).
    \]
    Let \(v\) be such that \(\|Jv\| \neq 0\) and let \(h_0 = cv\) for some \(c > 0\) so that \(\|Jh_0\|^2 = c^2\|Jv\|^2\). By sending \(c \to \infty\) we can thus make \(\|\mathrm{LHS}\|\) arbitrarily large, leading to a contradiction since \(\|\mathrm{RHS}\|\) is uniformly bounded over \(h_0 \in \SR^d\).

    (b)  Let \(h\) be such that \(h'Jh \neq 0\). Since \(J\) is assumed symmetric and non-zero, it is guaranteed that such an \(h\) exists. For any \(r \geq 0\), we have that
    \[
        \alpha = P_{(1 + r)h}(\Psi(Z,U) \leq ((1 + r)h)'J((1 + r)h))
    \]
    In particular,
    \[
        0 = \alpha - \alpha = P_{(1 + r)h}(\Psi(Z,U) \leq ((1 + r)h)'J((1 + r)h)) - P_h(\Psi(Z,U) \leq h'J h)
    \]
    and thus 
    \begin{align}
0  =\lim_{r\downarrow0}&\bigg\{\frac{1}{r}\left[P_{(1+r)h}(\Psi(Z,U)\leq(1+r)^{2}h^{\prime}Jh-P_{h}\left(\Psi(Z,U)\leq(1+r)^{2}h^{\prime}Jh\right)\right]\nonumber \\
 & +\frac{1}{r}\left[P_{h}\left(\Psi(Z,U)\leq(1+r)^{2}h^{\prime}Jh\right)-P_{h}\left(\Psi(Z,U)\leq h^{\prime}Jh\right)\right]\bigg\}\label{eq:hexp1}
\end{align}

    Letting \(F_h(x) = P_h(\Psi(Z,U) \leq x)\) and applying the uniform integrability in Lemma 1(a) of \citet{hirano2012impossibility} to justify exchanging limits and integrals as in the proof of \Cref{lemma:aquant}, we obtain for any \(h\)
    \begin{align*}
        h'\Gamma_{\theta_\star}E_h[\bm{1}\{\Psi(X,U) \leq h'J h\}(X - h)] = 2(h'J h)F_h'(h'J h)
    \end{align*}
    From here, take a constant \(c > 0\) and consider the behavior of the LHS and RHS as \(c \to \infty\). Notice that for any \(c > 0\) \(\|c h'\Gamma_{\theta_\star}\| \lesssim c\) while \(\|E_h[\bm{1}\{\Psi(X,U) \leq h'J h\}(X-h)]\| \lesssim 1\) by Cauchy-Schwarz. Meanwhile, \(2((ch)'J (ch)) \propto c^2F_{ch}'((ch)'J(ch))\).  Recall \(F_{h}(h'Jh) = P_h(\Psi(Z,U)  -h'Jh<= 0)\) and $\Psi(Z,U) - h'Jh \sim \calL_h$, \(F_h(h'Jh)\) corresponds to the CDF of \(\calL_{h}\) evaluated at zero.  By assumption, there exists an \(\eps > 0\) such that \(F_h'(h'Jh) \geq \eps\) for all \(h\). Since \(c^2 F_{ch}'((ch)'J(ch)) \to \infty\) as \(c \to \infty\) we arrive at a contradiction.
\end{proof}
\textbf{Proof of \Cref{thm:imposs-main}}.
\begin{proof}
    \Cref{thm:imposs-main} follows directly from \Cref{thm:lim-exp} along with \Cref{prop:lim-exp2}.
\end{proof}
\textbf{Proof of \Cref{thm:impossibility-flat}.}
\begin{proof}
    The first claim follows immediately from the definition of similarity as well as the fact that \(\calP(h)\) is differentiable at zero. For the second claim, it suffices to show that for an indefinite symmetric \(d \times d\) real matrix \(J\) the isotropic set \(\calH_J = \{h \in \SR^d: h'Jh = 0\}\) spans \(\SR^d\). To do so, let us diagonlize \(J = Q\Lambda Q'\) where \(Q\) is orthogonal satisfying \(Q'Q = I\) and \(\Lambda = \diag(\lambda_1,\dots,\lambda_d)\) is a \(d \times d\) diagonal matrix containing the eigenvalues of \(J\). Because \(Q\) is orthogonal, it suffices to show that the isotropic set of \(\Lambda\), \(\calH_\Lambda = \{h \in \SR^d: h'\Lambda h = 0\}\) spans \(\SR^d\).

    Without loss of generality, let us assume that \(\lambda_1 > 0\) and \(\lambda_2 < 0\). Let \(e_1,\dots,e_d\) denote the standard basis vectors in \(\SR^d\). We wish to show that each \(e_j \in \text{span}(\calH_\Lambda)\) for \(j = 1,\dots,d\). If \(\lambda_j = 0\) then trivially \(e_j \in \calH_\Lambda \subseteq \text{span}(\calH_\Lambda)\). Suppose that \(\lambda_j > 0\). Define \(t_j = (-\lambda_j/\lambda_2)^{1/2}\). Consider \(u_j^+ = e_j + t_je_2\) and \(u_j^- = e_j - t_je_2\). Then, notice that
    \begin{align*}
        u_j^+\Lambda u_j^+ &= \lambda_j + t_j^2\lambda_2 = 0 \andbox u_j^- \Lambda u_j^- = \lambda_j + t_j^2\lambda_2 = 0
    \end{align*}
    so that \(u_j^+ \in \calH_\Lambda\) and \(u_j^- \in \calH_\Lambda\). Since \(e_j = \frac{1}{2}(u_j^+ + u_j^-)\) it follows that \(e_j \in \text{span}(\calH_\Lambda)\). The case where \(\lambda_j < 0\) follows symmetrically.

    The claim in \Cref{rem:diffpower} follows from \Cref{lemma:diff-local-power}.
\end{proof}
\textbf{Proof of \Cref{corr:semiparametric-imposs}}
\begin{proof}
    Follows directly from \Cref{thm:imposs-main}.
\end{proof}

\begin{lem}[]
    \label{lemma:aquant}
    Suppose that \(\Psi(Z,U)\) is a statistic in the Gaussian shift experiment \(\{N(h, \Gamma_{\theta_\star}^{-1}): h \in \SR^d\}\) and let \(\calH \subseteq \SR^d\) be a cone such that, for some \(\alpha \in (0,1)\), 
    \[
        \alpha = P_h\left(\Psi(Z,U) \leq \frac{1}{2}h'\nabla^2 g(\theta_\star) h\right),\;\;\;\;\text{ for all } h \in \calH.
    \] 
    Let \(F_\Psi(\cdot)\) denote the CDF of \(\Psi(Z,U)\) under \(h = 0\). Assume that the derivative of \(F_\Psi\) exists at zero. Then \(h'\Gamma_{\theta_\star}E_0[\bm{1}\{\Psi(Z,U) \leq 0\}Z] = 0\) for all \(h \in \calH\).
\end{lem}

\begin{proof}
    The proof of the following lemma closely follows that of Proposition 1(c) in \citet{hirano2012impossibility}. To simplify notation, let \(J = \frac{1}{2}\nabla^2 g(\theta_\star)\). For any \(r \geq 0\) we have that
    \[
        \alpha = P_{rh}\left( \Psi(Z,U) \leq (rh)'J(rh)\right) .
    \]
    Evaluating the above expression at \(r > 0\) and \(r = 0\) yields
    \[
        0 = \alpha - \alpha = P_{rh}\left(\Psi(Z,U) \leq (rh)'J(rh)\right) - P_0(\Psi(Z,U) \leq 0).
    \]
    and thus
    \begin{equation}
        \label{eq:lim1}
        \begin{split}
            0 &= \lim_{r \downarrow 0} \bigg\{ \frac{1}{r}\left[P_{rh}(\Psi(Z,U) \leq (rh)'J(rh)) - P_0\left(\Psi(Z,U) \leq (rh)'J(rh)\right)\right] \\ 
              &\;\;\;\;\;\;\;\;\;\;\;\;\;\;\;\;\;\;\;\;\;\;\;\;\;\;+ \frac{1}{r}\left[P_0\left(\Psi(Z,U) \leq (rh)'J(rh)\right) - P_0\left(\Psi(Z,U) \leq 0\right)\right]\bigg\}.
        \end{split}
    \end{equation}
    Each of the terms on the RHS of \eqref{eq:lim1} exist, so we can write the limit of the sum as the sum of the limits. Let \(\phi(\cdot|\mu,\Sigma)\) denote the pdf of a normal distribution with mean \(\mu\) and variance \(\Sigma\). Consider the first term. Applying the uniform integrability condition in Lemma 1(a) of \citet{hirano2012impossibility} to justify interchanging limits and integrals below, we obtain 
    \begin{align*}
        \lim_{r \downarrow 0} \frac{1}{r}\big[&P_{rh}(\Psi(Z,U) \leq (rh)'J(rh)) - P_0(\Psi(Z,U) \leq (rh)'J(rh))\big]  \\
        &= \lim_{r \downarrow 0} \int_{[0,1]}\int \bm{1}\left\{\Psi(z,u) \leq (rh)'J(rh)\right\}\times \frac{1}{r}[\phi(z|rh, \Gamma_{\theta_\star}^{-1}) - \phi(z|0, \Gamma_{\theta_\star}^{-1})]\,dz\,du \\ 
        &= \int_{[0,1]}\int \lim_{r \downarrow 0} \bm{1}\left\{\Psi(z,u) \leq (rh)'J(rh)\right\}\times \frac{1}{r}[\phi(z|rh, \Gamma_{\theta_\star}^{-1}) - \phi(z|0, \Gamma_{\theta_\star}^{-1})]\,dz\,du \\ 
        &= \int_{[0,1]}\int \bm{1}\{\Psi(z,u) \leq 0\}\left(\frac{\partial }{\partial \tilde h} \phi(z|\tilde h, \Gamma_{\theta_\star}^{-1}\right)_{\tilde h = 0} h \,dz\,du \\ 
        &= h'\Gamma_{\theta_\star}\left\{\int_{[0,1]} \int \bm{1}\{\Psi(z,u) \leq 0\}z\phi(z|0,\Gamma_{\theta_\star}^{-1})\,dz\,du\right\}
    \end{align*} 
    Since the  derivative of \(F_\Psi(\cdot)\) at zero exists and \(\frac{\partial }{\partial r} (rh)'\mJ(rh)\big|_{r = 0} = 0\), the second term on the RHS of \eqref{eq:lim1} evaluates to zero. Thus, we obtain for any \(h \neq  0\) that 
    \[
        0 = h'\Gamma_{\theta_\star}\left\{\int_{[0,1]} \int \bm{1}\{\Psi(z,u) \leq 0\}z\phi(z|0,\Gamma_{\theta_\star}^{-1})\,dz\,du\right\}
    \]
    which gives the result.
\end{proof}

\begin{lem}[]
    \label{lemma:flatpower}
    Let \(\Psi(Z,U)\) be a statistic in the Gaussian shift experiment \(\{N(h, \Gamma_{\theta_\star}^{-1}), h \in \SR^d\}\) such that for (i) for some \(\alpha \in (0,1)\) and cone \(\calH \subset \SR^d\), 
    \[
        \alpha = P_h\left(\Psi \leq \frac{1}{2}h'\nabla^2 g(\theta_\star) h\right)\;\;\text{ for all }h \in \calH,
    \] 
    and (ii) the CDF of \(\Psi(Z,U)\) under \(h = 0\), \(F_\Psi(\cdot)\) is differentiable at zero. Consider the level \(\alpha\) test based on \(T\), that is the test that rejects if \(\Psi(Z,U) \leq 0\). Define \(\calP(h) = P_h(\Psi(Z,U) \leq 0)\) the power curve for this test. This power curve is flat around zero in the direction \(h\) in the sense that \(D_h \calP(0)\) exists and is equal to zero.
\end{lem}

\begin{proof}
    Consider a deviation in the direction \(h\). Define \(\calP_h(r) = P_{rh}(\Psi(Z,U) \leq 0)\). We want to show that 
   \[
       \frac{\partial }{\partial r}\calP_h(r) \big|_{r = 0} = \lim_{r \downarrow 0} \frac{P_{rh}(\Psi(Z,U) \leq 0) - P_0(\Psi(Z,U) \leq 0)}{r} = 0
   \]
   Let us expand the above limit and, as in the proof of \Cref{lemma:aquant}, invoke Lemma 1(a)  in \citet{hirano2012impossibility} to justify exchanging a limit and an integral below.
   \begin{align*}
       \frac{\partial }{\partial r}\calP_h(r) \big|_{r = 0}
       &= \lim_{r \downarrow 0} \int_{[0,1]}\int \bm{1}\{\Psi(z,u) \leq 0\}\times  \frac{1}{r}[\phi(z | rh, \Gamma_{\theta_\star}^{-1}) - \phi(z | 0, \Gamma_{\theta_\star}^{-1})]\,dz\,du\\
       &= \int_{[0,1]}\int \bm{1}\{\Psi(z,u) \leq 0\}\times \lim_{r\downarrow 0} \frac{1}{r}[\phi(z | rh, \Gamma_{\theta_\star}^{-1}) - \phi(z | 0, \Gamma_{\theta_\star}^{-1})]\,dz\,du \\ 
       &= h'\Gamma_{\theta_\star}\int_{[0,1]}\int \bm{1}\{\Psi(z,u) \leq  0\} z \phi(z | 0, \Gamma_{\theta_\star}^{-1})\,dz\,du \\ 
       &= h'\Gamma_{\theta_\star}E_0[\bm{1}\{\Psi(Z,U) \leq 0\} Z] \\
       &=  0
   \end{align*}
   where the final equality comes from \Cref{lemma:aquant}.
\end{proof}

\begin{remark}[]
    \label{rem:rem1}
    The proof of \Cref{lemma:flatpower} could be obtained almost directly from the proof of \Cref{lemma:aquant}. However, the statement of \Cref{lemma:aquant} additionally implies that \(\Cov_0(\bm{1}\{\Psi(Z,U) \leq 0\}, Z) = 0\), which is also an interesting restriction on any \(\alpha\)-quantile unbiased estimate. \qed
\end{remark}

\begin{lem}[]
    \label{lemma:diff-local-power}
    Suppose \Cref{assm:lan} holds at \(\theta_\star\) with Fisher information \(\Gamma_{\theta_\star}\). Let \(\Psi_n\) be a real-valued statistic and consider the test \(\Xi_n=\bm{1}\{\Psi_n\geq 0\}\). Let \(\calP\) be the local asymptotic power curve defined in \eqref{eq:local-power-curve}. Suppose that \(\Psi_n \overset{h}{\rightsquigarrow} \calL_h\) for each \(h \in \SR^d\), and let \(F_0\) be the CDF of \(\calL_0\). If \(F_0\) is continuous at zero, then \(\calP(h)=\lim_{n\to\infty}P_{\theta_{n,h}}(\Xi_n=1)\) for each \(h\), \(\calP\) is differentiable at \(0\), and
    \[
        \nabla \calP(0)=\Gamma_{\theta_\star}\,E_0\!\left[\bm{1}\{\Psi(Z,U)\geq 0\}Z\right],
    \]
    where \(Z \sim N(0,\Gamma_{\theta_\star}^{-1})\), \(U \sim \mathrm{Unif}(0,1)\) independently of \(Z\), and \(\Psi(Z,U)\) is a randomized statistic in the Gaussian shift experiment \(\{N(h,\Gamma_{\theta_\star}^{-1}) : h \in \SR^d\}\) such that \(\Psi(Z,U) \sim \calL_h\) for all \(h\).
\end{lem}

\begin{proof}
    By Theorem 7.10 in \citet{vanDerVaart1998}, there exists a randomized statistic \(\Psi(Z,U)\) in the Gaussian shift experiment \(\{N(h,\Gamma_{\theta_\star}^{-1}) : h \in \SR^d\}\) such that \(\Psi(Z,U) \sim \calL_h\) for all \(h\). Since \(F_0\) is continuous at zero, \(P_0(\Psi(Z,U)=0)=0\), and by mutual absolute continuity of \(N(h,\Gamma_{\theta_\star}^{-1})\) in \(h\) we also have \(P_h(\Psi(Z,U)=0)=0\) for all \(h\). Thus, by the Portmanteau theorem,
    \[
        \calP(h) = \lim_{n\to\infty} P_{\theta_{n,h}}(\Psi_n \geq 0) = P_h(\Psi(Z,U)\geq 0).
    \]

    Fix \(h \in \SR^d\) and define \(G(z,u)=\bm{1}\{\Psi(z,u)\geq 0\}\). For \(r \in \SR\), let \(P_{rh}\) denote the law of \(Z \sim N(rh,\Gamma_{\theta_\star}^{-1})\). The likelihood ratio satisfies
    \[
        \frac{dP_{rh}}{dP_0}(Z) = \exp\!\left(rh'\Gamma_{\theta_\star}Z - \tfrac{1}{2}r^2 h'\Gamma_{\theta_\star}h\right),
    \]
    so \(\calP(rh)=E_0\!\left[G(Z,U)\exp\!\left(rh'\Gamma_{\theta_\star}Z - \tfrac{1}{2}r^2 h'\Gamma_{\theta_\star}h\right)\right]\). Differentiating at \(r=0\) and invoking Lemma 1(a) of \citet{hirano2012impossibility} to justify exchanging limits and integrals we obtain:
    \[
        \left.\frac{\partial}{\partial r}\calP(rh)\right|_{r=0}
        = h'\Gamma_{\theta_\star}E_0\!\left[G(Z,U)Z\right].
    \]
    Since the right-hand side is linear in \(h\), \(\calP\) is differentiable at \(0\) with gradient \(\nabla \calP(0)=\Gamma_{\theta_\star}E_0[G(Z,U)Z]\).
\end{proof}

\section{Proofs for Section~\ref{sec:Inference}}
\label{sec:Inference-Proofs}
\subsection{Proofs of the Main Theorems}
We first introduce notation  for the results in Section \ref{subsec:Bivariate-with-indefinite}. Let $\mathcal{C}_{u}(c)$ and $\mathcal{C}_{\ell}(c)$ be
the upper and lower boundaries of $\partial\mathcal{S}^{+}(\tau,c)$,
\begin{align}
\mathcal{C}_{u}(c)=\left\{ \left(C_{u,1}(x_{1},c),C_{u,2}(x_{1},c)\right):x_{1}\in\mathbb{R}\backslash\left(-x^{*}_1,x^{*}_1\right)\right\} \label{Cu(c)}
\end{align}
where
\begin{align*}
C_{u,1}(x_{1},c) & =x_{1}-\frac{c(1-\rho)x_{1}}{\sqrt{(1+\rho)^{2}X_{2}(x_{1})^{2}+(1-\rho)^{2}x_{1}^{2}}},\\
C_{u,2}(x_{1},c) & =X_{2}(x_{1})+\frac{c(1+\rho)X_{2}(x_{1})}{\sqrt{(1+\rho)^{2}X_{2}(x_{1})^{2}+(1-\rho)^{2}x_{1}^{2}}},
\end{align*}
\[
x_{1}^{*}=\begin{cases}
\frac{\sqrt{c^{2}(1-\rho)^{2}-(1+\rho)\tau}}{\sqrt{2}\sqrt{1-\rho}} & \text{if }\tau\leq\frac{c^{2}(1-\rho)^{2}}{1+\rho}\\
0 & \text{otherwise}
\end{cases},
\]
and the lower boundary is 
\begin{align}
\mathcal{C}_{\ell}(c)=\left\{ \left(C_{\ell,1}(x_{1},c),C_{\ell,2}(x_{1},c)\right):x_{1}\in\mathbb{R}\right\} \label{Cl(c)},
\end{align}
where
\begin{align*}
C_{\ell,1}(x_{1},c) & =x_{1}+\frac{c(1-\rho)x_{1}}{\sqrt{(1+\rho)^{2}X_{2}(x_{1})^{2}+(1-\rho)^{2}x_{1}^{2}}},\\
C_{\ell,2}(x_{1},c) & =X_{2}(x_{1})-\frac{c(1+\rho)X_{2}(x_{1})}{\sqrt{(1+\rho)^{2}X_{2}(x_{1})^{2}+(1-\rho)^{2}x_{1}^{2}}}.
\end{align*}
Details on the calculation of $\mathcal{C}_\ell$ and $\mathcal{C}_u$ are given in Lemma \ref{lem:S+}.

If $\tau\leq\frac{c^{2}(1-\rho)^{2}}{1+\rho}$, let the kink of $\mathcal{C}_{u}(\tau,c)$
be 
\begin{equation}
K=\left(0,\frac{\sqrt{2}\sqrt{c^{2}(1-\rho)+\tau}}{\sqrt{1-\rho^{2}}}\right).\label{eq:H}
\end{equation}
Let $\bar{r}(\theta_{1})$ denote the distance between $O=(\theta_{1},X_{2}(\theta_{1}))$ and  $K$,
\begin{equation}
\bar{r}(\theta_{1})=d\left((\theta_{1},X_{2}(\theta_{1})),K\right).\label{eq:rbar}
\end{equation}
Let $B\left((x_{1},x_{2}),r\right)$ denote the ball centered at $(x_{1},x_{2})$
with radius $r$, and $\partial B\left((x_{1},x_{2}),r\right)$ its
boundary (i.e. the circle). Let $\overset{\frown}{AB}$ be the arc
from $A$ to $B$, and $\overline{AB}$ the line segment.

\textbf{Proof of Proposition \ref{thm:main}.}
\begin{proof}
This follows from Proposition \ref{prop:tau=000020large} and \ref{prop:tau<}.
\end{proof}
\begin{prop}
\label{prop:tau=000020large}Let $\tau\geq\frac{c^{2}(1-\rho)^{2}}{1+\rho}$ and $|\rho|<1$.
For all $\theta=(\theta_{1},\theta_{2})\in\mathcal{S}_{0}^{+}(\tau)$,
$(Z_{1},Z_{2})\sim N(0,I_{2})$,
\[
P\left((Z_{1}+\theta_{1},Z_{2}+\theta_{2})\in\mathcal{S}^{+}(\tau,c)\right)\geq1-\alpha.
\]
where
\[
\mathcal{S}^+(\tau,c)=\left\{ (x_{1},x_{2}):(x_{1}-\theta_{1})^{2}+(x_{2}-\theta_{2})^{2}\leq c^{2},(\theta_{1},\theta_{2})\in\mathcal{S}^+_{0}(\tau)\right\} .
\]

\end{prop}

\begin{proof}
The proof is based on Lemma \ref{lem:key=000020lemma}, with $\bar{\mathcal{S}}=\mathcal{S}^{+}(\tau,c)$.
Condition 1 of Lemma \ref{lem:key=000020lemma} holds trivially. We
now verify Condition 2 of Lemma \ref{lem:key=000020lemma}. Let $r>c$.
By Lemma \ref{lem:S0^=000020cap=000020B}, $\partial B\left(\theta,r\right)$
intersects $\mathcal{S}_{0}^{+}(\tau)$ at a minimum of two points
$I$ and $J$, with $I$ to the left of $J$. See Figure
\ref{fig:Coverage-1-1}. By Lemma \ref{lem:I=000020and=000020J}.\ref{enu:C1=000020tau<-2},
there is a point $P$ in $\partial B\left(\theta,r\right)$ above
curve $\mathcal{C}_{u}(\tau,c)$. Therefore, there is at least one
point on $\partial B\left(\theta,r\right)$ between $P$ and $I$
that intersects $\mathcal{C}_{u}(\tau,c)$. Let the closest point
(if there's more than one point) to $I$ be point $A$. Similarly,
define $B$ as the point on $\partial B\left(\theta,r\right)$ between
$P$ and $J$ that intersects $\mathcal{C}_{u}(\tau,c)$ and is closest
to $J$. By Lemma \ref{lem:I=000020and=000020J}.\ref{enu:C2-1-1-1},
there is a point $Q$ on $\partial B\left(\theta,r\right)$ below
curve $\mathcal{C}_{\ell}(\tau,c)$. Therefore, there is at least
one point on $\partial B\left(\theta,r\right)$ between $Q$ and $I$
that intersects $\mathcal{C}_{\ell}(\tau,c)$. Let the closest point
(if there's more than one point) to $I$ be point $C$. Similarly
define $D$ between $Q$ and $J$. Therefore, by construction $\overset{\frown}{AIC}\subset\mathcal{S}^{+}(\tau,c)$
and $\overset{\frown}{BJD}\subset\mathcal{S}^{+}(\tau,c)$. 

To show that $\text{length}\left(\overset{\frown}{AIC}\right)+\text{length}\left(\overset{\frown}{BJD}\right)\geq4r\arcsin\frac{c}{r}$,
it suffices to show that 
\[
\text{length}\left(\overline{AC}\right)\geq2c\text{ and }\text{length}\left(\overline{BD}\right)\geq2c.
\]
By contradiction, assume that $\text{length}\left(\overline{AC}\right)<2c$.
Let $AC$ intersects $\mathcal{S}_{0}^{+}(\tau)$ at point $G$, then
$B(G,c)\not\subseteq\mathcal{S}^{+}(\tau,c)$, which contradicts the
definition of $\mathcal{S}^{+}(\tau,c)$. Therefore, Condition 1 and
2 of Lemma \ref{lem:key=000020lemma} hold, and 
\begin{equation}
P\left((Z_{1}+\theta_{1},Z_{2}+\theta_{2})\in\mathcal{S}^{+}(\tau,c)\right)\geq1-\alpha.\label{eq:S+>1-alpha}
\end{equation}
This completes the proof. 
\begin{figure}[h]
\begin{centering}
\includegraphics[scale=0.5]{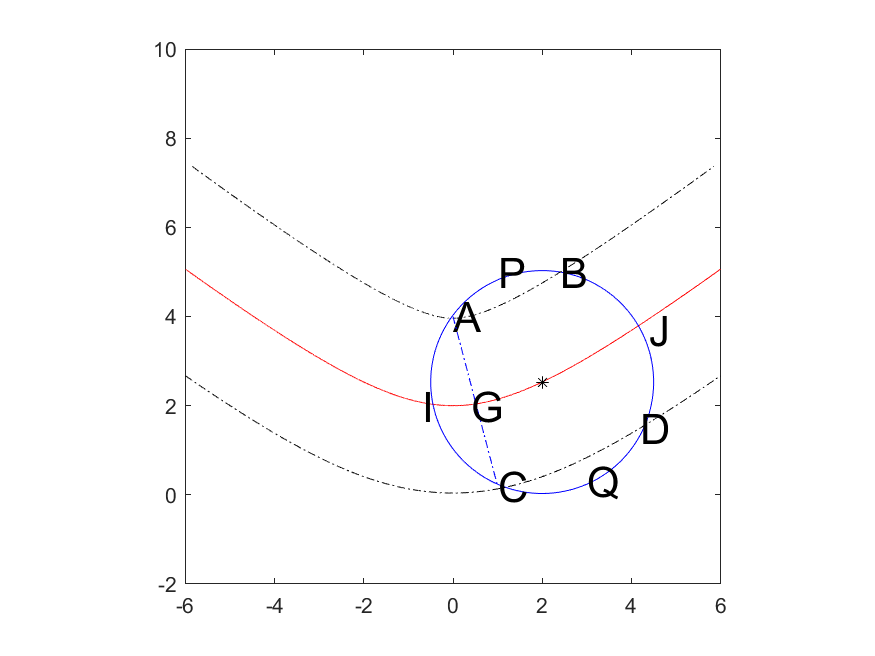}\caption{Illustration of Proposition \ref{prop:tau=000020large}.}\label{fig:Coverage-1-1}
\par\end{centering}
\noindent\begin{minipage}[t]{1\columnwidth}%
{\scriptsize Red curves represent $\mathcal{S}_{0}^{+}(\tau)$. Black dashed curves represent
the boundaries of $\mathcal{S}^{+}(\tau,c)$. ``{*}'' represents $\theta$,
and the blue curves represent $\partial B(\theta,r)$ for some $r>c$.}%
\end{minipage}
\end{figure}
\end{proof}
\begin{prop}
\label{prop:tau<}Let $0<\tau<\frac{(1-\rho)^{2}c^{2}}{1+\rho}$ and
$\rho\in[0,1)$. For all $\theta=(\theta_{1},\theta_{2})\in\mathcal{S}_{0}(\tau)$,
$(Z_{1},Z_{2})\sim N(0,I_{2})$,
\[
P\left((Z_{1}+\theta_{1},Z_{2}+\theta_{2})\in\mathcal{S}(\tau,c)\right)\geq1-\alpha.
\]
\end{prop}

\begin{proof}
The proof is based on Lemma \ref{lem:key=000020lemma} with $\bar{\mathcal{S}}=\mathcal{S}(\tau,c)$.
WLOG, let $\theta\in\mathcal{S}_{0}^{+}(\tau)$. By Lemma \ref{lem:I=000020and=000020J}.\ref{enu:C1=000020tau<-2}
and \ref{lem:I=000020and=000020J}.\ref{enu:C2-1-1-1}, for all $|\theta_{1}|\geq x_{1}^{*}$, the coverage is at least $1-\alpha$
with the same argument as in Proposition \ref{prop:tau=000020large}.

For $\theta_{1}\in(-x_{1}^{*},x_{1}^{*})$. If $r>\bar{r}(\theta_{1})$, where $\bar{r}(\theta_{1})$ is defined in (\ref{eq:rbar}),
$\partial B\left(O,r\right)$ intersects with $\mathcal{C}_{u}$ at
least two points, since (i) by Lemma \ref{lem:I=000020and=000020J}.\ref{lem:min=000020dist=000020H},
\begin{align*}
\inf_{x\in\mathcal{C}_{u}}d\left(x,\theta\right) & =d\left(K,\theta\right)=\bar{r}(\theta_{1})<r,
\end{align*}
(ii) 
\[
\sup_{x\in\mathcal{C}_{u},x_{1}>0}d\left(x,\theta\right)=\sup_{x\in\mathcal{C}_{u},x_{1}<0}d\left(x,\theta\right)=\infty.
\]
Therefore 
\[
\left|\operatorname{length}\left(\partial B\left(\theta,r\right)\cap\mathcal{S}^{+}(\tau,c)\right)\right|\geq4r\arcsin\frac{c}{r}
\]
following from the same argument in Proposition \ref{prop:tau=000020large}. 

Then we show that for $r\in(c,r(\theta_{1})]$,
\begin{equation}
\left|\operatorname{length}\left(\partial B\left(\theta,r\right)\cap\bar{\mathcal{S}}\right)\right|=2\pi r\geq4r\arcsin\frac{c}{r}.\label{eq:tau<,r}
\end{equation}
By Lemma \ref{lem:I=000020and=000020J}.\ref{enu:B=000020cap=000020Cl=000020},
$\partial B\left(\theta,r\right)\cap\mathcal{C}_{\ell}(c)\cap\left\{ (0,x):x\geq0\right\} =\emptyset$.
By Lemma \ref{lem:I=000020and=000020J}.\ref{lem:min=000020dist=000020H},
$\partial B\left(\theta,r\right)\cap\mathcal{C}_{u}(c)=\emptyset$.
Therefore, $B\left(\theta,r\right)\cap\left\{ (0,x):x\geq0\right\} \subseteq\mathcal{S}^{+}(\tau,c)$.
Case 1: if $$B\left(\theta,r\right)\cap\left\{ (0,x_{2}):x_{2}\leq0\right\} =\emptyset,$$
then (\ref{eq:tau<,r}) holds. Case 2: if $\exists(x_{1},0)\in\partial B\left(\theta,r\right)$,
 by Lemma \ref{lem:=000020x<xbar}.\ref{enu:covered=000020},
$$B\left(\theta,r\right)\cap\left\{ (0,x):x\leq0\right\} \subseteq\mathcal{S}^{-}(\tau,c).$$
In sum, $\partial B\left(\theta,r\right)\subseteq\mathcal{S}(\tau,c)$.
\end{proof}
\textbf{Proof of Theorem \ref{thm:d=00003D2}.}
\begin{proof}
There exists a subsequence $P_{n_{j}}\in\mathcal{P}_{n}$ such that
\begin{align*}
\liminf_{n}\inf_{P\in\mathcal{P}_{n}}P\left(\hat{T}_n(g(\theta_{P}))\leq c^2\right) & =\lim_{j\rightarrow\infty}P_{n_{j}}\left(\hat{T}_n(g(\theta_{n_{j}}))\leq c^2\right)\text{ where }\theta_{n}=\theta_{P_{n}}.
\end{align*}
Since $\Theta$ is compact, the sequence $\{\theta_{n_{j}}\}$ is
bounded and thus has a further subsequence $n_{j^{\prime}}$ such
that $\lim_{n_{j^{\prime}}}\theta_{n_{j^{\prime}}}=\theta_{\infty}\in\Theta$,
 $\lim_{n_{j^{\prime}}}r_{n_{j^{\prime}}}(\theta_{n_{j^{\prime}}}-\theta_{\star})= h\in\mathbb{R}_{[\pm\infty]}^{2}$, $\lim_{n_{j^{\prime}}}R_{P_{n_j}}= R$, and $\rho_{P_n}\rightarrow\rho$.
With slight abuse of notation, we will refer to this convergent subsequence
as $\{\theta_{n}\}$ from here on. In addition, we use $n$ instead of $P_n$ for subscript in $\lambda,R$, and $\rho$.

Case 1. If $\lim_{n}\theta_{n}=\theta_{\infty}\not=\theta_{\star}$,
standard minimum distance arguments (see, for example Section 9.1
in Newey and McFadden (1994)) imply that
\[
\lim_{n}P_{n}\left(\hat{T}_{n}(g(\theta_{n}))\leq Q(\chi_{1}^{2},1-\alpha)\right)\geq1-\alpha.
\]

Case 2. Suppose $\theta_{\infty}=\theta_{\star}$ and $\lim_{n}r_{n}(\theta_{n}-\theta_{\star})=h\in\mathbb{R}^{2}$.
We first normalize the problem to match the notation in Proposition\textbf{
}\ref{thm:main}. WLOG, assume $\lambda_{n,2}\geq\lambda_{n,1}$.
There exists an orthogonal matrix $R_{n}$ such that 
\[
\text{sign}(g(\theta_{n})-g(\theta_\star))\Sigma_{n}^{1/2}H\Sigma_{n}^{1/2}
=R_{n}^{\prime}\scalebox{0.85}{$\begin{bmatrix}
\lambda_{n,1} & 0\\
0 & \lambda_{n,2}
\end{bmatrix}$}R_{n}
=\frac{\lambda_{n,2}-\lambda_{n,1}}{2}R_{n}^{\prime}\scalebox{0.85}{$\begin{bmatrix}
\rho_{n}-1 & 0\\
0 & \rho_{n}+1
\end{bmatrix}$}R_{n}.
\]

Define
\[
\vartheta_{n}=R_{n}\Sigma_{n}^{-1/2}\theta_{n},\;\vartheta_{\star,n}=R_{n}\Sigma_{n}^{-1/2}\theta,\;\quad\tilde{g}_n(\vartheta)=\frac{\text{sign}(g(\theta_{n})-g(\theta_{\star}))}{\lambda_{n,2}-\lambda_{n,1}}g\left(\Sigma_{n}^{1/2}R_{n}^{\prime}\vartheta\right).
\]
By construction, 
\begin{align*}
\frac{\partial\tilde{g}(\vartheta_{\star,n})}{\partial\vartheta}=0,
&\;
\frac{\partial^{2}\tilde{g}(\vartheta_{\star,n})}
{\partial\vartheta\,\partial\vartheta^{\prime}}
=\frac{\text{sign}(g(\theta_{n})-g(\theta_\star))}{\lambda_{n,2}-\lambda_{n,1}}
R_{n}\Sigma_{n}^{1/2}H\Sigma_{n}^{1/2}R_{n}^{\prime} 
&=\frac{1}{2}
\scalebox{0.85}{$\begin{bmatrix}
\rho_{n}-1 & 0\\
0 & \rho_{n}+1
\end{bmatrix}$}.
\end{align*}

Let $\hat{\vartheta}_{n}=R_{n}\Sigma^{-1/2}\hat{\theta}_{n}$. Under Assumption
\ref{assu:uniform=000020normal},
it holds that $r_{n}(\hat{\vartheta}_{n}-\vartheta_{n})\xrightarrow{d}N(0,I_{2})$
uniformly when $\lim_{n}r_{n}(\theta_{n}-\theta_{\star})=h\in\mathbb{R}^{2}$. Let 
$\tilde{H}
=\scalebox{0.85}{$\begin{bmatrix}
\rho-1 & 0\\
0 & \rho+1
\end{bmatrix}$}.$

Under Assumption \ref{assu:uniform=000020normal}, by the almost sure
representation theorem, there exists a probability space with random
variables $\mathbb{Z}_{n}$ and $\mathbb{Z}$ defined on it such that
(i) $\mathbb{Z}_{n}$ has the same distribution as $r_n(\hat{\vartheta}_{n}-\vartheta_{n})$,
(ii) $\mathbb{Z}\sim N(0,I_{2})$, and (iii) $\mathbb{Z}_{n}\xrightarrow{\text{a.s.}}\mathbb{Z}$.
Define $\tilde{h}_{n}=r_{n}(\vartheta_{n}-\vartheta_{\star,n})$, 
\begin{align*}
\hat{T}_{n}(g(\theta_{n}))	&=\inf_{\theta:g(\theta)=g(\theta_{n})}\left\Vert r_{n}R_{n}\hat{\Sigma}^{-1/2}(\hat{\theta}-\theta)\right\Vert ^{2}\\
	&=\inf_{\theta:g(\theta)=g(\theta_{n})}\left\Vert r_{n}R_{n}\Sigma^{-1/2}(\hat{\theta}-\theta)\right\Vert ^{2}+o_{p}(1)\\
	&=\inf_{\vartheta:\tilde{g}_{n}(\vartheta_{\star,n}+\vartheta-\vartheta_{\star,n})=\tilde{g}_{n}(\vartheta_{n})}\left\Vert r_{n}(\hat{\vartheta}-\vartheta_{n})+\tilde{h}_{n}-r_{n}(\vartheta-\vartheta_{\star,n})\right\Vert ^{2}+o_{p}(1)\\
	&=\inf_{\vartheta:\tilde{g}_{n}(\vartheta_{\star,n}+r_{n}^{-1}x)=\tilde{g}_{n}(\vartheta_{n})}\left\Vert r_{n}(\hat{\vartheta}-\vartheta_{n})-(x-\tilde{h}_{n})\right\Vert ^{2}+o_{p}(1).
\end{align*}
The second line follows from $r_n(\hat{\theta}-\theta)=O_p(1)$, with $\theta$ denoting the optimizer, and from the consistency of $\hat{\Sigma}$. The third and fourth lines follow from rearranging terms.

It follows that $\hat{T}_{n}\sim T_{n}+o_p(1)$, where
\[
T_{n}=\inf_{\vartheta:\tilde{g}_n(\vartheta_{\star,n}+r_{n}^{-1}x)=\tilde{g}_n(\vartheta_{n})}\left\Vert \mathbb{Z}_{n}-(x-\tilde{h}_{n})\right\Vert ^{2}.
\]
Let 
\begin{equation}
T=\inf_{x:x^{\prime}\tilde{H}x=\tilde{h}^{\prime}\tilde{H}\tilde{h}}\left\Vert \mathbb{Z}-(x-h)\right\Vert ^{2}.\label{eq:T}
\end{equation}
By Lemma \ref{lem:Tn=000020=00003D=000020T+o(1)}, $T_{n}=T+o_{p}(1).$
Moreover, by Lemma \ref{lem:continus=000020T}, $T$ is continuously
distributed. Hence
\[
\lim_{n}P_{n}(T_{n}\leq c^{2})=P(T\leq c^{2})\geq1-\alpha,
\]
where the inequality follows from Proposition \ref{thm:main} with
condition either 
\begin{align*}
    \tilde{h}^{\prime}\tilde{H}\tilde{h}&=\lim_n 4 r_{n}^{2}\left(\tilde{g}(\vartheta_{n})-\tilde{g}(\vartheta_{\star})\right)\\
&=\lim_n\frac{4\text{sign}(g(\theta_{n})-g(\theta_{\star}))}{\lambda_{n,2}-\lambda_{n,1}}r_{n}^{2}\left(g(\theta_{n})-g(\theta_{\star})\right)\geq\frac{c^{2}(1-\rho)^{2}}{1+\rho}
\end{align*}

or $\rho\in[0,1-\eta]$.

Case 3. Suppose $\theta_{n}\rightarrow\theta_{\star}$ and $\lim_{n}r_{n}(\theta_{n}-\theta_{\star})\rightarrow\infty$.
Define $s_{n}=\frac{1}{||\theta_{n}-\theta_{\star}||}\ll r_{n}$.
By construction, $\left\Vert s_{n}(\theta_{n}-\theta_{\star})\right\Vert =1$,
so there exists a subsequence such that $\lim s_{n}(\theta_{n}-\theta_{\star})=\lim h_{n}=h\in\mathbb{R}^{d}\backslash\{0_{d}\}$.
Similar to Case 2, let $\mathbb{Z}_n$ has the same distribution as $\hat{\Sigma}^{-1/2}r_n(\hat{\theta}_n-\theta_n)$,
\begin{align*}
T_{n} & =\inf_{x:g(\theta_{n}+r_{n}^{-1}x)=g(\theta_{n})}\left\Vert \mathbb{Z}_{n}-\hat{\Sigma}^{-1/2}x\right\Vert ^{2}.
\end{align*}
Note that $\hat{T}_{n}\sim T_{n}$.
By Lemma \ref{lem:Tn=000020=00003D=000020T+o1=0000202}, $T_{n}=T+o_{p}(1)$,
where
\begin{equation}
T=\inf_{x:h^{\prime}Hx=0}\left\Vert \mathbb{Z}-\Sigma_{n}^{-1/2}x\right\Vert ^{2}.\label{eq:T2}
\end{equation}
Since $h^{\prime}H\neq0$, $T\sim\chi_{1}^{2}$, it holds that
\[
\lim_{n}P_{n}(T_{n}\leq c^{2})=P(T\leq c^{2})=1-\alpha.
\]
\end{proof}
\textbf{Proof of Theorem \ref{thm:d>2}.}
\begin{proof}
There exists a subsequence   $P_{n_j}\in\mathcal{P}_{n_j}$ such that 
\begin{align*}
\liminf_{n}\inf_{P\in\mathcal{P}_{n}}P\left(\hat{T}_{n}(g(\theta_{n}))\leq\hat{c}\right) & =\lim_{j}P_{n_j}\left(\hat{T}_{n}(g(\theta_{n_j}))\leq\hat{c}\right).
\end{align*}
Since $\Theta$ and $\mathcal{S}$ are compact, the sequences $\{\theta_{n_j}\}$
and $\Sigma_{n_j}$ have further subsequences $n_k$ such that $\lim_{k}\theta_{n_k}=\theta_{\infty}\in\Theta$,
$\lim_{n_k}r_{n_k}(\theta_{n_k}-\theta_{\star})\rightarrow h\in\mathbb{R}_{[\pm\infty]}^{d}$,
and $\lim_{k}\Sigma_{n_k}\rightarrow\Sigma\in\mathcal{S}$. With slight
abuse of notation, we will refer to this convergent subsequence as
$\{n\}$ from here on. 

Case 1. If $\lim_{n}\theta_{n}=\theta_{\infty}\not=\theta_{\star}$,
standard minimum distance arguments apply and will show that $\hat{T}_{n}(g(\theta_n))\overset{P_{n}}{\rightsquigarrow}\chi_{1}^{2}.$ Let $z\in\mathcal{H}_z$.
By construction, $\hat{h}_{n}=r_{n}(\hat{\theta}-\theta_{\star})-\hat{\Sigma}_n^{1/2}z\in\mathcal{H}$.
In addition, $\hat{h}_{n}/r_{n}\xrightarrow{p}(\theta_{\infty}-\theta_{\star})\neq0$,
and since $H$ has full rank, $\frac{\hat{h}_{n}^{\prime}H}{||\hat{h}_{n}||}\xrightarrow{p}\frac{(\theta_{\infty}-\theta_{\star})^{\prime}H}{||\theta_{\infty}-\theta_{\star}||}\neq0$.
By Lemma \ref{lem:T*approximation}, 
\begin{equation}
\hat{T}_{n}^{*}(\hat{h}_{n})=\inf_{h^{\prime}x=0}\left\Vert \mathbb{Z}-\Sigma^{-1/2}x\right\Vert ^{2}+o_{p}(1).\label{eq:T*=000020app}
\end{equation}
The critical value 
\begin{align*}
\hat{c} & \geq Q\left(\left.\hat{T}_{n}^{*}(\hat{h}_{n})\right|\mathbb{Z}\in\mathcal{H}_{z};\frac{1-\alpha}{1-\eta}\right)\\
 & =Q\left(\left.\inf_{h^{\prime}x=0}\left\Vert \mathbb{Z}-\Sigma^{-1/2}x\right\Vert ^{2}\right|\mathbb{Z}\in\mathcal{H}_{z};\frac{1-\alpha}{1-\eta}\right)+o_{p}(1)
\end{align*}
where the inequality follows from $\hat{h}_{n}\in\mathcal{H}$ and
the equality follows from (\ref{eq:T*=000020app}) and the continuity
of $\inf_{h^{\prime}x=0}\left\Vert \mathbb{Z}-\Sigma^{-1/2}x\right\Vert ^{2}$.
By Lemma \ref{lem:con=000020quantile}, $$Q\left(\left.\inf_{h^{\prime}x=0}\left\Vert \mathbb{Z}-\Sigma^{-1/2}x\right\Vert ^{2}\right|\mathbb{Z}\in\mathcal{H}_{z};\frac{1-\alpha}{1-\eta}\right)\in\left[Q(\chi_{1}^{2},1-\alpha),q_{\chi_{1}^{2},1-\alpha+\eta}\right).$$
By the continuity of the limit distribution of $\hat{T}_{n}$, it
holds that 
\[
\lim_{n}P_{n}\left(\hat{T}_{n}\leq\hat{c}\right)\in\left[1-\alpha,1-\alpha+\eta\right).
\]

Case 2. $\theta_{\infty}=\theta_{\star}$ and $\lim_{n}r_{n}(\theta_{n}-\theta_{\star})=h\in\mathbb{R}^{d}$.
Similar to the proof of Theorem \ref{thm:d=00003D2} Case 2, $\hat{T}_{n}(g(\theta_n))\sim T+o_{p}(1),$ $\hat{T}^*_{n}(h_n)\sim T+o_{p}(1),$
where $h_n=r_n(\theta_n-\theta_\star)$,
$$T=\inf_{x:x^{\prime}Hx=h^{\prime}Hh}\Bigl\Vert\mathbb{Z}-\Sigma^{-1/2}(x-h)\Bigr\Vert^{2}.$$ 
By Lemma \ref{lem:continus=000020T},
$T$ is continuously distributed, thus $\hat{c}\xrightarrow{p}Q(T|\mathbb{Z}\in \mathcal{H}_z,\frac{1-\alpha}{1-\eta}). $ Note that $h_{n}\in\mathcal{H}$
is equivalent to $r_{n}(\hat{\theta}_{n}-\theta_{n})\in\mathcal{H}_{z}$.
To see the coverage rate, 
\begin{align*}
P\left(\hat{T}_{n}(g(\theta_n))\leq\hat{c}\right)\geq & P\left(\hat{T}_{n}(g(\theta_n))\leq\hat{c},h_{n}\in\mathcal{H}\right)\\
= & P\left(\hat{T}_{n}(g(\theta_n))\le c,\mathbb{Z}_n\in\mathcal{H}_{z}\right)+o(1)\\
= & P\left(\left.T\le c\right|\mathbb{Z}\in\mathcal{H}_{z}\right)P\left(\mathbb{Z}\in\mathcal{H}_{z}\right)+o(1)\\
= & \frac{1-\alpha}{1-\eta}(1-\eta)+o(1).
\end{align*}

Case 3. If $\theta_{n}\rightarrow\theta_{\star}$ and $\lim_{n}r_{n}(\theta_{n}-\theta_{\star})\rightarrow\infty$.
Similar to the proof of Theorem \ref{thm:d=00003D2} Case 3, $\hat{T}_{n}(g(\theta_n))\sim T+o_{p}(1)$, $\hat{T}^*_{n}(h_n)\sim T+o_{p}(1)$,
where $T$ is defined in (\ref{eq:T2}), which is $\chi_{1}^{2}$.
The same argument as in Case 1 applies here, together with Lemma \ref{lem:con=000020quantile}, it holds that  
\[
\lim_{n}P_{n}\left(\hat{T}_{n}(g(\theta_n))\leq\hat{c}\right)\in\left[1-\alpha,1-\alpha+\eta\right).
\]
The lower bound is binding when $k=1$.
\end{proof}

\subsection{Lemmas}

\begin{lem}
\label{lem:key=000020lemma}Fix $\theta=(\theta_{1},\theta_{2})\in\mathbb{R}^{2}$,
$c=\sqrt{Q(\chi_{1}^{2},1-\alpha)}$. If the set $\bar{\mathcal{S}}$
satisfies 
\begin{enumerate}
\item \label{enu:con1}$B\left(\theta,c\right)\subset\bar{\mathcal{S}}$.
\item \label{enu:con2}For all $r>c$, $\left|\operatorname{length}\left(\partial B\left(\theta,r\right)\cap\bar{\mathcal{S}}\right)\right|\geq4r\arcsin\frac{c}{r}.$
\end{enumerate}
Let $\hat{\theta}-\theta\sim N(0,I_{2})$. Then 
\[
P\left(\hat{\theta}\in\bar{\mathcal{S}}\right)\geq1-\alpha.
\]
\end{lem}
\begin{proof}
The key idea of Lemma \ref{lem:key=000020lemma} is to compare the
coverage probability of $\bar{\mathcal{S}}$ that of an auxiliary
acceptance set
\[
\mathcal{S}_{\text{aux}}=\left\{ (x_{1},x_{2}):(x_{2}-\theta_{2})^{2}\leq c^{2}\right\} .
\]
It is trivial that $P\left(\hat{\theta}\in\mathcal{S}_{\text{aux}}\right)=1-\alpha.$
We will show that the coverage probability of $\bar{\mathcal{S}}$
is bounded below by that of $\mathcal{S}_{\text{aux}}$. 

To simplify the comparison, we switch to polar coordinates. Let $\hat{\theta}=(\theta_{1}+r\cos\omega,\theta_{2}+r\sin\omega)$,
\begin{align}
P\left(\hat{\theta}\in\mathcal{S}_{\text{aux}}\right) & =\frac{1}{2\pi}\int_{r=0}^{+\infty}\int_{\omega=-\frac{\pi}{2}}^{\frac{3}{2}\pi}\mathbf{1}\left[(r\sin\omega)^{2}\leq c^{2}\right]\text{d}\omega\exp(-\frac{r^{2}}{2})r\text{d}r\nonumber \\
 & =\int_{r=0}^{c}\exp(-\frac{r^{2}}{2})r\text{d}r+\int_{r=c}^{+\infty}\frac{4\arcsin\frac{c}{r}}{2\pi}\exp(-\frac{r^{2}}{2})r\text{d}r.\label{eq:cov=000020aux}
\end{align}
To see (\ref{eq:cov=000020aux}), note that if $r\leq c$, then $(r\sin\omega)^{2}\leq c^{2}$
 for all $\omega\in[-\frac{1}{2}\pi,\frac{3}{2}\pi]$; if $r>c$,
then
\begin{align*}
(r\sin\omega)^{2} & \leq c^{2},\omega\in[-\frac{\pi}{2},\frac{3}{2}\pi]\\
\Leftrightarrow & \omega\in\left[-\arcsin(\frac{c}{r}),\arcsin(\frac{c}{r})\right]\cup\left[\pi-\arcsin(\frac{c}{r}),\pi+\arcsin(\frac{c}{r})\right].
\end{align*}

Now consider $P\left(\hat{\theta}\in\bar{\mathcal{S}}\right)$. By
Condition \ref{enu:con1} and \ref{enu:con2}, 
\begin{align}
 & P\left(\hat{\theta}\in\bar{\mathcal{S}}\right)=\frac{1}{2\pi}\int_{r=0}^{+\infty}\frac{1}{r}\left|\operatorname{length}\left(\partial B\left(\theta,r\right)\cap\bar{\mathcal{S}}\right)\right|\exp(-\frac{r^{2}}{2})r\text{d}r\nonumber \\
 & \quad\geq\int_{r=0}^{c}\exp(-\frac{r^{2}}{2})r\text{d}r+\int_{r=c}^{+\infty}\frac{4\arcsin\frac{c}{r}}{2\pi}\exp(-\frac{r^{2}}{2})r\text{d}r.\label{eq:C12}
\end{align}
This lower bound matches the expression in (\ref{eq:cov=000020aux}),
which completes the proof. 
An illustration of Lemma \ref{lem:key=000020lemma} is provided in
Figure \ref{fig:Lemma-:}.
\begin{figure}[h]
\begin{centering}
\includegraphics[scale=0.7]{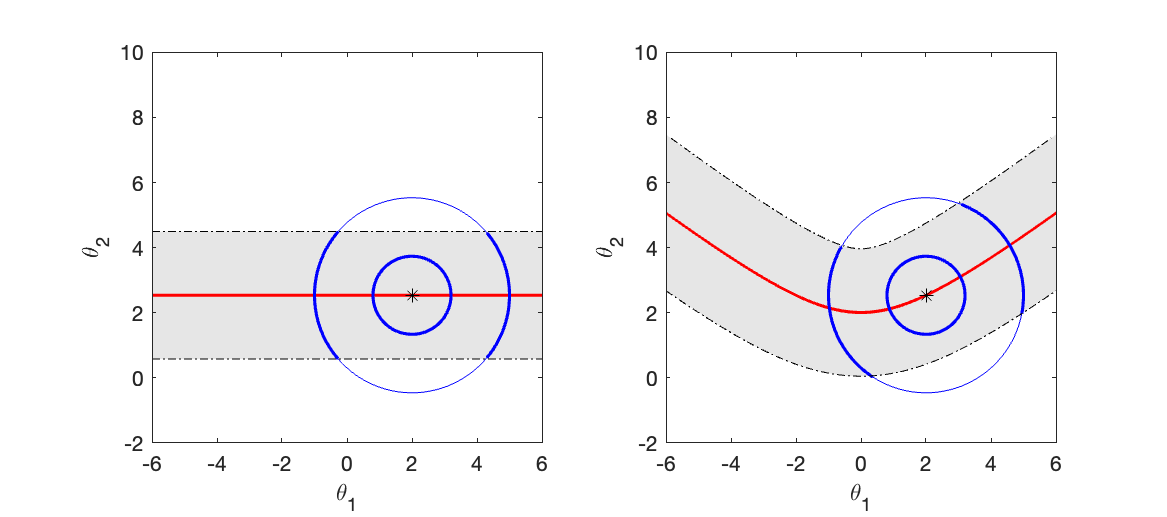}
\par\end{centering}
\caption{Lemma \ref{lem:key=000020lemma}: Acceptance Region of Linear and
Curved Null. }\label{fig:Lemma-:}

\noindent\begin{minipage}[t]{1\columnwidth}%
{\scriptsize The red curve shows the null parameter space $\mathcal{S}_{0}(\tau)$.
Shaded areas denote the acceptance regions $\mathcal{S}_{\text{aux}}$
(left) and $\bar{\mathcal{S}}$ (right). ``{*}'' represents the
true value $\theta$, and the blue circles represent $B(\theta,r)$
with bold segments indicating the portions inside the acceptance regions.
If, for all $r$, the bold segment in the right panel is longer than that
in the left, then the acceptance rate of $\bar{\mathcal{S}}$ is at
least $1-\alpha$.}%
\end{minipage}

\end{figure}
\end{proof}
\begin{lem}
\label{lem:S+}The boundary $\partial\mathcal{S}^{+}(\tau,c)$ can
be characterized by two curves. The upper one is $\mathcal{C}_{u}(c)$ defined in (\ref{Cu(c)}) and  the lower one is $\mathcal{C}_{\ell}(c)$ defined in (\ref{Cl(c)}).
\end{lem}

\begin{proof}
$\mathcal{C}_{u}(c)$ and $\mathcal{C}_{\ell}(c)$ are obtained by
shifting $\mathcal{S}_{0}^{+}(\tau)$ a distance $c$ along its normal
direction. Note that with $\tau\in(0,\frac{c^{2}(1-\rho)^{2}}{1+\rho}]$,
for all $x\in(-x_{1}^{*},x_{1}^{*})$, $\left(C_{u,1}(x_{1},c),C_{u,2}(x_{1},c)\right)$
is an interior point of $\mathcal{S}^{+}(\tau,c)$, thus not included
in $\mathcal{C}_{u}(c)$. By construction, $\mathcal{C}_{u}(c)\cup\mathcal{C}_{\ell}(c)\subseteq\mathcal{S}(\tau,c)$.

Then we show that $\mathcal{C}_{u}$ and $\mathcal{C}_{\ell}$ are
the boundaries of $\mathcal{S}^{+}(\tau,c)$. First, we show that
$\mathcal{C}_{u}$ is convex. Let $C_{u,1}^{\prime}$ and $C_{u,1}^{\prime\prime}$
be the first and second order derivatives of $C_{u,1}$ with respect
to $x$, 
\begin{align*}
 &\frac{\text{d}^{2}C_{u,2}}{\text{d}C_{u,1}^{2}} =\frac{C_{u,2}^{\prime\prime}C_{u,1}^{\prime}-C_{u,2}^{\prime}C_{u,1}^{\prime\prime}}{\left(C_{u,1}^{\prime}\right)^{3}}\\
 & =\frac{(1-\rho)\tau\left((\rho+1)\tau+2(1-\rho)x_{1}^{2}\right)^{3/2}}{\sqrt{1+\rho}\left((1-\rho)x_{1}^{2}+\tau\right)^{3/2}}\left(\left((1+\rho)\tau+2(1-\rho)x_{1}^{2}\right)^{3/2}-c\left(1-\rho^{2}\right)\tau\right)^{-1}.
\end{align*}
The sign of $\frac{\text{d}^{2}C_{u,2}}{\text{d}C_{u,1}^{2}}$ is
the same as $\left((1+\rho)\tau+2(1-\rho)x_{1}^{2}\right)^{3/2}-c\left(1-\rho^{2}\right)\tau$.
If $\tau\geq\frac{c^{2}(1-\rho)^{2}}{1+\rho}$, then 
\begin{align*}
 & \left((1+\rho)\tau+2(1-\rho)x_{1}^{2}\right)^{3/2}-c\left(1-\rho^{2}\right)\tau\geq\left((1+\rho)\tau+0\right)^{3/2}-c\left(1-\rho^{2}\right)\tau\\
 & \quad\quad=\tau(1+\rho)\left((1+\rho)^{1/2}\tau^{1/2}-c\left(1-\rho\right)\right)\\
 & \quad\quad\geq\tau(1+\rho)\left((1+\rho)^{1/2}\frac{c(1-\rho)}{\sqrt{1+\rho}}-c\left(1-\rho\right)\right)=0.
\end{align*}
If $\tau\in\left[0,\frac{c^{2}(1-\rho)^{2}}{1+\rho}\right]$, 
\begin{align*}
 & \left((1+\rho)\tau+2(1-\rho)x_{1}^{2}\right)^{3/2}-c\left(1-\rho^{2}\right)\tau\geq\left((1+\rho)\tau+2(1-\rho)x_{1}^{*2}\right)^{3/2}-c\left(1-\rho^{2}\right)\tau\\
 & \quad\quad=c^{3}(1-\rho)^{3}-c\left(1-\rho^{2}\right)\tau\geq0.
\end{align*}
Thus $\frac{\text{d}^{2}C_{u,2}}{\text{d}C_{u,1}^{2}}\geq0$ and $\mathcal{C}_{u}$
is convex.

Next, by Lemma \ref{lem:The-perpendicular-bisector}, the connecting
line of $\left(C_{u,1}(x_{1},c),C_{u,2}(x_{1},c)\right)$ and $(x_{1},X_{2}(x_{1}))$
is orthogonal to the tangent line of $\mathcal{C}_{u}$ at $\left(C_{u,1}(x_{1},c),C_{u,2}(x_{1},c)\right)$.
In addition, since $\mathcal{C}_{u}$ is convex, $\mathcal{C}_{u}$
is above its tangent line. This implies that $\forall x_{1}$
\begin{align}
d\left((x_{1},X_{2}(x_{1})),\mathcal{C}_{u}\right) & =d\left(\left(x_{1},X_{2}(x_{1})\right),\left(C_{u,1}(x_{1},c),C_{u,2}(x_{1},c)\right)\right)=c.\label{eq:lower=000020bd}
\end{align}
This also implies that for all $(C_{u,1}(x_{1},c),C_{u,2}(x_{1},c))\in\mathcal{C}_{u}$,
\[
c\geq d\left((C_{u,1}(x_{1},c),C_{u,2}(x_{1},c)),\mathcal{S}_{0}^{+}(\tau)\right)\geq c,
\]
where the first equality follows from the fact that 
\[
d\left((x_{1},X_{2}(x_{1})),(C_{u,1}(x_{1},c),C_{u,2}(x_{1},c))\right)=c
\]
and the second inequality follows from (\ref{eq:lower=000020bd}), i.e.
$$d\left(\left(C_{u,1}(x_{1},c),C_{u,2}(x_{1},c)\right),\mathcal{S}_{0}^{+}(\tau)\right)\geq d\left(\mathcal{C}_{u},\mathcal{S}_{0}^{+}(\tau)\right)=\inf_{x_{1}}d\left(\mathcal{C}_{u},(x_{1},X_{2}(x_{1}))\right)=c.$$
Therefore, $\mathcal{C}_{u}$ is the upper part of the boundary of $\mathcal{S}^{+}(\tau,c)$. 

To show that $\mathcal{C}_{\ell}$ is on the boundary of $\mathcal{S}_{0}^{+}(\tau)$,
note that by Lemma \ref{lem:The-perpendicular-bisector}, the connecting
line of $\left(C_{\ell,1}(x_{1},c),C_{\ell,2}(x_{1},c)\right)$ and
$(x_{1},X_{2}(x_{1}))$ is orthogonal to the tangent line of $\mathcal{S}_{0}^{+}(\tau)$
at $\left(x_{1},X_{2}(x_{1})\right)$. In addition, since 
\[
\frac{\text{d}^{2}X_{2}(x_{1})}{\text{d}x_{1}^{2}}=\frac{(1-\rho)\tau}{\sqrt{1+\rho}\left(\tau+(1-\rho)x_{1}^{2}\right)^{3/2}}\geq0,
\]
$\mathcal{S}_{0}^{+}(\tau)$ is convex. Thus $\mathcal{S}_{0}^{+}(\tau)$
is above its tangent line. This implies that 
\[
d\left(\left(C_{\ell,1}(x_{1},c),C_{\ell,2}(x_{1},c)\right),\mathcal{S}_{0}^{+}(\tau)\right)=c.
\]
Thus $\mathcal{C}_{\ell}$ is the lower part of the boundary of $\mathcal{S}^{+}(\tau,c)$.
\end{proof}
\begin{lem}
\label{lem:The-perpendicular-bisector}Let $x_{1}\in\mathbb{R}\backslash(-x_{1}^{*},x_{1}^{*})$. 
\begin{enumerate}
\item \label{perp1}The perpendicular bisector of $\left(C_{\ell,1}(x_{1},c),C_{\ell,2}(x_{1},c)\right)$
and $\left(C_{u,1}(x_{1},c),C_{u,2}(x_{1},c)\right)$ is tangent to
$\mathcal{S}_{0}^{+}(\tau)$ at $(x_{1},X_{2}(x_{1}))$.
\item \label{enu:The-connecting-line}The connecting line of $\left(C_{u,1}(x_{1},c),C_{u,2}(x_{1},c)\right)$
and $\left(C_{\ell,1}(x_{1},c),C_{\ell,2}(x_{1},c)\right)$ is orthogonal
to the tangent line of $\mathcal{C}_{u}$ and $\mathcal{C}_{\ell}$
at these points.
\end{enumerate}
\end{lem}

\begin{proof}
To show \ref{perp1}. It is easy to verify that 
\begin{align*}
\frac{1}{2}\left(C_{\ell,1}(x_{1},c)+C_{u,1}(x_{1},c)\right) & =x_{1},\\
\frac{1}{2}\left(C_{\ell,2}(x_{1},c)+C_{u,2}(x_{1},c)\right) & =X_{2}(x_{1}).
\end{align*}
In addition, 
\begin{align*}
\frac{C_{u,2}(x_{1},c)-C_{\ell,2}(x_{1},c)}{C_{u,1}(x_{1},c)-C_{\ell,1}(x_{1},c)} & =-\frac{(1+\rho)X_{2}(x_{1})}{(1-\rho)x_{1}},\quad\frac{\text{d}X_{2}(x_{1})}{\text{d}x_{1}}=\frac{(1-\rho)x_{1}}{(1+\rho)X_{2}(x_{1})}\\
\Rightarrow & \frac{C_{u,2}(x_{1},c)-C_{\ell,2}(x_{1},c)}{C_{u,1}(x_{1},c)-C_{\ell,1}(x_{1},c)}\frac{\text{d}X_{2}(x_{1})}{\text{d}x_{1}}=-1.
\end{align*}
This completes the proof.

To show \ref{enu:The-connecting-line}, straightforward calculation shows that 
\begin{align*}
\frac{\text{d}C_{u,2}}{\text{d}C_{u,1}} & =\frac{C_{u,2}^{\prime}}{C_{u,1}^{\prime}}=\frac{\text{d}X_{2}(x_{1})}{\text{d}x_{1}},\\
\frac{\text{d}C_{\ell,2}}{\text{d}C_{\ell,1}} & =\frac{C_{\ell,2}^{\prime}}{C_{\ell,1}^{\prime}}=\frac{\text{d}X_{2}(x_{1})}{\text{d}x_{1}}.
\end{align*}
By the first part, the connecting line of $\left(C_{u,1}(x_{1},c),C_{u,2}(x_{1},c)\right)$
and $\left(C_{\ell,1}(x_{1},c),C_{\ell,2}(x_{1},c)\right)$ is orthogonal
to the tangent line of $\mathcal{S}_{0}^{+}(\tau)$ at $(x_{1},X_{2}(x_{1}))$.
Therefore, it is also orthogonal to the tangent line of $\mathcal{C}_{i}(c)$
at $\left(C_{i,1}(x_{1},c),C_{i,2}(x_{1},c)\right)$ with $i=u,\ell$. 
\end{proof}
\begin{lem}
\label{lem:I=000020and=000020J}Let $r>c$, $\theta=(\theta_{1},\theta_{2})\in\mathcal{S}_{0}^{+}(\tau).$ 
\begin{enumerate}
\item \label{enu:C1=000020tau<-2}If $\theta_{1}\not\in(-x_{1}^{*},x_{1}^{*})$,
there exists $(x_{1},x_{2}^{(1)})\in\partial B\left((\theta_{1},\theta_{2}),r\right)$,
$(x_{1},x_{2}^{(2)})\in\mathcal{C}_{u}(\tau,c)$ such that $x_{2}^{(1)}>x_{2}^{(2)}$.
\item \label{enu:C2-1-1-1}There exists $(x_{1},x_{2}^{(1)})\in\partial B\left((\theta_{1},\theta_{2}),r\right)$,
$(x_{1},x_{2}^{(2)})\in\mathcal{C}_{\ell}(\tau,c)$ such that $x_{2}^{(1)}<x_{2}^{(2)}$.
\item \label{lem:min=000020dist=000020H} If $\tau\leq\frac{c^{2}(1-\rho)^{2}}{1+\rho}$,
$\theta_{1}\in(-x_{1}^{*},x_{1}^{*})$, then 
\[
d\left(\theta,\mathcal{C}_{u}(c)\right)=d(\theta,K),\;\text{with }K\;\text{in }(\ref{eq:H}).
\]
\item \label{enu:B=000020cap=000020Cl=000020}Suppose $\rho\geq0$, $\tau\leq\frac{c^{2}(1-\rho)^{2}}{1+\rho}$,
$\theta_{1}\in(-x_{1}^{*},x_{1}^{*})$, and $r\in\left(c,\bar{r}(\theta_{1})\right)$.
If $$\left(C_{\ell,1}(x_{1},c),C_{\ell,2}(x_{1},c)\right)\in\partial B\left(\theta,r\right),$$
then $C_{\ell,2}(x_{1},c)<0$. 
\end{enumerate}
\end{lem}

\begin{proof}
Prove \ref{enu:C1=000020tau<-2}. Let $x_{1}=C_{u,1}(\theta_{1},c)$
and $x_{2}^{(2)}=C_{u,2}(\theta_{1},c)$. By construction, $(x_{1},x_{2}^{(2)})\in\mathcal{C}_{u}(\tau,c)$
and $d\left(\theta,(x_{1},x_{2}^{(2)})\right)=c<r$. In addition,
there exists $x_{2}$ large enough such that $d\left(\theta,(x_{1},x_{2})\right)>r$.
By continuity, there exists $x_{2}^{(1)}\in(x_{2}^{(2)},x_{2})$ such
that $d\left(\theta,(x_{1},x_{2}^{(2)})\right)=r$.

The proof of \ref{enu:C2-1-1-1} is an analog of \ref{enu:C1=000020tau<-2}.

To prove \ref{lem:min=000020dist=000020H}, let $x_{1}\not\in(-x_{1}^{*},x_{1}^{*})$.
The distance between $\theta$ and $\left(C_{u,1}(x_{1},c),C_{u,2}(x_{1},c)\right)\in\mathcal{C}_{u}(c)$
is
\begin{align*}
h(x_{1})= & \left(C_{u,1}(x_{1},c)-\theta_{1}\right)^{2}+\left(C_{u,2}(x_{1},c)-X_{2}(\theta_{1})\right)^{2}.
\end{align*}
Taking the first order derivative 
\begin{align*}
\frac{\mathrm d h(x_{1})}{\mathrm d x_{1}}
&=\scalebox{1}{$
(x_{1}-\theta_{1})\left(\left((1+\rho)\tau+2(1-\rho)x_{1}^{2}\right)^{3/2}
-c(1-\rho^{2})\tau\right)$} \\
&\quad\times\scalebox{1}{$
\frac{2\left((1-\rho)^{2}x_{1}(x_{1}+\theta_{1})+(1-\rho^{2})x_{1}^{2}
+(1+\rho)\tau+(1+\rho)\sqrt{(1-\rho)x_{1}^{2}+\tau}
\sqrt{\theta_{1}^{2}(1-\rho)+\tau}\right)}
{(1+\rho)\sqrt{\tau+(1-\rho)x_{1}^{2}}
\left((1+\rho)\tau+2(1-\rho)x_{1}^{2}\right)^{3/2}
\left(\sqrt{\tau+(1-\rho)x_{1}^{2}}
+\sqrt{\theta_{1}^{2}(1-\rho)+\tau}\right)}
$}.
\end{align*}

Note that (i) $|x_1|\geq x_{1}^{*}\geq|\theta_{1}|$ thus $x_1(x_1+\theta_{1})\geq0$;
(ii) $\tau\leq\frac{c^{2}(1-\rho)^{2}}{1+\rho}$ thus
\[
\left((1+\rho)\tau+2(1-\rho)x_{1}^{2}\right)^{3/2}-c(1-\rho^{2})\tau\geq\left((1+\rho)\tau+2(1-\rho)x_{1}^{*2}\right)^{3/2}-c(1-\rho^{2})\tau=0.
\]
Therefore, the sign of $\frac{\text{d}h(x_{1})}{\text{d}x_{1}}$ is
the same as $(x_{1}-\theta_{1})$, and thus $\frac{\text{d}h(x)}{\text{d}x}<0$
for $x<-x_{1}^{*}<0$, and $\frac{\text{d}h(x)}{\text{d}x}>0$ for
$x>x_{1}^{*}>0$. Hence, $h(x)$ is minimized at $x_{1}^{*}$, i.e.
point $H$.

To prove \ref{enu:B=000020cap=000020Cl=000020}, by contradiction,
assume that there exists $A=(C_{\ell,1}(x_{1}),C_{\ell,2}(x_{1}))\in\partial B\left((\theta_{1},\theta_{2}),r\right)$
and $C_{\ell,2}(x_{1})\geq0$. WLOG, assume that $C_{\ell,1}(x_{1})\geq0$.
Since $C_{\ell,2}(x_{1})$ is increasing in $x_{1}$, and $C_{\ell,2}\left(\frac{\sqrt{c^{2}(\rho+1)^{2}-(\rho+1)\tau}}{\sqrt{2}\sqrt{1-\rho}}\right)=0$,
we have 
\[
x_{1}>\frac{\sqrt{c^{2}(1+\rho)^{2}-(1+\rho)\tau}}{\sqrt{2}\sqrt{1-\rho}}\geq x_{1}^{*}.
\]
Let $A^{\prime}=(C_{u,1}(x_{1}),C_{u,2}(x_{1}))$. By Lemma \ref{lem:The-perpendicular-bisector},
the perpendicular bisector of $AA^{\prime}$ is tangent of $\mathcal{S}_{0}^{+}(\tau)$
at $(x_{1},X_{2}(x_{1}))$. Since $\mathcal{S}_{0}^{+}(\tau)$ is
convex, $\theta$ is above the perpendicular bisector. This further
implies that 
\[
r=d\left(\theta,A\right)>d\left(\theta,A^{\prime}\right).
\]
However, by Lemma \ref{lem:I=000020and=000020J}.\ref{lem:min=000020dist=000020H},
we have 
\[
d\left(\theta,A^{\prime}\right)\geq d\left(\theta,K\right)=\bar{r}(\theta_{1}),
\]
which is a contradiction. Therefore, such $x_{1}$ does not exist.
\end{proof}
\begin{lem}
\label{lem:S0^=000020cap=000020B}Let $r>0$. $\mathcal{S}_{0}^{+}(\tau)$
intersects $\partial B\left((\theta_{1},X_{2}(\theta_{1})),r\right)$
at a minimum of two points.
\end{lem}

\begin{proof}
Let $h(x_{1})$ be the distance between $(x_{1},X_{2}(x_{1}))\in\mathcal{S}_{0}^{+}(\tau)$
and the center of the circle $(\theta_{1},\theta_{2})=(\theta_{1},X_{2}(\theta_{1}))$,
i.e.
\begin{align*}
h(x_{1}) & =\left(x_{1}-\theta_{1}\right)^{2}+\left(X_{2}(x_{1})-X_{2}(\theta_{1})\right)^{2}\\
 & =\frac{\left(\sqrt{\tau+(1-\rho)x_{1}^{2}}-\sqrt{\tau+(1-\rho)\theta_{1}^{2}}\right)^{2}}{1+\rho}+(x_{1}-\theta_{1})^{2}.
\end{align*}
It is easy to see that $h(\theta_{1})=0$, $h(-\infty)=\infty$ and
$h(+\infty)=\infty$. Since $h(x_{1})$ is a continuous function,
there exists $x_{1}^{(1)}<\theta_{1}<x_{1}^{(2)}$ such that 
\[
h(x_{1}^{(1)})=h(x_{1}^{(2)})=r^{2}.
\]
Thus $\mathcal{S}_{0}^{+}(\tau)$ intersects $\partial B\left((\theta_{1},\theta_{2}),r\right)$
at $\left(x_{1}^{(1)},X_{2}(x_{1}^{(1)})\right)$ and $\left(x_{1}^{(2)},X_{2}(x_{1}^{(2)})\right)$.
\end{proof}
\begin{lem}
\label{lem:=000020x<xbar}Let $\tau\leq\frac{(1-\rho)^{2}c^{2}}{1+\rho}$,
with $\rho\geq0$. Suppose $|x_{1}|<x_{1}^{*}$ and let $O=(x_{1},X_{2}(x_{1}))\in\mathcal{S}_{0}(\tau)$.
Define $\mathcal{C}_{j}^{-}=\left\{ (x_{1},x_{2}):(x_{1},-x_{2})\in\mathcal{C}_{j}(c)\right\} $
where $j=\ell,u$, $\mathcal{C}_{j}$ is defined as in Lemma \ref{lem:S+}.
For all $r<\bar{r}(x_{1})$, let $(x_{1}^{(1)},0),(x_{1}^{(2)},0)\in\partial B(O,r)$.
 Then
\end{lem}

\begin{enumerate}
\item \label{enu:x-axis}$\max\left\{ |x_{1}^{(1)}|,|x_{1}^{(2)}|\right\} \leq\bar{x}_{1}:=\frac{\sqrt{2}\sqrt{c^{2}(\rho+1)-\tau}}{\sqrt{1-\rho^{2}}}$
.
\item \label{enu:covered=000020} If $(x_{1},x_{2})\in\partial B(O,r)\cap\mathcal{C}_{\ell}^{-}$,
then $x_{2}\geq0$. Moreover, $\partial B(O,r)\cap\mathcal{C}_{u}^{-}=\emptyset$.
\end{enumerate}
\begin{proof}
Proof of Part \ref{enu:x-axis}. WLOG, assume $x_{1}^{(1)}\geq0$. Let $k=\frac{\sqrt{2}\sqrt{c^{2}(1-\rho)+\tau}}{\sqrt{1-\rho^{2}}}$
be the vertical axis of $K$. Let
\begin{align*}
h(x_{1}) & =\left\Vert O-(\bar{x}_{1},0)\right\Vert ^{2}-\left\Vert O-K\right\Vert ^{2}\\
 & =(\bar{x}_{1}-x_{1})^{2}+X_{2}(x_{1})^{2}-x_{1}^{2}-\left(k-X_{2}(x_{1})\right)^{2}.
\end{align*}
It suffices to show that $h(x_1^{(1)})\geq0$. Note that 
\begin{align*}
\frac{\partial h(x_{1})}{\partial x_{1}}
&=\scalebox{1}{$
\frac{2\sqrt{2}\left((1+\rho)^2(1-\rho)\left(c^{2}(1-\rho)+\tau\right)
\left(\tau+(1-\rho)x_{1}^{2}\right)\right)^{-1/2}}
{\left(c^{2}(1-\rho)+\tau\right)x_{1}(1-\rho)
+\sqrt{(1+\rho)\left(c^{4}(1-\rho^{2})+2c^{2}\rho\tau-\tau^{2}\right)}
\sqrt{\tau+(1-\rho)x_{1}^{2}}}
\,\tilde{h}(x_{1})
$} \\
\tilde{h}(x_{1})
&=\scalebox{0.9}{$
2(1-\rho)\left(c^{2}\tau(1-3\rho)+\tau^{2}-2c^{4}(1-\rho)\rho\right)x_{1}^{2}
-(1+\rho)\left(c^{4}(1-\rho^{2})+2c^{2}\rho\tau-\tau^{2}\right)\tau
$}.
\end{align*}

The sign of $\frac{\partial h(x_{1})}{\partial x_{1}}$ depends on
$\tilde{h}(x_{1})$. Observe that
\begin{align*}
\tilde{h}(0) & =-(1+\rho)\left(c^{4}-(\tau-\rho c^{2})^{2}\right)\tau\leq0
\end{align*}
where the inequality follows from $\tau\leq\frac{(1-\rho)^{2}c^{2}}{1+\rho}\leq(1+\rho)c^{2}$.
Moreover
\[
\tilde{h}(x_{1}^{*})=-2c^{2}(1-\rho)\rho\left(\tau-c^{2}(\rho-1)\right)^{2}\leq0.
\]
Since $\tilde{h}(x_{1})$ is monotone in $x_{1}$ for $x_{1}>0$,
we have
\begin{align*}
\tilde{h}(x_{1}) & \leq\max\left\{ \tilde{h}(0),\;\tilde{h}(x_{1}^{*})\right\} \leq0.
\end{align*}
Thus $h(x_{1})$ decreases in $x_{1}$, and $h(x_{1}^{(1)})\geq0$ follows from
\begin{align*}
h&(x_{1})  \geq h(x_{1}^{*})\\
 & =\frac{8c^{4}\rho^{2}(1-\rho^{2})^{-1}}{c^{2}(1+\rho^{2})-(1+\rho)\tau+\sqrt{(\rho+1)\left(c^{2}(1+\rho)-\tau\right)}\sqrt{c^{2}(1-\rho)^{2}-(\rho+1)\tau}}\geq0.
\end{align*}

Proof of Part \ref{enu:covered=000020}. We can verify that $(\bar{x}_{1},0)\in\mathcal{C}_{\ell}^{-}$.
Moreover, for all $(x_{1},x_{2})\in\mathcal{C}_{\ell}^{-}$, if $x_{2}<0$,
then $|x_{1}|>\bar{x}_{1}$. By Part \ref{enu:x-axis}, such points
cannot lie on $\partial B(O,r)$. Finally, $\partial B(O,r)\cap\mathcal{C}_{u}^{-}$
because
\[
d(O,\mathcal{C}_{u}^{-})>d(O,\mathcal{C}_{u})=d(O,K)=\bar{r}(x_{1})>r.
\]
The inequality from the symmetry of $\mathcal{C}_{u}$ and $\mathcal{C}_{u}^{-}$
about the $x_{1}$-axis, combined with $X_{2}(x_{1})>0$. The equality
follows from Lemma \ref{lem:I=000020and=000020J}.\ref{lem:min=000020dist=000020H}.
\end{proof}
\begin{lem}
\label{lem:Tn=000020=00003D=000020T+o(1)}Suppose $\vartheta_{\star,n}\rightarrow\vartheta_\star$, and $\tilde{g}_n$ satisfies (i) $\frac{\partial \tilde{g}_n (\vartheta_{\star,n})}{\partial\vartheta}=0$, (ii) $\frac{\partial^2\tilde{g}_n(\vartheta_{\star,n})}{\partial\vartheta\partial\vartheta^\prime}\rightarrow H$ with a full rank $H$, (iii)  $\frac{\partial^2\tilde{g}_n(\vartheta)}{\partial\vartheta\partial\vartheta^\prime}$ is Lipschitz continuous in $\vartheta$ with Lipschitz coefficient $M\in\mathbb R_+$   for all $\vartheta\in B(\vartheta_\star,\epsilon)$ where $\epsilon>0$. Let $\mathbb{Z}\sim N(0,I_{d})$ and $\mathbb{Z}_{n}=\mathbb{Z}+o_{p}(1)$. Let $h_{n}=r_{n}(\vartheta_{n}-\vartheta_{\star,n})$.
If $\lim_{n}h_n=h\in\mathbb{R}^{d}$,
then 
\[
T_{n}^{(1)}=T^{(1)}+o_{p}(1),
\]
where 
\begin{align*}
T_{n}^{(1)} & =\inf_{\tilde{g}_n(\vartheta_{\star,n}+r_{n}^{-1}x)=\tilde{g}_n(\vartheta_{n})}\left\Vert \mathbb{Z}_{n}-(x-h_{n})\right\Vert ^{2},\quad T^{(1)}=\inf_{x^{\prime}Hx=h^{\prime}Hh}\left\Vert \mathbb{Z}-(x-h)\right\Vert ^{2}.
\end{align*}
\end{lem}

\begin{proof}
Let $x_{n}^{*},x^{*}\in\mathbb{R}^{k}$ be minimizers of $T_{n}^{(1)}$
and $T^{(1)}$, respectively. Standard quadratic arguments imply $x_{n}^{*},x^{*}=O_{p}(1)$. 

Step 1. Prove $T^{(1)}\leq T_{n}^{(1)}+o_{p}(1)$. By feasibility
of $x_{n}^{*}$, 
\begin{align*}
\tilde{g}_n(\vartheta_{\star,n}+r_{n}^{-1}x_{n}^{*})
&=\scalebox{1}{$
\tilde{g}_n(\vartheta_{n})
=\tilde{g}_n(\vartheta_{\star,n}+r_{n}^{-1}h_{n})
$} \\
\scalebox{1}{$
\tilde{g}_n(\vartheta_{\star,n})
+r_{n}^{-1}\frac{\partial \tilde{g}_n(\vartheta_{\star,n})}
{\partial\vartheta^{\prime}}x_{n}^{*}
+r_{n}^{-2}x_{n}^{*\prime}
\frac{\partial^{2}\tilde{g}_n(\bar{\vartheta})}
{\partial\vartheta\partial\vartheta^{\prime}}x_{n}^{*}
$}
&=\scalebox{1}{$
\tilde{g}_n(\vartheta_{\star,n})
+r_{n}^{-1}\frac{\partial \tilde{g}_n(\vartheta_{\star,n})}
{\partial\vartheta^{\prime}}h_{n}
+r_{n}^{-2}h_{n}^{\prime}
\frac{\partial^{2}\tilde{g}_n(\bar{\bar{\vartheta}})}
{\partial\vartheta\partial\vartheta^{\prime}}h_{n}
$} \\
\Rightarrow\;
\scalebox{1}{$
x_{n}^{*\prime}
\frac{\partial^{2}\tilde{g}_n(\bar{\vartheta})}
{\partial\vartheta\partial\vartheta^{\prime}}x_{n}^{*}
=h_{n}^{\prime}
\frac{\partial^{2}\tilde{g}_n(\bar{\bar{\vartheta}})}
{\partial\vartheta\partial\vartheta^{\prime}}h_{n}
$}
\end{align*}
where $\bar{\vartheta}$ is between $\vartheta_{\star,n}$ and $\vartheta_{\star,n}+r_{n}^{-1}x_{n}^{*}$,
and $\bar{\bar{\vartheta}}$ is between $\vartheta_{\star,n}$ and $\vartheta_{\star,n}+r_{n}^{-1}h_{n}$.
By Assumption \ref{assu:uniform=000020normal}.\ref{enu:-is-twice=000020diff},
\begin{align}
x_{n}^{*\prime}\left[H+o_{p}(1)\right]x_{n}^{*} & =\left(h+o(1)\right)^{\prime}\left[H+o_{p}(1)\right]\left(h+o(1)\right)\nonumber \\
\Rightarrow x_{n}^{*\prime}Hx_{n}^{*} & =h^{\prime}Hh+o_{p}(1).\label{eq:x*Hx*=00003DhHh}
\end{align}
Case 1. $h^{\prime}Hh=0$ and $H$ is positive/ negative definite.
(\ref{eq:x*Hx*=00003DhHh}) implies $x_{n}^{*}=o_{p}(1)$. Hence
\[
T^{(1)}\leq\left\Vert \mathbb{Z}+h\right\Vert ^{2}=T_{n}^{(1)}+o_{p}(1).
\]
For Case 2 and Case 3 below, we construct $\tilde{x}_{n}^{*}=x_{n}^{*}+o_{p}(1)$
such that $\tilde{x}_{n}^{*\prime}H\tilde{x}_{n}^{*}=h^{\prime}Hh$.
It then follows that 
\[
T^{(1)}\leq\left\Vert \mathbb{Z}-(\tilde{x}_{n}^{*}-h)\right\Vert ^{2}=\left\Vert \mathbb{Z}_{n}-(x_{n}^{*}-h)\right\Vert ^{2}+o_{p}(1)=T_{n}^{(1)}+o_p(1).
\]
Case 2. $h^{\prime}Hh=0$ and $H$ is indefinite. By Lemma \ref{lem:y=000020H=000020x},
there is $y_{n}$ such that 
\[
x_{n}^{*\prime}Hy_{n}=0,\quad y_{n}^{\prime}Hy_{n}=-\text{sign}(x_{n}^{*\prime}Hx_{n}^{*}).
\]
Let $\tilde{x}_{n}^{*}=x_{n}^{*}+\xi_{n}y_{n}$, where $\xi_{n}=\sqrt{|x_{n}^{*\prime}Hx_{n}^{*}|}$. By (\ref{eq:x*Hx*=00003DhHh}), $\xi_n=o_{p}(1).$
In addition,
\[
\tilde{x}_{n}^{*\prime}H\tilde{x}_{n}^{*}=\left(x_{n}^{*}+\xi_{n}y_{n}\right)^{\prime}H\left(x_{n}^{*}+\xi_{n}y_{n}\right)=x_{n}^{*\prime}Hx_{n}^{*}-\text{sign}(x_{n}^{*\prime}Hx_{n}^{*})\xi_{n}^{2}=0.
\]
Case 3. $h^{\prime}Hh\neq0$. Let $(1+\xi_{n})^{2}=\frac{h^{\prime}Hh}{x_{n}^{*\prime}Hx_{n}^{*}}.$
By (\ref{eq:x*Hx*=00003DhHh}), $\xi_{n}=o_{p}(1)$. Let $\tilde{x}_{n}^{*}=x_{n}^{*}+\xi_{n}x_{n}^{*}$.
By construction, $\tilde{x}_{n}^{*\prime}H\tilde{x}_{n}^{*}=h^{\prime}Hh$. 

Step 2. Prove $T_{n}^{(1)}\leq T^{(1)}+o_{p}(1)$. Case 1. $h^{\prime}Hh=0$ and
$H$ is positive or negative definite. Here $x^{*}=0$ and $T=\mathbb{Z}^{\prime}\mathbb{Z}$.
Thus,
\[
T_{n}^{(1)}\leq\left\Vert \mathbb{Z}_{n}\right\Vert ^{2}+o_p(1)=T^{(1)}+o_{p}(1).
\]
For Case 2 and Case 3 below, we show that there exists $\eta_{n}=o_{p}(1)$
such that 
\[
\tilde{g}_n(\vartheta_{\star,n}+r_{n}^{-1}(x^{*}+\eta_{n})=\tilde{g}_n(\vartheta_{n}).
\]
Then the conclusion follows from 
\[
T_{n}^{(1)}\leq\left\Vert \mathbb{Z}_{n}-(x^{*}+\eta_{n}-h_{n})\right\Vert ^{2}=\left\Vert \mathbb{Z}-(x^{*}-h)\right\Vert ^{2}+o_{p}(1)=T^{(1)}+o_{p}(1).
\]
Case 2. $h^{\prime}Hh=0$ and $H$ is indefinite. Assume $H=\text{diag}\left\{ \lambda_{1},...,\lambda_{m},-\lambda_{m+1},...,-\lambda_{d}\right\} $
with $\lambda_{1},...,\lambda_{d}>0$. If $H$ is not diagonal, write
$H=P^{\prime}\Lambda P$ with diagonal $\Lambda$ and transform $x^{*}$
by $Px^{*}$. Define 
\[
y^{*}=-\text{sign}\left(\tilde{g}_n(\vartheta_{\star,n}+r_{n}^{-1}x^{*})-\tilde{g}_n(\vartheta_{n})\right)(x_{1}^{*},...,x_{m}^{*},-x_{m+1}^{*},...,-x_{d}^{*}).
\]
Then
\[
y^{*\prime}Hy^{*}=x^{*\prime}Hx^{*}=0,\;y^{*\prime}Hx^{*}=-\text{sign}\left(\tilde{g}_n(\vartheta_{\star,n}+r_{n}^{-1}x^{*})-\tilde{g}_n(\vartheta_{n})\right)\sum_{i=1}^{d}\lambda_{i}x_{i}^{*2}.
\]
Let 
\[
u_{n}(\xi)=r_{n}^{2}\left(\tilde{g}_n\left(\vartheta_{\star,n}+r_{n}^{-1}(x^{*}+\xi y^{*})\right)-\tilde{g}_n(\vartheta_{n})\right).
\]
Define 
\[
\xi_{n}=\underset{\xi}{\arg\min}|\xi|\text{ s.t. }u_{n}(\xi)=0.
\]
We show $\xi_{n}=o_{p}(1)$. Note that 
\[
u_{n}(0)=r_{n}^{2}\left(\tilde{g}_n\left(\vartheta_{\star,n}+r_{n}^{-1}x^{*}\right)-\tilde{g}_n(\vartheta_{n})\right),
\]
and for all $\epsilon>0$,
\begin{align*}
u_{n}(\epsilon) & =2\epsilon y^{*\prime}Hx^{*}+o_{p}(1)=-2\epsilon\text{sign}\left(u_{n}(0)\right)\sum_{i=1}^{d}\lambda_{i}x_{i}^{*2}+o_{p}(1).
\end{align*}
By Lemma \ref{lem:lambda=000020x=000020epsilin}, there is $N$ such
that for all $n\geq N$, 
\[
P(|\xi_{n}|\leq\epsilon)\geq P\left(u_{n}(\epsilon)u_{n}(0)\leq0\right)\geq1-\epsilon,
\]
which implies that $\xi_{n}=o_{p}(1)$. Let $\eta_{n}=\xi_{n}y^{*}$,
the conclusion follows. 

Case 3. $h^{\prime}Hh\neq0$. Define 
\[
\xi_{n}=\underset{\xi}{\arg\min|\xi|}\text{ s.t. }r_{n}^{2}\left(\tilde{g}_n(\vartheta_{\star,n}+(1+\xi)r_{n}^{-1}x^{*})-\tilde{g}_n(\vartheta_{n})\right)=0,
\]
and $\xi_{n}=\infty$ if  no solution exists. We show $\xi_{n}=o_{p}(1)$,
i.e. for all $\epsilon>0$, there is $N$ such that for all $n>N$
$P\left(|\xi_{n}|>\epsilon\right)<\epsilon.$
To see this, let 
\[
u_{n}\left(\xi\right)=r_{n}^{2}\left(\tilde{g}_n(\vartheta_{\star,n}+r_{n}^{-1}(1+\xi)x^{*})-\tilde{g}_n(\vartheta_{n})\right).
\]
Then 
\[
u_{n}(\epsilon)=\left((1+\epsilon)^{2}-1\right)h^{\prime}Hh+o_{p}(1),\;u_{n}(-\epsilon)=\left((1-\epsilon)^{2}-1\right)h^{\prime}Hh+o_{p}(1).
\]
Therefore, there is $N$ such that for all $n>N$, 
\[
P\left(u_{n}(\epsilon)u_{n}(-\epsilon)<0\right)\geq1-\epsilon.
\]
By the continuity of $u_{n}$, $\left\{ u_{n}(\epsilon)>0\text{ and }u_{n}(-\epsilon)<0\right\} $
implies $|\xi_{n}|\leq\epsilon$. Therefore, for all $n>N$,
\[
P(|\xi_{n}|\leq\epsilon)\geq P\left(u_{n}(\epsilon)>0\text{ and }u_{n}(-\epsilon)<0\right)\geq1-\epsilon.
\]
The conclusion follows from $\eta_{n}= \xi_{n}x^{*}.$
\end{proof}
\begin{lem}
\label{lem:Tn=000020=00003D=000020T+o1=0000202}Suppose Assumption
\ref{assm:differentiability} holds. Let $\mathbb{Z}\sim N(0,I_{d})$ and
$\mathbb{Z}_{n}=\mathbb{Z}+o_{p}(1)$. Suppose $\Sigma_{n}=\Sigma+o_{p}(1)$
with $\Sigma\in\mathcal{S}$, where $\mathcal{S}$ is defined in Assumption
\ref{assu:uniform=000020normal}.\ref{enu:compact=000020Sigma}. If
$\lim s_{n}(\theta_{n}-\theta_{\star})=h\in\mathbb{R}^{d}$ with $h^{\prime}H\neq0$,
for some sequence $s_{n}\rightarrow\infty$ with $s_{n}/r_{n}\rightarrow0$,
then 
\[
T_{n}^{(2)}=T^{(2)}+o_{p}(1)
\]
where
\begin{align}
T_{n}^{(2)} & =\inf_{g(\theta_{n}+r_{n}^{-1}x)=g(\theta_{n})}\left\Vert \mathbb{Z}_{n}-\Sigma_{n}^{-1/2}x\right\Vert ^{2},\quad T^{(2)}=\inf_{h^{\prime}Hx=0}\left\Vert \mathbb{Z}-\Sigma^{-1/2}x\right\Vert ^{2}.\label{eq:T^(2)_n}
\end{align}
\end{lem}
\begin{proof}
Let $x_{n}^{*}$ and $x^{*}$ denote the optimizers of $T_{n}^{(2)}$
and $T^{(2)}$, respectively. It is easy to verify that $x_{n}^{*},x^{*}=O_{p}(1)$. 

Step 1. Prove $T^{(2)}\leq T_{n}^{(2)}+o_{p}(1)$. Let $\theta_{n}^{*}=\theta_{n}+r_{n}^{-1}x_{n}^{*}$.
By the feasibility constraint,
\begin{align*}
0= & s_{n}r_{n}\left(g(\theta_{n}^{*})-g(\theta_{n})\right).
\end{align*}
Expanding $g(\theta_{n}^{*})$ around $\theta_{n}$ using a second-order
Taylor expansion, and linearizing $\nabla g(\theta_{n})$ around $\theta_{\star}$,
we obtain 
\begin{align}
0= & s_{n}r_{n}\left(\frac{\partial g(\theta_{n})}{\partial\theta^{\prime}}(\theta_{n}^{*}-\theta_{n})+(\theta_{n}^{*}-\theta_{n})^{\prime}\frac{\partial g(\bar{\theta})}{\partial\theta\partial\theta^{\prime}}(\theta_{n}^{*}-\theta_{n})\right)\nonumber \\
= & s_{n}(\theta_{n}-\theta_{\star})\frac{\partial g(\bar{\bar{\theta}})}{\partial\theta\partial\theta^{\prime}}r_{n}(\theta_{n}^{*}-\theta_{n})+\frac{s_{n}}{r_{n}}r_{n}(\theta_{n}^{*}-\theta_{n})^{\prime}\frac{\partial g(\bar{\theta})}{\partial\theta\partial\theta^{\prime}}r_{n}(\theta_{n}^{*}-\theta_{n})\nonumber \\
= & h_{n}^{\prime}\frac{\partial g(\bar{\bar{\theta}})}{\partial\theta\partial\theta^{\prime}}x_{n}^{*}+\frac{s_{n}}{r_{n}}x_{n}^{*\prime}\frac{\partial g(\bar{\theta})}{\partial\theta\partial\theta^{\prime}}x_{n}^{*}\nonumber \\
= & h^{\prime}Hx_{n}^{*}+o_{p}(1).\label{eq:h'Hx}
\end{align}
where $\bar{\theta}$ is between $\theta_{n}$ and $\theta_{n}^{*}$
and $\bar{\bar{\theta}}$ is between $\theta_{n}$ and $\theta_{\star}$.
Since $h^{\prime}H\neq0$, there is $\iota$ such that $h^\prime H\iota=1$, and $a_{n}=o_{p}(1)$ such that 
\[
h^{\prime}H(x_{n}^{*}+a_{n}\iota)=0.
\]
Thus $x_{n}^{*}+a_{n}\iota$ is feasible for the problem in $T^{(2)}$,
which implies that
\[
T^{(2)}\leq\left\Vert \mathbb{Z}-\Sigma_{n}^{-1/2}(x_{n}^{*}+a_{n}\iota)\right\Vert ^{2}=T_{n}^{(2)}+o_{p}(1).
\]

Step 2, we show that $T_{n}^{(2)}\leq T^{(2)}+o_{p}(1)$. By the same
algebra in (\ref{eq:h'Hx}), for $b=O_{p}(1)$, 
\begin{align*}
s_{n}r_{n}\left(g(\theta_{n}+r_{n}^{-1}(x^{*}+b\iota))-g(\theta_{n})\right)=h^{\prime}H(x^{*}+b\iota)+o_{p}(1)=b+o_{p}(1),
\end{align*}
where the last equality follows from $h^{\prime}Hx^{*}=0$. Define
$b_{n}$ as the solution for 
\[
b_{n}=\underset{b}{\arg\min|b|}\text{ s.t. }g(\theta_{n}+r_{n}^{-1}(x^{*}+b\iota))-g(\theta_{n})=0,
\]
and $b_{n}=\infty$ if there is no solution for $g(\theta_{n}+r_{n}^{-1}(x^{*}+b\iota))-g(\theta_{n})=0$.
Next, we show that $b_{n}=o_{p}(1)$, i.e. for all $\epsilon>0$,
there is $N$ such that for all $n>N$, 
$P(|b_{n}|>\epsilon)<\epsilon.$
Let 
\[
u_{n}(b)=s_{n}r_{n}\left(g(\theta_{n}+r_{n}^{-1}(x^{*}+b\iota))-g(\theta_{n})\right).
\]
Then 
\[
u_{n}(\epsilon)=\epsilon+o_{p}(1),\;u_{n}(-\epsilon)=-\epsilon+o_{p}(1).
\]
Therefore, there is $N$ such that for all $n>N$, 
\[
P\left(u_{n}(\epsilon)u_{n}(-\epsilon)<0\right)\geq1-\epsilon.
\]
This implies that
\[
T_{n}^{(2)}\leq\left\Vert \mathbb{Z}_{n}-\Sigma_{n}^{-1/2}(x^{*}+b_n\iota)\right\Vert ^{2}=\left\Vert \mathbb{Z}-\Sigma^{-1/2}x^{*}\right\Vert ^{2}+o_{p}(1)=T^{(2)}+o_{p}(1).
\]
Step 1 and 2 complete the proof.
\end{proof}
\begin{lem}
\label{lem:T*approximation}Suppose Assumption \ref{assu:uniform=000020normal}.\ref{enu:compact=000020Sigma}
and \ref{assu:uniform=000020normal}\ref{enu:consistentSigma} hold.
Let $\hat{h}_{n}$ and $s_{n}$ be sequences such that $s_{n}\rightarrow\infty$,
$||\hat{h}_{n}||/s_{n}\xrightarrow{p}b\neq0$, and $\frac{\hat{h}_{n}^{\prime}}{||\hat{h}_{n}||}H\xrightarrow{p}h^{\prime}$.
For $\hat{T}_{n}^{*}(h)$ defined in (\ref{eq:T*}) with $\mathbb{Z}\sim N(0,I_{d})$,
it holds that 
\[
\hat{T}_{n}^{*}(\hat{h}_{n})=\inf_{h^{\prime}x=0}\left\Vert \mathbb{Z}-\Sigma^{-1/2}x\right\Vert ^{2}+o_{p}(1).
\]
\end{lem}

\begin{proof}
Let $x=t-h_{n}$, we can rewrite $\hat{T}_{n}^{*}(\hat{h}_{n})$ as
\begin{equation}
\hat{T}_{n}^{*}(\hat{h}_{n})=\inf_{x:\frac{\hat{h}_{n}^{\prime}H}{||\hat{h}_{n}||}x=-\frac{x^{\prime}Hx}{2||\hat{h}_{n}||}}\left\Vert \mathbb{Z}-(\hat{\Sigma})^{-1/2}x\right\Vert ^{2}.\label{eq:T*-1}
\end{equation}
Note that the optimizer $x_{n}^{*}=O_{p}(1)$, thus 
\[
h^{\prime}x_{n}^{*}=\left(\frac{\hat{h}_{n}^{\prime}H}{||\hat{h}_{n}||}+o(1)\right)x_{n}^{*}=-\frac{x_{n}^{*\prime}Hx_{n}^{*}}{2||\hat{h}_{n}||}+o_{p}(1)=o_{p}(1)
\]
where the last equality follows from $\hat{h}_{n}/s_{n}\xrightarrow{p}b\neq0$.
The remainder follows by the same continuity and perturbation arguments
as in Lemma \ref{lem:Tn=000020=00003D=000020T+o1=0000202}.
\end{proof}
\begin{lem}
\label{lem:lambda=000020x=000020epsilin}Let $\varepsilon_{n}=o_{p}(1)$,
$\mathbb{X}\sim N(h,I_{d})$ with $h\in\mathbb{R}^{d}$, and let $x^{*}$
be the solution to 
\[
\inf_{x^{\prime}Hx=0}T(x),\text{ where }T(x)=\left\Vert \mathbb{X}-x\right\Vert ^{2},
\]
where $H=\text{diag}\{\lambda_{1},...,\lambda_{m},-\lambda_{m+1},...,-\lambda_{d}\}$
with $\lambda_{i}>0$. Then for all $\epsilon>0$, there is $N$ such
that for all $n>N$, 
\[
P\left(\epsilon\sum_{i=1}^{d}\lambda_{i}x_{i}^{*2}+\varepsilon_{n}>0\right)\geq1-\epsilon.
\]
\end{lem}

\begin{proof}
Step 1. Let $\ell_{1}=\frac{\sqrt{\lambda_{1}}}{\sqrt{\lambda_{1}+\lambda_{d}}},\;\ell_2=\sqrt{1-\ell_1^2}$.
Then
\[
\inf_{x^{\prime}Hx=0}T(x)\leq\inf_{\ell_{1}x_{1}-\ell_{2}x_{d}=0,\;x_{2:d-1}=0}T(x)=T(0)-\left(\ell_2\mathbb{X}_{1}+\ell_{1}\mathbb{X}_{d}\right)^{2}.
\]
By the continuity
of $T(x)$, for all $\epsilon>0$, there is $C_{\epsilon}>0$ such that 
\begin{align*}
P\left(\sum_{i=1}^{d}\lambda_{i}x_{i}^{*2}<C_{\epsilon}\right) & \leq P\left(T(0)-\inf_{x^{\prime}Hx=0}T(x)<\left(\Phi^{-1}(\frac{1}{2}+\frac{\epsilon}{4})\right)^{2}\right)\\
 & \leq P\left(\left(\ell_2\mathbb{X}_{1}+\ell_{1}\mathbb{X}_{d}\right)^{2}<\left(\Phi^{-1}(\frac{1}{2}+\frac{\epsilon}{4})\right)^{2}\right).
\end{align*}
Since $\ell_2\mathbb{X}_{1}+\ell_{1}\mathbb{X}_{d}\sim N\left(\ell_2h_{1}+\ell_{1}h_{d},1\right)$,
the squared term follows a noncentral $\chi_{1}^{2}$, so the probability
is bounded by $\epsilon/2$.

Step 2. Since $\varepsilon_{n}=o_{p}(1)$, there is $N$ such that
for all $n\geq N$, 
\[
P\left(\varepsilon_{n}\geq-\frac{C_{\epsilon}\epsilon}{2}\right)\geq1-\frac{\epsilon}{2}.
\]
Combining with Step 1,
\begin{align*}
P\left(\epsilon\sum_{i=1}^{d}\lambda_{i}x_{i}^{*2}+\varepsilon_{n}>0\right) & \geq P\left(\sum_{i=1}^{d}\lambda_{i}x_{i}^{*2}\geq C_{\epsilon},\varepsilon_{n}>-\frac{C_{\epsilon}\epsilon}{2}\right)\\
 & \geq P\left(\sum_{i=1}^{d}\lambda_{i}x_{i}^{*2}\geq C_{\epsilon}\right)+P\left(\varepsilon_{n}>-\frac{C_{\epsilon}\epsilon}{2}\right)-1\\
 & \geq1-\frac{\epsilon}{2}+1-\frac{\epsilon}{2}-1=1-\epsilon.
\end{align*}
\end{proof}
\begin{lem}
\label{lem:continus=000020T}Let $D(Y)=\inf_{x:x^{\prime}Hx=c}\left\Vert Y-x\right\Vert $
where $Y$ is continuously distributed and $c\geq0$. Then $D(Y)$
is continuously distributed. 
\end{lem}

\begin{proof}
It suffices to show that for all $a\geq0$, $P\left(D(Y)=a\right)=0$.
Let $S=\left\{ x:x^{\prime}Hx=c\right\} $.

Case 1. $a=0$. Here $P(D(Y)=0)=P\left(Y\in S\right)$. Since $H$
is indefinite, $S$ is a variety of dimension at most $d-1$, and
hence $S$ has Lebesgue measure zero. Because $Y$ has a continuous
distribution, so $P(Y\in S)=0$.

Case 2. $a>0$. Suppose, by contradiction, that $P\left(D(Y)=a\right)>0$.
Then the set $S_{a}=\left\{ y:D(y)=a\right\} $ must have positive
Lebesgue measure. Since $D(y)$ is continuous in $y$, $S_{a}$ is
a closed set. Hence, there exists a ball $B(o,r)$ with $a>r>0$ such
that $B(o,r)\subseteq S_{a}$. Let $o^{\prime}$ be the projection
of $o$ onto $S$, i.e. $o^{\prime}\in S$ and $||o-o^{\prime}||=a$.
Choose a point $k\in\partial B(o,r)\cap\overline{oo^{\prime}}$. Such
$k$ exists since $o$ lies inside $\partial B(o,r)$ while $o^{\prime}$
lies outside. By construction,
\[
||k-o^{\prime}||=||o-o^{\prime}||-||o-k||=a-r<a
\]
which contradicts the fact that $k\in B(o,r)\subseteq S_{a}$. Therefore,
no such $a>0$ can exist.
\end{proof}
\begin{lem}
\label{lem:con=000020quantile}Let $\mathbb{Z}\sim N(0,I_{d})$, and
 $\mathcal{H}_{z}$ be a set satisfying $P\left(\mathbb{Z}\in\mathcal{H}_{z}\right)\geq1-\eta$.
Then 
\[
Q\left(\left.\inf_{h^{\prime}x=0}\left\Vert \mathbb{Z}-x\right\Vert ^{2}\right|\mathbb{Z}\in\mathcal{H}_{z};\frac{1-\alpha}{1-\eta}\right)\in\left[Q(\chi_{1}^{2},1-\alpha),\;Q(\chi_{1}^{2},1-\alpha+\eta)\right).
\]
\end{lem}

\begin{proof}
To show the upper bound,
\begin{align*}
 & P\left(\left.\inf_{h^{\prime}x=0}\left\Vert \mathbb{Z}-x\right\Vert ^{2}\leq Q(\chi_{1}^{2},1-\alpha+\eta)\right|\mathbb{Z}\in\mathcal{H}_{z}\right)\\
= & \frac{P\left(\inf_{h^{\prime}x=0}\left\Vert \mathbb{Z}-x\right\Vert ^{2}\leq Q(\chi_{1}^{2},1-\alpha+\eta),\mathbb{Z}\in\mathcal{H}_{z}\right)}{P\left(\mathbb{Z}\in\mathcal{H}_{z}\right)}\\
> & \frac{P\left(\inf_{h^{\prime}x=0}\left\Vert \mathbb{Z}-x\right\Vert ^{2}\leq Q(\chi_{1}^{2},1-\alpha+\eta)\right)+P\left(\mathbb{Z}\in\mathcal{H}_{z}\right)-1}{P\left(\mathbb{Z}\in\mathcal{H}_{z}\right)}\\
= & \frac{1-\alpha+\eta+1-\eta-1}{1-\eta}=\frac{1-\alpha}{1-\eta}.
\end{align*}
To show the lower bound, 
\begin{align*}
 & P\left(\left.\inf_{h^{\prime}x=0}\left\Vert \mathbb{Z}-x\right\Vert ^{2}\leq Q(\chi_{1}^{2},1-\alpha)\right|\mathbb{Z}\in\mathcal{H}_{z}\right)\\
= & \frac{P\left(\inf_{h^{\prime}x=0}\left\Vert \mathbb{Z}-x\right\Vert ^{2}\leq Q(\chi_{1}^{2},1-\alpha),\mathbb{Z}\in\mathcal{H}_{z}\right)}{P\left(\mathbb{Z}\in\mathcal{H}_{z}\right)}\\
\leq & \frac{P\left(\inf_{h^{\prime}x=0}\left\Vert \mathbb{Z}-x\right\Vert ^{2}\leq Q(\chi_{1}^{2},1-\alpha)\right)}{P\left(\mathbb{Z}\in\mathcal{H}_{z}\right)}=\frac{1-\alpha}{1-\eta}.
\end{align*} 
\end{proof}
\begin{lem}
\label{lem:y=000020H=000020x}Let $H$ be a $d\times d$ full rank indefinite matrix.
For all $d\times 1$ vector $x$, there exists $d\times1$ vector $y$ such that $y^{\prime}Hx=0$ and $y^{\prime}Hy=-\text{sign}(x^{\prime}Hx)$.
\end{lem}

\begin{proof}
If $x^{\prime}Hx=0$, the conclusion holds trivially with $y=0_{d}$.
WLOG, assume that $x^{\prime}Hx>0$. First assume that $H$ is diagonal.
If not, write $H=P^{\prime}\Lambda P$, where $\Lambda$ is diagonal,
$\tilde{x}=Px$ and $\tilde{y}=Py$; the same argument then applies.
Without loss of generality, let 
\[
H=\text{diag}(\lambda_{1},...,\lambda_{m},-\lambda_{m+1},...,-\lambda_{d}),\quad\lambda_{1},...,\lambda_{d}>0.
\]
Case 1. If $x_{m+1},...,x_{d}=0$, let $y=(0_{d-1},1)$. It is easy to verify that $y^\prime Hx=0$ and $y^\prime Hy=-\lambda_d<0.$ Case 2. Suppose
$x_{d}\neq0$, let $y=(x_{1},...,x_{m},0_{d-m-1},y_{d})$, $y_{d}=\frac{\sum_{i=1}^{m}\lambda_{i}x_{i}^{2}}{\lambda_{d}x_{d}}.$
Then
\[
y^{\prime}Hx=\sum_{i=1}^{m}\lambda_{i}x_{i}^{2}-\lambda_{d}x_{d}y_{d}=\sum_{i=1}^{m}\lambda_{i}x_{i}^{2}-\lambda_{d}x_{d}\frac{\sum_{i=1}^{m}\lambda_{i}x_{i}^{2}}{\lambda_{d}x_{d}}=0,
\]
\begin{align*}
y^{\prime}Hy & =\sum_{i=1}^{m}\lambda_{i}x_{i}^{2}-\lambda_{d}y_{d}^{2}=-\frac{\sum_{i=1}^{m}\lambda_{i}x_{i}^{2}}{\lambda_{d}x_{d}^{2}}\left(\sum_{i=1}^{m}\lambda_{i}x_{i}^{2}-\lambda_{d}x_{d}^{2}\right)\\
 & \leq-\frac{\sum_{i=1}^{m}\lambda_{i}x_{i}^{2}}{\lambda_{d}x_{d}^{2}}x^{\prime}Hx<0.
\end{align*}
Then conclusion holds with a simple normalization of $y$. 
\end{proof}

\end{document}